\documentclass[sn-mathphys-num]{sn-jnl}

\usepackage{amsmath,amssymb,amsfonts}%
\usepackage{amsthm}%
\usepackage{mathrsfs}%
\usepackage[title]{appendix}%
\usepackage{xcolor}%
\usepackage{textcomp}%
\usepackage{manyfoot}%
\usepackage{booktabs}%
\usepackage{algorithm}%
\usepackage{algorithmicx}%
\usepackage{algpseudocode}%
\usepackage{listings}%
\usepackage{graphicx}%
\usepackage{multirow}%

\theoremstyle{thmstyleone}%
\newtheorem{theorem}{Theorem}
\newtheorem{proposition}[theorem]{Proposition}%
\newtheorem{lemma}{Lemma}

\theoremstyle{thmstyletwo}%
\newtheorem{remark}{Remark}%
\newtheorem{assumption}{Assumption}
\newtheorem*{assumption*}{Assumption}
\newtheorem*{condition*}{Condition}

\theoremstyle{thmstylethree}%
\newtheorem{claim}{Claim}

\usepackage{pifont}
\newcommand{\argmax}{\arg\!\max}

\global\long\def\expect{\mathbb{E}}%
\global\long\def\prob{\mathbb{P}}%
\global\long\def\real{\mathbb{R}}%
\global\long\def\op{o_{P}}%
\global\long\def\transpose{\top}%
\global\long\def\define{:=}%
\global\long\def\var{\mathrm{Var}}%
\global\long\def\cov{\mathrm{Cov}}%

\def\indicator#1{\mathbb{I}_{\left\{#1\right\}}}
\def\supp#1{\text{supp}(#1)}

\def\onevec#1{\textbf{1}_{#1}}
\def\zerovec#1{\textbf{0}_{#1}}
\newcommand{\cmark}{\ding{52}}%
\newcommand{\xmark}{\ding{56}}%

\def\uconst{h}

\def\yinan#1{\textcolor{black}{#1}}

\begin{document}

\title[Article Title]{Testing High-Dimensional Mediation Effect with Arbitrary Exposure--Mediator Coefficients}


\author[1]{\fnm{Yinan} \sur{Lin}}\email{linyn@cqnu.edu.cn}

\author[2]{\fnm{Zijian} \sur{Guo}}\email{zijguo@stat.rutgers.edu}

\author[3]{\fnm{Baoluo} \sur{Sun}}\email{stasb@nus.edu.sg}

\author*[3]{\fnm{Zhenhua} \sur{Lin}}\email{linz@nus.edu.sg}

\affil[1]{\orgdiv{National Center for Applied Mathematics in Chongqing}, \orgname{Chongqing Normal University}, \orgaddress{\street{37 Daxuecheng Middle Road}, \city{Chongqing}, \postcode{401331}, \state{Chongqing}, \country{China}}}

\affil[2]{\orgdiv{Department of Statistics}, \orgname{Rutgers University}, \orgaddress{\street{110 Frelinghusen Road}, \city{Piscataway}, \postcode{08854}, \state{New Jersey}, \country{USA}}}

\affil*[3]{\orgdiv{Department of Statistics and Data Science}, \orgname{National University of Singapore}, \orgaddress{\street{6 Science Drive 2}, \city{Singapore}, \postcode{117546}, \state{Singapore}, \country{Singapore}}}


\abstract{In response to the unique challenge created by high-dimensional mediators in mediation analysis, this paper presents a novel procedure for testing the nullity of the mediation effect in the presence of high-dimensional mediators. The procedure incorporates two distinct features. Firstly, the test remains valid under all cases of the composite null hypothesis, including the challenging scenario where both exposure--mediator and mediator--outcome coefficients are zero. Secondly, it does not impose structural assumptions on the exposure--mediator coefficients, thereby allowing for  an arbitrarily strong exposure--mediator relationship. To the best of our knowledge, the proposed test is the first of its kind to provably possess these two features in high-dimensional mediation analysis. The validity and consistency of the proposed test are established, and its numerical performance is showcased through simulation studies. The application of the proposed test is demonstrated by examining the mediation effect of DNA methylation between smoking status and lung cancer development.}




\maketitle

\section{Introduction}\label{sec:intro}

Originally motivated in psychology \citep{baron1986moderator}, mediation analysis has found widespread application in various scientific fields, including medicine, genomics, economics, and many others \citep{zeng2021statistical, celli2022causal}, over the past few decades. 
Given independent and identically distributed (i.i.d.) triples $(Y_{i}, A_{i}, M_{i})_{i=1}^{n}$, mediation analysis focuses on examining the effect of exposures $A_i\in \real^q$ on an outcome $Y_i\in \real$, which may be mediated by some potential intermediate variables $M_i\in \real^{p}$ known as mediators. 
In the case where the relationships among $Y_{i}$, $A_{i}$, and $M_{i}$ are linear, one can consider the following linear structural equation model (LSEM):
\begin{align}
	M_i &= \beta_{A} A_i + E_{i}, 
	\label{eq:mediator-exposure}
	\\
	Y_i &= \theta_{A}^{\transpose} A_i + \theta_{M}^{\transpose} M_i + Z_{i},
	\label{eq:outcome-mediator}
\end{align}
for each $i=1,\ldots, n$. 
Here, both $E_i\in\real^p$ and $Z_i\in\real$ are random noise. The matrix $\beta_{A} \in \real^{p\times q}$ and the vectors    
$\theta_{A}\in \real^{q}$ and  $\theta_{M}\in \real^{p}$ contain the unknown regression coefficients that  encode the relationships among the exposures, mediators and outcome. 
In this paper, we consider the high-dimensional setting where the dimension $p$ of the mediators may diverge with the sample size $n$. 

In the LSEM \eqref{eq:mediator-exposure} and \eqref{eq:outcome-mediator}, the joint mediated effect through the mediators, also referred to as the natural indirect effect or mediation effect, is captured by the parameter $\gamma = \beta_{A}^{\top} \theta_{M} \in \mathbb{R}^{q}$ \citep{robins1992identifiability,vanderweele2014mediation}. 
It is known that, if the LSEM \eqref{eq:mediator-exposure} and \eqref{eq:outcome-mediator} is correctly specified with the absence of measured baseline covariates, $\gamma$ admits a causal interpretation under a counterfactual framework; see Appendix \ref{sec:causal-interpret} for details. 	
Our focus is on the hypothesis testing problem,
\begin{equation}
	\mathrm{H}_{0}:\gamma=\zerovec{q}
	\quad\text{v.s.}\quad
	\mathrm{H}_{a}:\gamma\neq \zerovec{q},
	\label{test-problem-product}
\end{equation}
where $\zerovec{q} \in \real^{q}$ is the vector of all zeros. 
Due to its practical importance, numerous statistical methods have been proposed for the  problem \eqref{test-problem-product}, primarily in the context of low-dimensional mediators \citep{ vanderweele2009conceptual, vanderweele2014mediation}. 

Nowadays, high-dimensional data are ubiquitous in many areas, such as bioinformatics \citep{zeng2021statistical}.
This has led to a growing need for new statistical methods for mediation analysis with high-dimensional mediators, where the number of potential mediators may be comparable to, or even larger than, the sample size. 
For instance, genome-wide association studies have investigated the impact of early-life trauma on cortisol stress reactivity in adulthood through hundreds of thousands of DNA methylation levels \citep{houtepen2016genome, guo2022high}.
Epidemiological studies have also confirmed the role of socioeconomic factors, mediated through molecular-level traits including methylation, in disease susceptibility \citep{tobi2018dna,huang2019genome}.
In neuroscience, there is interest in identifying brain regions, comprising a large volume of voxels, whose activity levels act as potential mediators in the relationship between a thermal stimulus and self-reported pain \citep{chen2018high}. At first glance, in  mediation analysis, one might attempt to address $\beta_{A}$ and $\theta_{M}$ separately through their respective equations in the LSEM. However, the coexistence of $\beta_{A}\ne 0$ and $\theta_{M}\ne 0$ does not imply $\gamma \ne 0$. In addition, separately analyzing the two equations do not account for the complex interplay between the exposure, mediator and outcome. Consequently, existing techniques for high-dimensional linear models 
cannot be readily adapted to analyze the two coupled equations in the LSEM.

Various methodologies have been proposed in response to the unique challenge of mediation analysis, especially in the presence of high-dimensional mediators. For example, one might first reduce the number of mediators through dimension reduction techniques, such as principal component analysis \citep{huang2016hypothesis,chen2018high} and variable screening  \citep{zhang2016estimating}. Alternatively,  studies such as \cite{huang2019genome, dai2022multiple, liu2022large} consider testing whether $\beta_{A,j}\theta_{M,j}=0$ for all $j=1,\ldots,p$, where $\beta_{A,j}$ and $\theta_{M,j}$ are respectively the $j$th coordinates of $\beta_{A}$ and $\theta_M$. However, while these individual null hypotheses are meaningful in their own right, they do not collectively equate to $\gamma=0$. In addition, \cite{huang2019genome} requires a sparse $\beta_A$, whereas \cite{dai2022multiple} and \cite{liu2022large} impose implicit conditions on $\beta_A$ through an assumption of the weak dependence across coordinates $1,\ldots,p$. In contrast, 
The recent works by \cite{zhou2020estimation} and \cite{guo2022statistical} directly address the overall mediation effect $\gamma$.  Specifically, 
\cite{zhou2020estimation} proposed a debiased estimator and a test for the mediation effect, under some structural assumptions on the exposure--mediator coefficients $\beta_{A}$ to ensure consistency of the estimator. 
\cite{guo2022statistical}, observing that the mediation effect $\gamma$ is the difference between the total effect $\gamma+\theta_{A}$ and the natural direct effect $\theta_{A}$, proposed to estimate $\gamma$ by the difference between an estimator of the total effect and an estimator of the natural direct effect, and developed a Wald test for the mediation effect. Their method relies on the sign consistency of the estimated mediator--outcome coefficient $\theta_{M}$, which demands relatively strong assumptions, such as the uniform signal strength condition \citep{zhang2014confidence} and the irrepresentable condition \citep{zhao2006model}. As a consequence, the exposure--mediator coefficient $\beta_{A}$ in \cite{guo2022statistical} cannot be too large in magnitude.

The null hypothesis $\gamma=0$ includes the case that both $\beta_A=0$ and $\theta_M = 0$, which is not that uncommon in  applications such as genome-wide studies due to {extreme sparse signals} \citep{huang2019genome, barfield2017testing}. 
In addition to structural assumptions on $\beta_{A}$, both tests of \cite{zhou2020estimation} and  \cite{guo2022statistical} do not address this peculiar case of the null hypothesis. The fundamental cause is that, when $\beta_A=0$ and $\theta_M = 0$, the standard deviations of their estimators for the mediation effect decay to zero at a rate faster than $n^{-1/2}$, rendering their asymptotic normality results invalid; see a numerical demonstration in Section \ref{sec:simu}.  In fact, the case of $\beta_A=0$ and $\theta_M=0$ is nontrivial even in the low-dimensional setting \citep{sobel1982asymptotic, barfield2017testing, zhou2020estimation}.

As a major contribution of this paper, we develop a test for the overall mediation effect $\gamma=0$ in the presence of high-dimensional mediators, with the following distinctive features. First, it does not impose assumptions on $\beta_A$, and thus is able to accommodate \textit{arbitrary} exposure--mediator coefficients $\beta_{A}$, including a dense vector $\beta_A$. Second, it remains valid in all cases of the null hypothesis, even in the challenging case of  $\beta_{A}=0$ and $\theta_{M}=0$.   See Table \ref{tab:method-cases} for a contrast between our method and the others. To the best of our knowledge, our test is the first  of its kind to enjoy both of these features in high-dimensional mediation analysis. 

Our test is built on a novel debiased estimator for the mediation effect $\gamma$, achieved by adapting the technique of the variance-enhancement projection direction \yinan{(VePD)} \citep{cai2021optimal}. Although the method of \cite{zhou2020estimation} is also based on debiasing a pilot estimator by projecting the sum of residuals, the projection direction in \cite{zhou2020estimation} is constructed to solely alleviate the bias of the pilot estimator. In contrast, in our procedure, in addition to correcting the bias, we construct the direction to also control the variance of {the estimator}. Therefore, the variance of our estimator would dominate the corresponding bias for any $\beta_{A}$, thus allowing for arbitrary exposure--mediation associations; see the discussion right after \eqref{eq:hgamma-estimator} for details. Compared to the work presented in \cite{cai2021optimal}, which primarily concentrates on testing a linear contrast with a predetermined high-dimensional loading vector, the development of our test method confronts a distinctive theoretical hurdle. First, this challenge stems from the consideration of a \emph{random} high-dimensional loading vector, which emerges due to the estimation of the unknown coefficient vector $\beta_{A}$. Second, the intricate interdependence between the two high-dimensional equations within the LSEM introduces a significant complication in our theoretical analysis. This complexity is  distinct from the typical situation encountered in the high-dimensional linear models, where a single equation is usually addressed \citep[e.g.,][]{cai2021optimal}.

\begin{table}[h]
	\caption{Validity of our method and the competing methods, and requirement on sign consistency}\label{tab:method-cases}%
	\begin{tabular}{ccccc}
		\hline
		\multicolumn{1}{c}{Method} & Sparse $\beta_{A}$ & Dense $\beta_{A}$ & $\beta_A=0$ and $\theta_M=0$ & \multicolumn{1}{c}{Sign Consistency} \\
		\hline
		\cite{zhou2020estimation}  & \cmark & * & \xmark & Not required \\
		\hline
		\cite{guo2022statistical}   & \cmark & \xmark & \xmark & Required \\
		\hline
		Our proposal & \cmark & \cmark & \cmark & Not required \\
		\hline
	\end{tabular}
	\footnotetext{Note: While there is no direct sparsity condition on $\beta_A$ in \cite{zhou2020estimation}, Assumption 2 therein imposes a structural requirement on the covariance structures of the mediators and exposures. The discussion following that assumption mentions that such requirement is related to the irrepresentable condition of \cite{zhao2006model}. When $\beta_{A}$ is dense, exposures and mediators are likely to be strongly correlated, making the irrepresentable condition hard to meet \citep{zhao2006model}.}
\end{table}

The rest of the paper is organized as follows. In Section \ref{sec:model-meth}, we present the proposed debiased estimator and the corresponding test for mediation effect. 
Theoretical investigations of the proposed test are provided in Section \ref{sec:theory}. 
Numerical studies on the proposed method and comparisons with other methods are presented in Section \ref{sec:simu}. In Section \ref{sec:application}, we showcase the proposed method in a real data application, which investigates the mediation effect between smoking status, DNA methylation, and lung cancer development. The paper concludes with a final remark in Section \ref{sec:concluding}.

\vspace{2mm}
\textbf{Notation.} 
For a vector $x$,  $x^{\transpose}$ is its transpose, $\|x\|_r$ represents its $\ell_r$ norm with $r=1,2,\infty$, and $\text{diag}(x)$ denotes the diagonal matrix whose diagonal is $x$.
For a matrix $X$, $\|X\|_2$ and $\|X\|_{\infty}$ represent its spectral norm and element-wise $\ell_{\infty}$ norm, respectively. Let $X_{i}$ (respectively, $X_{\cdot j}$) be its $i$th row (respectively, $j$th column) and $X_{jk}$ represents its $(j,k)$th element. In particular, for the matrix $\beta_{A}$ in \eqref{eq:mediator-exposure},  $\beta_{A,j}$  is its $j$th column and $\beta_{A,jk}$ is its $(j,k)$th element. If $X$ is a squared matrix, $\Lambda_{\min}(X)$ and $\Lambda_{\max}(X)$ denote its smallest and largest eigenvalues, respectively. 
$\onevec{q}$ and $\zerovec{q}$ are the vectors consisting of all ones and all zeros, respectively, in $\real^{q}$. Occasionally, we may omit the subscript $q$ when it's clear from the context.
$I_n$ represents the identity matrix of size $n$. $\indicator{\cdot}$ denotes the indicator function.
$N_{q}(\mu, \Sigma)$ represents the $q$-dimensional Gaussian distribution with mean $\mu$ and covariance matrix $\Sigma$, while $N(0,1)$ represents the standard Gaussian distribution in $\real$. 
For two non-negative sequences $\{a_n\}$ and $\{b_n\}$, we write $a_n \lesssim b_n$ (respectively, $ a_n\gtrsim b_n$) if there is a constant $c>0$ not depending on $n$, such that $a_n \le c b_n$ (respectively, $a_n\ge c b_n$) for all sufficiently large $n$. We write $a_n \asymp b_n$ if and only if both $a_n \lesssim b_n$ and $a_n \gtrsim b_n$. Moreover, $a_n \ll b_n$ if $a_n/b_n \to 0$ as $n\to \infty$. 	
\vspace{2mm}

\section{Testing Mediation Effect in High Dimensions}\label{sec:model-meth}

\subsection{Estimating  Mediation Effect via a Debiased Approach}\label{sec:estimation}

Given i.i.d. triples $(Y_{i}, A_{i}, M_{i})_{i=1}^{n}$, without loss of generality, we assume centered $A_i$ and $E_i$, that is, $\expect A_i=0$ and $\expect M_i=0$ for each $i=1, \ldots, n$. In practice, this assumption can be satisfied by centering $Y_i$, $A_i$ and $M_i$, that is, by instead considering  $Y_{i}-\bar{Y}$, $A_{i}-\bar{A}$ and $M_{i}-\bar{M}$, where $\bar{Y}=n^{-1}\sum_{i=1}^{n}Y_{i}$, $\bar{A}=n^{-1}\sum_{i=1}^{n}A_{i}$ and $\bar{M}=n^{-1}\sum_{i=1}^{n}M_{i}$. 	
To simplify our discussions, we rewrite  \eqref{eq:mediator-exposure} and \eqref{eq:outcome-mediator} into a matrix form, 
\begin{align}
	M &= A\beta_{A}^{\transpose}  + E, 
	\label{eq:mediator-exposure-matrix}
	\\
	Y &= A \theta_{A} + M\theta_{M}  + Z,
	\label{eq:outcome-mediator-matrix}
\end{align}
where $M\in \real^{n\times p}$, $A\in \real^{n\times q}$, $Y\in \real^{n\times 1}$, $E\in \real^{n\times p}$ and $Z \in \real^{n\times 1}$. Further, we set $X=(A, M)\in \real^{n\times (q+p)}$ and $\theta = (\theta_A^{\transpose}, \theta_M^{\transpose})^{\transpose} \in \real^{q+p}$, and  \eqref{eq:outcome-mediator-matrix} becomes
\begin{equation}
	Y = X \theta + Z.
	\label{eq:Y-X-equation}
\end{equation}

Let $\hat{\theta}=(\hat{\theta}_A^{\transpose}, \hat{\theta}_M^{\transpose})^{\transpose}$ be an initial estimator for the high-dimensional vector $\theta$, such as the Lasso estimator
\[
\hat{\theta} = \underset{b\in \real^{q+p}}{\arg\min} \,\, n^{-1} \|Y-Xb\|_2^2 + \lambda_n \|b\|_1
\]
with $\lambda_n$ being the Lasso tuning parameter. Given the ordinary least-squares estimator $\hat{\beta}_A= ((A^{\transpose}A)^{-1}A^{\transpose} M)^{\transpose}$ for the coefficient matrix $\beta_A$, a pilot estimator for $\gamma=\beta_{A}^{\transpose}\theta_{M}$ is then given by $\tilde{\gamma}=\hat{\beta}_A^{\transpose}\hat{\theta}_M$. From \eqref{eq:mediator-exposure-matrix} and \eqref{eq:outcome-mediator-matrix},  we observe
\begin{equation}
	\tilde{\gamma} - \gamma 
	= 
	\hat{\Sigma}_A^{-1} \hat{\Sigma}_{AM} (\hat{\theta}_M - \theta_M) + n^{-1}\hat{\Sigma}_{A}^{-1}A^{\transpose}E_{M}, 
	\label{eq:bias-of-init}
\end{equation}
where $\hat{\Sigma}_A = n^{-1} A^{\transpose}A$, $\hat{\Sigma}_{AM}=n^{-1} A^{\transpose} M$, 
\begin{equation*}
	E_{M}=(E_{M,1}, \ldots, E_{M,n})^{\transpose}=E \theta_{M} \in \real^n,
	\label{eq:def-eps1}
\end{equation*}
and $E_{M,i}=E_i^{\transpose} \theta_{M}$ for $i=1,\ldots, n$.
\yinan{Intuitively, the second term on the right-hand side of \eqref{eq:bias-of-init} is asymptotically normal under regularity conditions. In contrast,} 
with high-dimensional mediators, the first term $\hat{\Sigma}_A^{-1} \hat{\Sigma}_{AM} (\hat{\theta}_M - \theta_M)$ in \eqref{eq:bias-of-init} becomes a non-negligible bias of $\tilde\gamma$ due to the high-dimensional penalty. To eliminate this bias, we employ a bias correction approach, as follows. 

Let $\tilde g_j$ be the $j$th column of $(\hat{\Sigma}_A^{-1} \hat{\Sigma}_{AM})^{\transpose}\in \real^{p\times q}$ and thus $\hat{\Sigma}_A^{-1} \hat{\Sigma}_{AM} (\hat{\theta}_M - \theta_M)=\big(\tilde{g}_1^{\transpose} (\hat{\theta}_M - \theta_M), \ldots, \tilde{g}_q^{\transpose} (\hat{\theta}_M - \theta_M) \big)^{\transpose}$. To estimate each coordinate of the bias, we first observe that, for any projection direction $u \in \real^{q+p}$, 
\begin{equation}
	\tilde{g}_j^{\transpose} (\hat{\theta}_M- \theta_M ) - n^{-1} u^{\transpose}  X^{\transpose} (X\hat{\theta}-Y) =  (\hat{\Sigma}_{X} u - g_j)^{\transpose}(\theta - \hat{\theta})+n^{-1} u^{\transpose} X^{\transpose} Z,
	\label{eq:gj-def}
\end{equation}
where $g_j = (\zerovec{q}^{\transpose}, \tilde{g}_j^{\transpose}) \in \real^{q+p}$ and $\hat{\Sigma}_{X}=n^{-1} X^{\transpose} X$. This motivates us to use $n^{-1} u^{\transpose}  X^{\transpose} (X\hat{\theta}-Y)$ as an estimate of the $j$th coordinate $\tilde{g}_j^{\transpose} ( \hat{\theta}_M-\theta_M )$ of the bias, and to choose a direction $u$ that ``minimizes'' the estimation error, that is, the right-hand side of \eqref{eq:gj-def}, for \emph{all} realizations of $g_j$. For this purpose, we first observe the following inequality regarding the first term of the right-hand side of \eqref{eq:gj-def},
\[
|(\hat{\Sigma}_X u - g_j )^{\transpose}(\theta-\hat{\theta})| \le \| \hat{\Sigma}_X u - g_j \|_{\infty} \|\theta-\hat{\theta}\|_{1},
\]
where the upper bound may be minimized by choosing an appropriate initial estimator $\hat\theta$ and imposing constraints on $\| \hat{\Sigma}_X u - g_j \|_{\infty}$ as in \eqref{VePD-constraint-bias}. For the second term  $n^{-1} u^{\transpose} X^{\transpose} Z$ of \eqref{eq:gj-def}, 
as its mean is zero, we find an effective $u$ to minimize its variance which is in the order of $u^{\transpose} \hat{\Sigma}_X u$. In addition, to facilitate statistical inference, an effective projection direction $u$ should ensure that the variance of the associated debiased estimator asymptotically dominates its bias, for all realizations of $g_j$. All of these considerations inspire us to  adapt the VePD technique \citep{cai2021optimal}. 

Specifically, for each $j=1,\ldots,q$, the projection direction $\hat u_j$ is 
\begin{align}
	\hat{u}_j=\underset{u\in \real^{q+p}}{\arg\min}\,\, u^{\transpose} \hat{\Sigma}_X u \quad \text{subject to } \quad  &\| \hat{\Sigma}_X u - g_j \|_{\infty} \le \|g_j\|_2 \lambda, 
	\label{VePD-constraint-bias}
	\\
	&|g_j^{\transpose}\hat{\Sigma}_X u-\|g_j\|_2^2| \le \|g_j\|_2^2 \lambda, 
	\label{VePD-constraint-var}
	\\
	&\|X u\|_{\infty} \le \|g_j\|_2\mu,
	\label{VePD-constraint-lindeberg}
\end{align}
where $\lambda \asymp \sqrt{\log(q+p) /n}$ and {$\mu \asymp \log n$} are two tuning parameters.
Details on solving this constrained optimization and selecting the tuning parameters can be found in Section 4.1 of \cite{guo2021inference}, \yinan{where the constrained minimization problem is equivalently transformed to an unconstrained minimization problem with a readily available optimizer.}
Consequently, with $\hat{U}=(\hat{u}_1, \ldots, \hat{u}_q)\in \real^{(p+q)\times q}$, the estimated bias term in \eqref{eq:bias-of-init} is $n^{-1} \hat{U}^{\transpose} X^{\transpose} (X\hat{\theta}-Y)$, based on which we propose the following debiased estimator for $\gamma$, 
\begin{equation}
	\hat{\gamma} = \tilde{\gamma} - n^{-1} \hat{U}^{\transpose} X^{\transpose} (X\hat{\theta}-Y)=\tilde{\gamma} + n^{-1} \hat{U}^{\transpose} X^{\transpose} (Y-X\hat{\theta})
	\in \real^{q}.
	\label{eq:hgamma-estimator}
\end{equation}

In the above, the constraint \eqref{VePD-constraint-bias} is introduced to tackle the first term on the right-hand side of \eqref{eq:gj-def}; this term may be viewed as the further bias associated with the estimator  $n^{-1} u^{\transpose}  X^{\transpose} (X\hat{\theta}-Y)$ for the $j$th coordinate  $\tilde{g}_j^{\transpose} ( \hat{\theta}_M-\theta_M )$ of the bias in \eqref{eq:bias-of-init}.  This idea of minimizing the variance of a projected sum of errors while constraining the bias is a common practice  in high-dimensional statistical inference with bias correction \citep{zhang2014confidence, javanmard2014confidence}. \cite{cai2021optimal} further introduced an additional constraint \eqref{VePD-constraint-var} to ensure that the variance of the {second term} on the right-hand side of \eqref{eq:gj-def} asymptotically dominates the first term for a \textit{fixed} counterpart of $g_j$, and {hence the variance of the debiased estimator would asymptotically dominate its bias}. In our scenario, each $g_j$ is \textit{random}  due to estimating the parameter $\beta_A$ of the additional equation \eqref{eq:mediator-exposure-matrix}. This distinction, combined with the presence of two coupled equations in our model, sets our theoretical analysis apart from previous works. {The constraint \eqref{VePD-constraint-lindeberg} is primarily a technical requirement to ensure that the error terms in \eqref{eq:gj-def} satisfy the Lindeberg's condition.} 

\subsection{Test Procedure}\label{sec:meth-tesing}
To develop a test for the hypothesis \eqref{test-problem-product}, we start with considering the asymptotic covariance matrix of $\hat{\gamma}$. Based on the combination of \eqref{eq:bias-of-init}, \eqref{eq:gj-def} and \eqref{eq:hgamma-estimator},  conditional on $\{X_i, i=1,\ldots, n\}$, the variance of $\hat\gamma$ is 
\begin{equation}
	V_0 = \frac{\sigma_E^2}{n} \hat{\Sigma}_A^{-1} + \frac{\sigma_Z^2}{n} \hat{U}^{\transpose} \hat{\Sigma}_X \hat{U},
	\label{eq:V0-def}
\end{equation}
where $\sigma_E^2 = \var[E_{M,i}]=\var[E_i^{\transpose} \theta_{M}]$ and $\sigma_Z^2 = \var[Z_{i}]$.
In practice, an estimator $\hat\sigma_Z^2$ of $\sigma_Z^2$ could be derived from the scaled Lasso method. Let $\sigma^2$ be the variance of the residual when regressing $Y_i$ on $A_i$. In light of  $\sigma_E^2 = \sigma^2 - \sigma_Z^2$, we can estimate  $\sigma_E^2$ by $\hat\sigma_E^2=\max\{\hat\sigma^2-\hat\sigma_Z^2, 0\}$, where $\hat\sigma^2$ is the sample version of $\sigma^2$. 
This allows us to obtain an estimated variance $\hat{V}_0=\frac{\hat{\sigma}_E^2}{n} \hat{\Sigma}_A^{-1} + \frac{\hat{\sigma}_Z^2}{n} \hat{U}^{\transpose} \hat{\Sigma}_X \hat{U}$.

When the null of \eqref{test-problem-product} holds, the scenario where $\theta_{M}=0$ and $\beta_{A}=0$ is of particular interest as it is commonly encountered in genome-wide analyses \citep{huang2019genome}, which may lead to a super-efficiency issue. 
To illustrate this, consider the case of $q=1$ for simplicity. 
Since $\sigma_E^2=0$ when $\theta_{M}=0$, the standard deviation $\sqrt{V_0}$ is reduced to $\sqrt{V_0}=\frac{\sigma_Z}{\sqrt{n}}(\hat{u}_{1}^{\transpose}\hat{\Sigma}_{X}\hat{u}_{1})^{1/2}$. The magnitude of this value is in the order of $n^{-1/2}\|g_1\|_2$, where $\|g_1\|_2=\op(1)$ in certain cases when  $\beta_{A}=0$, according to Lemma \ref{lemma:bound-G_21} in Appendix \ref{sec:lemma}.
Consequently, the magnitude of $\sqrt{V_0}$ converges to zero at the rate $n^{-1/2}$ when $\beta_A\neq 0$ but potentially at a rate faster than $n^{-1/2}$ when $\beta_{A}=0$, making it difficult to \yinan{find an accuracy estimator of $\sqrt{V_0}$ for constructing} a valid test. 
To address this issue, we introduce a ridge to $V_0$ and consider the enlarged covariance matrix
\begin{equation}
	V=V_0+\frac{\tau}{n} I_q \qquad\text{ and its estimator } \qquad\hat V=\hat V_0+\frac{\tau}{n} I_q,
	\label{eq:V-def}
\end{equation}
where $I_q$ is the $q\times q$ identity matrix, and $\tau>0$ is a positive constant; a similar strategy was adopted in \cite{guo2021group} with recommended $\tau=0.5$ or $\tau=1$.

Based on the debiased estimator $\hat{\gamma}$ from \eqref{eq:hgamma-estimator} and the corresponding estimated variance $\hat{V}$ from \eqref{eq:V-def}, we propose the test statistics $\|T\|_\infty$ with 
\begin{equation*}
	T=\left(\hat{\gamma}_j / \sqrt{\hat{V}_{jj}},~ j=1,\ldots, q\right),
\end{equation*}
where $\hat{\gamma}_j$ is the $j$th element in $\hat{\gamma}$ and $\hat{V}_{jj}$ is the $j$th diagonal element in $\hat{V}$.
By applying the Bonferroni correction criterion, we adopt the following test
\begin{equation}
	\phi_{\alpha} = \indicator{\|T\|_{\infty}>\Phi^{-1}(1-\alpha/(2q))},
	\label{eq:test-function-bonf}
\end{equation}
where $\Phi(\cdot)^{-1}$ is the quantile function of the standard Gaussian distribution. The corresponding p-value is given by 
\[
P = \min_{1\le j \le q} q \cdot 2 \prob\left(Z>|T_j|\right),
\]
where $Z\sim N(0,1)$. We reject the null hypothesis in \eqref{test-problem-product} when $\phi_{\alpha}=1$ or equivalently when $P<\alpha$. In this case, there is statistically significant evidence for the presence of the overall mediation effect between the outcome and exposures.

\begin{remark}\label{rem:ext}
	An extension for incorporating additional covariates is provided in Appendix \ref{sec:incorp-covariates}.
\end{remark}

\section{Theoretical Results}\label{sec:theory}

\subsection{\label{sec:theory-assump}Assumptions}
To state our assumptions, we first introduce the concepts of subGaussianity and norm-subGaussianity. A real-valued random variable $S$ is subGaussian with a parameter $\sigma>0$ if $\prob(|S-\expect S|\geq t)\leq 2e^{-t^2/(2\sigma^2)}$. When $S\in\real^p$ is a random vector, it is subGaussian with a parameter $\sigma>0$ if $v^\top S$ is subGaussian with the parameter $\sigma$ for all $v\in\real^p$ such that $\|v\|_2=1$. 
A random vector $S\in \real^{p}$ is norm-subGaussian  with a parameter $\sigma>0$ \citep{jin2019short}, if $\prob\left(\|S-\expect S\|_2\ge t\right) \le 2 e^{-t^2/(2\sigma^2)}$ for all $t\in \real$. Norm-subGaussianity generalizes the usual subGaussianity as	both subGaussian and bounded random vectors are norm-subGaussian \citep[Lemma 1,][]{jin2019short}. 
Recall  $X=(A, M)$, $\theta=(\theta_A^{\transpose}, \theta_M^{\transpose})^{\transpose}$,  $Z=(Z_{1}, \ldots, Z_{n})^{\transpose}$, $E=(E_1, \ldots, E_n)^{\transpose}$, and $E_{M}=(E_{M,1}, \ldots, E_{M,n})^{\transpose}=E \theta_{M}$ in \eqref{eq:mediator-exposure-matrix}  and \eqref{eq:Y-X-equation}. 
To study the theoretical properties of the proposed test, we require the following assumptions. 

\begin{enumerate}
	\item[(A1)] $X_1,\ldots,X_n$ are i.i.d. $(q+p)$-dimensional centered subGaussian  random vectors with its covariance $\Sigma_X$ satisfying $c_0 \le \Lambda_{\min}(\Sigma_X) \le \Lambda_{\max}(\Sigma_X) \le C_0$ for positive constants $C_0 \ge c_0 > 0$. 
	The error vectors $E_1,\ldots,E_n$ are i.i.d.  and satisfy the moment conditions $\expect[E_i|A_i]=0$ and $\var[E_i|A_i]=\Sigma_{E}$ for some unknown symmetric positive semi-definite matrix $\Sigma_{E}$ with $\|\Sigma_{E}\|_2<\infty$. In addition,  conditional on $A_i$, each $E_i$ is norm-subGaussian with a parameter $\sigma$.
	Also, the {i.i.d.} variables $Z_i$  are subGaussian and satisfy the moment conditions $\expect[Z_{i}|X_i]=0$ and $\expect[Z_{i}^2|X_i]=\sigma_Z^2$ for some unknown positive constant $0< \sigma_Z^2 < \infty$. 
	\item[(A2)] $\max_{1\le i \le n}\max\left\{\expect[E_{M,i}^{2+\nu}|X_{i}], \expect[Z_{i}^{2+\nu}|X_{i}]\right\} \le M_0$ for  constants $\nu, M_0 >0$.
\end{enumerate}
Assumptions similar to (A1) are commonly adopted in the high-dimensional statistics and high-dimensional mediation analysis \citep{buhlmann2011statistics,guo2022statistical}. 
Assumption (A1) implies $\expect[E_{M,i}|A_i]=0$ and $\expect[E_{M,i}^2|A_i]=\sigma_E^2$ for some positive constant $0\le \sigma_E^2<\infty$. We remark that no independence between $X_i$, $E_i$ and $Z_i$ is required. In contrast, \cite{zhou2020estimation} assumed independence between $E_i$ and $A_i$ and independence between $Z_{i}$ and $X_i$. In addition to these independence assumptions, \cite{guo2022statistical} further assumed that $E_i$ and $Z_{i}$ are independent. In fact, our conditional moment assumptions on $E_i$ and $Z_{i}$ in Assumption (A1), satisfied under the independence assumptions  of \cite{zhou2020estimation} and \cite{guo2022statistical}, can accommodate more general dependence structures among $X_i$, $E_i$ and $M_i$. Moreover, Assumption (A1) guarantees the existence of the ordinary least-squares estimator $\hat{\beta}_{A}$ with probability tending to one, even when $q$ grows with $n$; see the proof of Lemma \ref{lemma:bound-G_21} in Appendix \ref{sec:lemma} for details.
Assumption (A2) poses further mild moment conditions on $E_{M,i}$ and $Z_i$.

We also assume the following  conditions on the initial estimators.
\begin{enumerate}
	\item[(B1)] With probability larger than $1-h(n)$, $\|\hat{\theta}-\theta\|_1 \lesssim s \sqrt{\log(q+p)/n}$, where $h(n)\to 0$ as $n \to \infty$.
	\item[(B2)] $| \hat{\sigma}_Z^2/\sigma_Z^2 - 1 | \stackrel{p}{\to} 0$ and $\hat{\sigma}_E^2+\hat{\sigma}_Z^2 \stackrel{p}{\to} \sigma_E^2+\sigma_Z^2$, as $n\to \infty$.
\end{enumerate}
The convergence of $\hat\sigma_Z^2$ in Assumption (B2) holds when $\sigma_Z^2$ is estimated by the scaled Lasso \citep{sun2012scaled}. 

Assumption (B1) can be satisfied by Lasso-type estimators under the compatibility condition and certain sparsity structures, where $s$ serves as a sparse parameter. To introduce the compatibility condition, given a set $\mathcal{S} \subset \{1, \ldots, q+p\}$, for any positive number $\eta\ge 1$, define the set
\begin{equation}
	\mathcal{C}(\mathcal{S}, \eta) = \{
	u \in \real^{q+p}: \|u_{\mathcal{S}^c}\|_1 \le \eta \|u_{\mathcal{S}}\|_1 \},
	\label{eq:sparse-cone}
\end{equation}
where $u_{\mathcal{S}}$ is the sub-vector of $u$ with coordinates in $\mathcal{S}$ and $\mathcal{S}^c$ is the complement of $\mathcal{S}$. 	
For a symmetric positive semi-definite matrix $\Sigma_0$, define the compatibility constant $\phi_0(\Sigma_0, \mathcal{S}, \eta)$ for $\Sigma_0$ with respect to $\mathcal{C}(\mathcal{S}, \eta)$ via
\[
\phi_0^2(\Sigma_0, \mathcal{S}, \eta) = \inf\left\{ \frac{u^{\transpose} \Sigma_0 u |\mathcal{S}|}{\|u\|_1^2} : u\in \mathcal{C}(\mathcal{S}, \eta), u\ne 0 \right\},
\]
where $|\mathcal S|$ is the cardinality of $\mathcal S$. 
We say a design matrix $X$ satisfies the compatibility condition over $\mathcal{S}$ with a parameter $\eta$ if $\phi_0(\hat{\Sigma}, \mathcal{S}, \eta)>0$, where $\hat{\Sigma}=n^{-1} X^{\transpose} X$ is the sample covariance matrix of $X$.
Similarly, we say that a (population) covariance matrix $\Sigma$ satisfies the compatibility condition over $\mathcal{S}$ with a parameter $\eta$ if $\phi_0(\Sigma, \mathcal{S}, \eta)>0$.
It is known that the compatibility condition is implied by the restricted eigenvalue condition, which is another well-known condition used to establish the consistency of Lasso-type estimators \citep{van2009conditions}. 

For the sparsity structures, in this paper, we consider the capped-$\ell_1$ sparse structure on the regression coefficient $\theta$, 
\begin{equation}
	\sum_{j=1}^{q+p} \min\{|\theta_j|/(\sigma_Z \lambda_0), 1\} \le s
	\label{eq:capped-l1}
\end{equation}
with $\lambda_0=\sqrt{2\log(q+p) / n}$,  
which has been extensively examined \citep{sun2013sparse, zhang2014confidence, cai2021optimal}. As highlighted in \cite{zhang2014confidence}, the capped-$\ell_1$ condition in \eqref{eq:capped-l1} holds when $\theta$ is $\ell_0$ sparse with $\|\theta\|_0\le s$ or $\ell_r$ sparse with $\|\theta\|_r^r / (\sigma_Z \lambda_0)^r \le s$ for $0<r\le 1$. 
The following proposition shows that Lasso-type estimators satisfy Assumption (B1).

\begin{proposition}
	\label{prop:CC-general-require}
	Suppose that Assumption (A1) holds and the population covariance matrix $\Sigma_X$ satisfies the compatibility condition over some $\mathcal{S}\subset \{1, \ldots, q+p\}$ with a parameter $\eta$. 
	Then, there exists a constant $C>0$, not depending on $q+p$, such that, for all sufficiently large $n$, with probability at least $1-C(q+p)^{-1}$, 
	the design matrix $X$ satisfies the compatibility condition over $\mathcal{S}$ with the parameter $\eta$. 
	Consequently, if additionally $\theta$ in \eqref{eq:Y-X-equation} possesses the capped-$\ell_1$ sparsity {and $\log(q+p)\le n/(32e^2)$} for all sufficiently large $n$,
	then the following statements hold with probability approaching one: 
	\begin{enumerate}
		\item The Lasso estimator $\hat{\theta}$ with the tuning parameter $\lambda_n \asymp \sqrt{\log(q+p) / n}$ satisfies Assumption (B1).
		\item The scaled Lasso estimators $(\hat{\theta}, \hat{\sigma}_{Z})$ \citep{sun2012scaled} satisfy Assumptions (B1) and (B2). 
	\end{enumerate}
\end{proposition}

\subsection{Validity and Consistency}\label{sec:theory-main-results}

We first present an asymptotic normality result about $\hat\gamma$. Let $G=(g_1, \ldots, g_q)\in \real^{(q+p)\times q}$, where $g_j$ is defined in \eqref{eq:gj-def}.

\begin{theorem}\label{thm:asym-norm-hgamma-decomp}
	Suppose that Assumptions (A1), (A2) and (B1) hold and $q\ll \min\{p^\zeta,\sqrt n\}$ for a (arbitrarily large but fixed) constant $\zeta>0$.
	Let $W = n^{-1} \hat{\Sigma}_{A}^{-1} A^{\transpose} E_{M} + n^{-1} \hat{U}^{\transpose} X^{\transpose} Z$ and $B = (\hat{\Sigma}_{X} \hat{U} - G)^{\transpose} (\theta - \hat{\theta})$. Then, we have
	\[
	\hat{\gamma} - \gamma = W + B, 
	\quad
	\text{with}
	\quad 
	W_{j}/\sqrt{V_{jj}} \stackrel{d}{\to} N(0, 1) ~\text{for each}~j=1,\ldots, q,
	\]
	where {$\hat{\gamma}$ is given in \eqref{eq:hgamma-estimator}}, $W_{j}$ is the $j$th element in $W$ and $V_{jj}$ is the $j$th diagonal element in $V$ defined in \eqref{eq:V-def}. Further, with probability tending to one, $\|DB\|_{\infty} \lesssim s q\log(q+p)/\sqrt{n}$, where $D$ is the diagonal matrix with diagonal elements $1/\sqrt{V_{jj}}, j=1,\ldots, q$. 
\end{theorem}

In the above theorem, we allow the number $q$ of exposures to grow with $n$ in a polynomial rate, even though in practice most applications typically investigate a limited set of exposures/treatments. Equipped with this theorem, we are ready to analyze the theoretical properties of the proposed test $\phi_{\alpha}$.  In our analysis we consider the null hypothesis parameter space $\mathcal{H}_0(s)=\left\{ \xi\in \Xi(s): \beta_A^{\transpose} \theta_M = \zerovec{q} \right\}$ and the local alternative parameter space
\[
\mathcal{H}_1(s, \delta) = \left\{ \xi \in \Xi(s): {\beta_A^{\transpose} \theta_M = \delta/\sqrt{n}} \right\}, \quad \text{for some } \delta \in \real^q,
\]
with 
\[
\Xi(s) = \left\{ 
\xi=(\theta, \sigma_Z, \Sigma_X)\left| 
\begin{array}{c}
	\sum_{j=1}^{q+p} \min\{|\theta_j|/(\sigma_Z \lambda_0), 1\} \le s,~ 0< \sigma_Z \le M_0, \\
	c_0 \le \Lambda_{\min}(\Sigma_X) \le \Lambda_{\max}(\Sigma_X) \le C_0
\end{array}
\right.
\right\}
\] that encodes a capped-$\ell_1$ condition only on the vector $\theta$, 
where $M_0>0$ and $C_0\ge c_0 >0$ are positive constants defined in Assumptions (A1) and (A2), and $\lambda_0=\sqrt{2\log(q+p) / n}$. 
Recall $\Phi(\cdot)$ denotes the cumulative distribution function of the standard normal distribution, and let $\prob_\xi$ be the probability measure induced by the parameter $\xi$. Define $$F(\alpha, x, q) = q \left\{ \Phi\left(x+\Phi^{-1}\left(1-{\alpha}/{(2q)}\right)\right) - \Phi\left(x-\Phi^{-1}\left(1-{\alpha}/{(2q)}\right)\right) \right\}$$ for any $x\in\real$, as well as  $\delta_{\max}=\max_{1\le k \le q}|\delta_{k}|$ and  $\beta_{A,\max}=\max_{1\le k \le q}\|\beta_{A,k}\|_2$.

\begin{theorem}\label{thm:test-bonf}
	Suppose that Assumptions (A1), (A2) and (B1), (B2) hold,  $s \ll \sqrt{n}/(q\log(q+p))$ and {$q\ll \min\{p^\zeta,\sqrt n\}$} for a constant $\zeta>0$.
	\begin{itemize}
		\item (Validity).  For all $\alpha\in(0,1)$, 
		$$\underset{n\to \infty}{\overline{\lim}} \sup_{\xi \in \mathcal{H}_0(s)} \prob_{\xi}(\phi_{\alpha}=1) \le \alpha,$$ 
		that is,  the proposed test $\phi_{\alpha}$  in \eqref{eq:test-function-bonf} is asymptotically valid.
		\item (Power). The  power of the test $\phi_{\alpha}$ under the local alternative is asymptotically lower bounded by $1-F(\alpha,\Delta_n,q)$, where  
		$$\Delta_n={C \delta_{\max}}/{({\sigma\sqrt{\log(pn)/n}}+\beta_{A,\max}+C)}$$ for some constant $C>0$. That is,  
		\begin{align*}
			\lim_{n\to \infty} \inf_{\xi \in \mathcal{H}_1(s, \delta)} 
			\frac{\prob_{\xi}(\phi_{\alpha}=1)}{1 - F\left(\alpha, \Delta_n, q\right)}
			\ge
			1
		\end{align*}
		when $\lim_{n\rightarrow\infty}\{1-F(\alpha,\Delta_n,q)\}>0$.
	\end{itemize}
\end{theorem}

In contrast with \cite{zhou2020estimation, guo2022statistical} that require additional  assumptions on $\beta_{A}$,  Theorem \ref{thm:test-bonf} shows that the proposed test remains valid for \textit{arbitrary} $\beta_{A}$, as the parameter space $\Xi$ encodes no assumptions on $\beta_A$. While $q$ (the number of exposures) is allowed to grow with $n$ in Theorem \ref{thm:test-bonf}, in practice the number of exposures or treatments is often low-dimensional.
\yinan{
Note that for any fixed (or sufficiently slowly diverging) $q\ge 1$ and $\alpha>0$, the function $F(\alpha, x, q)$ is monotone decreasing with respect to $x$. That is, a larger value of $\Delta_n$ results in a higher power for the proposed test $\phi_{\alpha}$. 
Furthermore, if mediation effect is of constant order, i.e., $\|\beta_{A}^{\transpose}\theta_{M}\|_{\infty}= O(1)$, then $\delta_{\max}\asymp \sqrt{n}$. This,  together with constant-order exposure-mediator coefficients ($\beta_{A,\max}=O(1)$), further implies that $\Delta_n\to \infty$ as $n\to \infty$. Consequently, the proposed test $\phi_{\alpha}$ achieves asymptotic power one. 
}

\section{Simulation Study}\label{sec:simu}

In this section, we conduct numerical simulations to assess the performance of the proposed test $\phi_{\alpha}$ (denoted as \yinan{$\text{Bonf}$-$1$}) described in Section \ref{sec:meth-tesing}. We employ the \yinan{scaled} Lasso estimator as the initial estimator for $\theta$ and select the related tuning parameter by the quantile-based penalty procedure in the \verb|R| package \verb|scalreg| \citep{sun2019scalreg}.
The tuning parameter $\lambda$ in \eqref{VePD-constraint-bias} and \eqref{VePD-constraint-var} is chosen by the searching procedure in the \verb|R| package \verb|SIHR| \citep{wang2022SIHR}. We set $\tau=1$ in \eqref{eq:V-def}.
Throughout the simulations, we consider the model specified by \eqref{eq:mediator-exposure} and \eqref{eq:outcome-mediator} as our data generating model, and adopt simulation settings similar to those of  \cite{zhou2020estimation} and \cite{guo2022statistical}.

\subsection{Size}\label{sec:simu-size}

The null hypothesis of \eqref{test-problem-product} is composed of the following four cases: 1) $\beta_A=0$ and $\theta_M=0$, 2) $\beta_A\ne 0$ and $\theta_M=0$, 3) $\beta_A=0$ and $\theta_M\ne 0$ and 4) $\beta_A^{\transpose} \theta_M=0$ but $\beta_A\ne 0$ and $\theta_M \ne 0$. 
In this section, we numerically assess the validity of the proposed test procedure and the alternative existing methods for each of the these cases. 

We consider $n=50$ and $n=300$ which correspond to small and large sample sizes, respectively. We set $p=n$ and consider two cases for $q$: $q=1$ and $q=3$.  
When the true coefficient vector $\theta_M\in \real^{p}$ is nonzero, \yinan{we investigate a hard sparsity case with a fixed sparsity parameter $s=5$: $\theta_{M,k} = 0.2 k \indicator{1 \le k \le s}$ for $k=1, \ldots, p$. We also explore other types of sparsity in Appendix \ref{sec:additional-simu}.}
When the true coefficient matrix $\beta_A\in\real^{p\times q}$ is nonzero, \yinan{we consider two scenarios: a sparse setting and a dense setting. In the sparse setting, each element of $\beta_A$ is given by
\[
\beta_{A,jk} = \begin{cases}
	0.2 \kappa_j(k-s) & \text{ if }s+1\leq k\leq 2s \\
	0 & \text{ otherwise,}
\end{cases}
\]
for $s=5$, $k=1,\ldots, p$ and $j=1,\ldots, q$, where each $\kappa_j(\cdot)$ is a random permutation of $\{1, \ldots, s\}$. In the dense setting, 
each element of $\beta_A$ is given by
\[
\beta_{A,jk} = \begin{cases}
	0.2 & \text{ if }s+1\leq k\leq p/2 \\
	0 & \text{ otherwise,}
\end{cases}
\]
for $s=5$, $k=1,\ldots, p$ and $j=1,\ldots, q$.}
For the coefficient $\theta_{A}$, we set $\theta_{A}=c_0 \onevec{q}$ with $c_0=0.5$ for $q=1$ and $c_0=0.3$ for $q=3$.

For the exposure $A_i=(A_{ij})_{j=1}^q$, the coordinates of $A_i$ are independently drawn from a distribution, for which we consider two cases: The Bernoulli distribution $\text{Bern}(0.5)$ and the Gaussian distribution $N(0, 0.5^2)$. 
\yinan{For error terms, we set} $Z_i\sim N(0,1)$ and $E_i\sim N_q(0,\Sigma_0)$ \yinan{with two scenarios for the covariance structure. In the first scenario, we use $\Sigma_{0}=\Sigma_{AR}=(0.5^{|i-j|})_{1\le i,j\le p}$, representing a covariance matrix with an AR(1) structure. In the second scenario, we consider $\Sigma_0=\Sigma_{CS}$ for a compound symmetric covariance matrix $\Sigma_{CS}$, where the diagonal elements of $\Sigma_{CS}$ are $8$ and the other elements are $6.4$. This choice of $\Sigma_{CS}$ is a more challenging setting where the mediators exhibit strong correlations, corresponding to a scenario where the irrepresentable condition may fail \citep{zhao2006model}.}
For each experimental setting, we independently generate $R=500$ replications. For each replication, the empirical size is computed as the proportion of rejections among the $R$ replications when the composite null hypothesis of \eqref{test-problem-product} is true.

\yinan{
For comprehensive comparisons, we consider several competitors. 
First, to assess the effect of $\tau$ in the proposed covariance matrix \eqref{eq:V-def}, we introduce a test, denoted $\text{Bonf}$-$0$, which is similar to $\text{Bonf}$-$1$ but with $\tau = 0$.
Next, recall that our proposed test \eqref{eq:test-function-bonf} is based on the Bonferroni correction. Another possible approach is to perform a $\chi^2$ test using the de-biased estimator $\hat{\gamma}$ from \eqref{eq:hgamma-estimator} and the corresponding estimated variance $\hat{V}$ from \eqref{eq:V-def}. We thus consider the $\chi^2$ tests with $\tau = 0$ and $\tau = 1$ in \eqref{eq:V-def}, referred to as $\chi^2$-$0$ and $\chi^2$-$1$, respectively. Note that when $q = 1$, the proposed Bonferroni-based test \eqref{eq:test-function-bonf} and this $\chi^2$ test are equivalent.
Additionally, we include the methods of \cite{zhou2020estimation} (denoted Zhou-$0$) and \cite{guo2022statistical} (denoted Guo-$0$), which also address the overall mediation effect $\gamma$ and the test problem \eqref{test-problem-product}. The tuning parameters for Zhou-$0$ and Guo-$0$ were selected using the \texttt{R} package developed by \cite{zhou2020estimation} and the method described in \cite{guo2022statistical}, respectively.
To examine whether the existing methods can benefit from the enlarged covariance strategy, as in $\hat{V}$ from \eqref{eq:V-def}, we also investigate tests in \cite{zhou2020estimation} and \cite{guo2022statistical} based on the covariance matrices with an additional enlarged factor $1/n \cdot I_q$ as in \eqref{eq:V-def}, referred to as Zhou-$1$ and Guo-$1$, respectively.
}

Table \ref{tab:size-Bern} presents the numerical results 
for the case when the exposures follow the Bernoulli distribution $\text{Bern}(0.5)$; the results for Gaussian exposures are provided in Appendix \ref{sec:additional-simu}. 
\yinan{
The results show that the proposed test, $\text{Bonf}$-$1$, remains valid across all cases of the composite null hypothesis in \eqref{test-problem-product} and for all settings of $\beta_A$. 
In contrast, Zhou-$0$ and Zhou-$1$ are invalid when $\theta_M = 0$, while Guo-$0$ and Guo-$1$ tend to fail when the sample size is small ($n=50$) or mediators are highly correlated ($\Sigma_0 = \Sigma_{\text{CS}}$).
Interestingly, although the theoretical analysis in \cite{guo2022statistical} requires a uniform signal strength condition for $\theta_{M}$ (i.e., $\theta_M$ should be sufficiently large), the method Guo-$0$ and Guo-$1$ appear to be conservative when $\Sigma_0=\Sigma_{AR}$, $\theta_{M}=0$ and $n$ is relatively large. This phenomenon might be attributed to the relatively aggressively selected tuning parameter that tends to yield a zero estimator of $\theta_M$ when $\theta_{M}=0$. 
We would like to remark that the results presented above do not suggest the failure of the existing methods. Rather, they underscore that these methods require stronger conditions to operate effectively. When these conditions are met, these methods often exhibit commendable power.
}

\yinan{
In addition, the results in Table \ref{tab:size-Bern} provide further insights.
First, the proposed test benefits primarily from two key components: the VePD de-biasing technique and the regularization $\tau/n \cdot I_q$ of the variance. Comparing the results of $\text{Bonf}$-$1$ with those of $\text{Bonf}$-$0$, the impact of $\tau$ appears to be relatively modest. Furthermore, Zhou-$1$ (Guo-$1$) shows only a slight adjustment over Zhou-$0$ (Guo-$0$). These findings suggest that the proposed test's effectiveness largely stems from the VePD technique, which differs fundamentally from the methods of \cite{zhou2020estimation} and \cite{guo2022statistical}.
Second, the performance of the proposed Bonferroni-based test ($\text{Bonf}$-$1$) is comparable to the $\chi^2$ test ($\chi^2$-$1$) in most simulation scenarios. This indicates that the proposed Bonferroni-based test is not overly conservative, even when compared to a method that accounts for the correlation structure of the de-biased estimator.
}

\setlength{\tabcolsep}{4pt} 
\renewcommand{\arraystretch}{0.85} 
\begin{table}[ht]
	\centering
	\caption{Empirical sizes for Bernoulli exposures with a nonzero $\theta_{A}=c_0 \onevec{q}$, where $c_0=0.5$ when $q=1$, and $c_0=0.3$ when  $q=3$.}
	\begin{tabular}{cccccccccccc}
		\cmidrule{1-12}
		Sparsity for $\theta_M$ & Case  & $q$   & $n$   & Zhou-$0$ & Zhou-$1$ & Guo-$0$ & Guo-$1$ & $\chi^2$-$0$ & $\chi^2$-$1$ & $\text{Bonf}$-$0$ & $\text{Bonf}$-$1$ \\
		\cmidrule{1-12}
		\multicolumn{1}{c}{\multirow{24}[24]{*}{\shortstack{Zero\\\newline{}($\theta_M=0$)}}} & \multicolumn{1}{c}{\multirow{4}[4]{*}{\shortstack{$\beta_A=0$\\\newline{}($\Sigma_0=\Sigma_{AR}$)}}} & \multirow{2}[2]{*}{1} & 50    & 0.180 & 0.164 & 0.208 & 0.208 & 0.056 & 0.046 & 0.056 & 0.046 \\
		&       &       & 300   & 0.662 & 0.652 & 0.000 & 0.000 & 0.034 & 0.032 & 0.034 & 0.032 \\
		\cmidrule{3-12}          &       & \multirow{2}[2]{*}{3} & 50    & 0.144 & 0.122 & 0.580 & 0.580 & 0.036 & 0.018 & 0.040 & 0.022 \\
		&       &       & 300   & 0.534 & 0.512 & 0.000 & 0.000 & 0.040 & 0.032 & 0.048 & 0.042 \\
		\cmidrule{2-12}          & \multicolumn{1}{c}{\multirow{4}[4]{*}{\shortstack{$\beta_A\ne 0$-Sparse\\\newline{}($\Sigma_0=\Sigma_{AR}$)}}} & \multirow{2}[2]{*}{1} & 50    & 0.154 & 0.152 & 0.232 & 0.232 & 0.066 & 0.052 & 0.066 & 0.052 \\
		&       &       & 300   & 0.198 & 0.196 & 0.000 & 0.000 & 0.038 & 0.036 & 0.038 & 0.036 \\
		\cmidrule{3-12}          &       & \multirow{2}[2]{*}{3} & 50    & 0.208 & 0.192 & 0.596 & 0.594 & 0.040 & 0.038 & 0.064 & 0.042 \\
		&       &       & 300   & 0.496 & 0.496 & 0.000 & 0.000 & 0.030 & 0.026 & 0.028 & 0.028 \\
		\cmidrule{2-12}          & \multicolumn{1}{c}{\multirow{4}[4]{*}{\shortstack{$\beta_A\ne 0$-Dense\\\newline{}($\Sigma_0=\Sigma_{AR}$)}}} & \multirow{2}[2]{*}{1} & 50    & 0.188 & 0.176 & 0.210 & 0.208 & 0.062 & 0.052 & 0.062 & 0.052 \\
		&       &       & 300   & 0.220 & 0.218 & 0.000 & 0.000 & 0.040 & 0.036 & 0.040 & 0.036 \\
		\cmidrule{3-12}          &       & \multirow{2}[2]{*}{3} & 50    & 0.164 & 0.152 & 0.556 & 0.556 & 0.058 & 0.034 & 0.054 & 0.040 \\
		&       &       & 300   & 0.698 & 0.696 & 0.000 & 0.000 & 0.042 & 0.042 & 0.052 & 0.046 \\
		\cmidrule{2-12}          & \multicolumn{1}{c}{\multirow{4}[4]{*}{\shortstack{$\beta_A=0$\\\newline{}($\Sigma_0=\Sigma_{CS}$)}}} & \multirow{2}[2]{*}{1} & 50    & 0.086 & 0.072 & 0.122 & 0.122 & 0.084 & 0.068 & 0.084 & 0.068 \\
		&       &       & 300   & 0.216 & 0.206 & 0.060 & 0.056 & 0.034 & 0.034 & 0.034 & 0.034 \\
		\cmidrule{3-12}          &       & \multirow{2}[2]{*}{3} & 50    & 0.096 & 0.074 & 0.288 & 0.288 & 0.050 & 0.030 & 0.056 & 0.024 \\
		&       &       & 300   & 0.236 & 0.218 & 0.068 & 0.066 & 0.054 & 0.046 & 0.054 & 0.048 \\
		\cmidrule{2-12}          & \multicolumn{1}{c}{\multirow{4}[4]{*}{\shortstack{$\beta_A\ne 0$-Sparse\\\newline{}($\Sigma_0=\Sigma_{CS}$)}}} & \multirow{2}[2]{*}{1} & 50    & 0.122 & 0.110 & 0.152 & 0.152 & 0.082 & 0.068 & 0.082 & 0.068 \\
		&       &       & 300   & 0.688 & 0.686 & 0.082 & 0.078 & 0.042 & 0.040 & 0.042 & 0.040 \\
		\cmidrule{3-12}          &       & \multirow{2}[2]{*}{3} & 50    & 0.108 & 0.098 & 0.336 & 0.336 & 0.072 & 0.048 & 0.070 & 0.036 \\
		&       &       & 300   & 0.746 & 0.742 & 0.100 & 0.098 & 0.062 & 0.058 & 0.058 & 0.050 \\
		\cmidrule{2-12}          & \multicolumn{1}{c}{\multirow{4}[4]{*}{\shortstack{$\beta_A\ne 0$-Dense\\\newline{}($\Sigma_0=\Sigma_{CS}$)}}} & \multirow{2}[2]{*}{1} & 50    & 0.094 & 0.082 & 0.138 & 0.138 & 0.074 & 0.056 & 0.074 & 0.056 \\
		&       &       & 300   & 0.362 & 0.354 & 0.076 & 0.074 & 0.038 & 0.034 & 0.038 & 0.034 \\
		\cmidrule{3-12}          &       & \multirow{2}[2]{*}{3} & 50    & 0.086 & 0.072 & 0.290 & 0.288 & 0.052 & 0.024 & 0.054 & 0.022 \\
		&       &       & 300   & 0.332 & 0.320 & 0.082 & 0.082 & 0.052 & 0.044 & 0.060 & 0.056 \\
		\cmidrule{1-12}
		\multicolumn{1}{c}{\multirow{24}[24]{*}{\shortstack{Hard\\ ($\theta_M\ne 0$)}}} & \multicolumn{1}{c}{\multirow{4}[4]{*}{\shortstack{$\beta_A=0$\\\newline{}($\Sigma_0=\Sigma_{AR}$)}}} & \multirow{2}[2]{*}{1} & 50    & 0.034 & 0.034 & 0.088 & 0.088 & 0.038 & 0.034 & 0.038 & 0.034 \\
		&       &       & 300   & 0.044 & 0.044 & 0.042 & 0.042 & 0.042 & 0.042 & 0.042 & 0.042 \\
		\cmidrule{3-12}          &       & \multirow{2}[2]{*}{3} & 50    & 0.044 & 0.042 & 0.168 & 0.166 & 0.054 & 0.050 & 0.056 & 0.054 \\
		&       &       & 300   & 0.038 & 0.038 & 0.044 & 0.044 & 0.034 & 0.034 & 0.036 & 0.034 \\
		\cmidrule{2-12}          & \multicolumn{1}{c}{\multirow{4}[4]{*}{\shortstack{$\beta_A\ne 0$-Sparse\\\newline{}($\Sigma_0=\Sigma_{AR}$)}}} & \multirow{2}[2]{*}{1} & 50    & 0.036 & 0.036 & 0.106 & 0.106 & 0.048 & 0.044 & 0.048 & 0.044 \\
		&       &       & 300   & 0.050 & 0.050 & 0.042 & 0.042 & 0.050 & 0.046 & 0.050 & 0.046 \\
		\cmidrule{3-12}          &       & \multirow{2}[2]{*}{3} & 50    & 0.048 & 0.048 & 0.192 & 0.192 & 0.046 & 0.046 & 0.062 & 0.056 \\
		&       &       & 300   & 0.074 & 0.074 & 0.044 & 0.044 & 0.044 & 0.042 & 0.038 & 0.036 \\
		\cmidrule{2-12}          & \multicolumn{1}{c}{\multirow{4}[4]{*}{\shortstack{$\beta_A\ne 0$-Dense\\\newline{}($\Sigma_0=\Sigma_{AR}$)}}} & \multirow{2}[2]{*}{1} & 50    & 0.036 & 0.036 & 0.086 & 0.086 & 0.048 & 0.046 & 0.048 & 0.046 \\
		&       &       & 300   & 0.048 & 0.048 & 0.042 & 0.042 & 0.044 & 0.042 & 0.044 & 0.042 \\
		\cmidrule{3-12}          &       & \multirow{2}[2]{*}{3} & 50    & 0.046 & 0.044 & 0.164 & 0.164 & 0.050 & 0.046 & 0.056 & 0.046 \\
		&       &       & 300   & 0.066 & 0.066 & 0.044 & 0.044 & 0.026 & 0.024 & 0.024 & 0.024 \\
		\cmidrule{2-12}          & \multicolumn{1}{c}{\multirow{4}[4]{*}{\shortstack{$\beta_A=0$\\\newline{}($\Sigma_0=\Sigma_{CS}$)}}} & \multirow{2}[2]{*}{1} & 50    & 0.058 & 0.058 & 0.058 & 0.058 & 0.032 & 0.030 & 0.032 & 0.030 \\
		&       &       & 300   & 0.062 & 0.062 & 0.060 & 0.060 & 0.048 & 0.048 & 0.048 & 0.048 \\
		\cmidrule{3-12}          &       & \multirow{2}[2]{*}{3} & 50    & 0.052 & 0.052 & 0.072 & 0.072 & 0.032 & 0.028 & 0.034 & 0.034 \\
		&       &       & 300   & 0.066 & 0.066 & 0.048 & 0.048 & 0.042 & 0.042 & 0.040 & 0.040 \\
		\cmidrule{2-12}          & \multicolumn{1}{c}{\multirow{4}[4]{*}{\shortstack{$\beta_A\ne 0$-Sparse\\\newline{}($\Sigma_0=\Sigma_{CS}$)}}} & \multirow{2}[2]{*}{1} & 50    & 0.056 & 0.056 & 0.058 & 0.058 & 0.048 & 0.048 & 0.048 & 0.048 \\
		&       &       & 300   & 0.062 & 0.062 & 0.060 & 0.060 & 0.046 & 0.044 & 0.046 & 0.044 \\
		\cmidrule{3-12}          &       & \multirow{2}[2]{*}{3} & 50    & 0.052 & 0.052 & 0.072 & 0.072 & 0.042 & 0.042 & 0.034 & 0.034 \\
		&       &       & 300   & 0.056 & 0.056 & 0.048 & 0.048 & 0.056 & 0.056 & 0.062 & 0.062 \\
		\cmidrule{2-12}          & \multicolumn{1}{c}{\multirow{4}[4]{*}{\shortstack{$\beta_A\ne 0$-Dense\\\newline{}($\Sigma_0=\Sigma_{CS}$)}}} & \multirow{2}[2]{*}{1} & 50    & 0.058 & 0.058 & 0.058 & 0.058 & 0.036 & 0.034 & 0.036 & 0.034 \\
		&       &       & 300   & 0.060 & 0.060 & 0.060 & 0.060 & 0.050 & 0.050 & 0.050 & 0.050 \\
		\cmidrule{3-12}          &       & \multirow{2}[2]{*}{3} & 50    & 0.050 & 0.050 & 0.074 & 0.074 & 0.036 & 0.034 & 0.038 & 0.038 \\
		&       &       & 300   & 0.072 & 0.072 & 0.048 & 0.048 & 0.046 & 0.046 & 0.046 & 0.044 \\
		\cmidrule{1-12}
	\end{tabular}%
	\label{tab:size-Bern}%
\end{table}%

\subsection{Power}\label{sec:simu-power}


\yinan{To examine the power behavior of the proposed test, we} adopt data generating processes similar to those described in Section \ref{sec:simu-size}, but with different settings for $\beta_{A}$ and $\theta_{M}$. Specifically, given positive constants $c_1, c_2 \in \real$, we set the true coefficient vector $\theta_M$ \yinan{with a hard-sparse structure: $\theta_{M,k} = 0.3 c_2  k \indicator{1 \le k \le s}$ for $k=1, \ldots, p$. Other sparsity types are considered in Appendix \ref{sec:additional-simu}.}
Also, set the true coefficient matrix $\beta_{A}$ with elements
\[
\beta_{A,jk} = c_1 \left(0.3 k \indicator{1 \le k \le s}+w_k \indicator{k\ge s+1}\right),
\]
with $w_k\sim N(0, 0.1^2)$ for $k=1, \ldots, p$ and $j=1,\ldots, q$.
Since both $c_1$ and $c_2$ are nonzero, the mediation effect $\gamma=\beta_{A}^{\transpose} \theta_{M}$ is nonzero and thus the alternative hypothesis of \eqref{test-problem-product} is true.

For simplicity we focus on the case of $q=1$ and consider two scenarios. In the first scenario, we fix $c_2=1$ (equivalent to fixing $\theta_{M}$) while vary $c_1$ (equivalent to varying $\beta_{A}$) with different values of $c_1\in\real$. In the second scenario, we fix $c_1=1$ (fix $\beta_{A}$) while vary $c_2$ (vary $\theta_{M}$) with different values of $c_2\in\real$. 
For each setting, the empirical power is computed as the proportion of rejections among the replications.

Table \ref{tab:power-Bern} presents the simulation results of proposed test \yinan{$\text{Bond}$-$1$} when the exposures follow the Bernoulli distribution $\text{Bern}(0.5)$, \yinan{in comparison with Zhou-$1$ and Guo-$1$}; the results for the Gaussian exposures are provided in Appendix \ref{sec:additional-simu}. 
\yinan{Note that $\chi^2$-$1$ is omitted as it is equivalent to $\text{Bond}$-$1$ when $q=1$.}
As shown in the table, the empirical power \yinan{of $\text{Bond}$-$1$} increases as the signal ({larger $c_1$ or $c_2$}) or the sample size $n$ increases, numerically supporting the results in Theorem \ref{thm:test-bonf}.
\yinan{Again, we observe that Zhou-$1$ and Guo-$1$ fail to control type-I errors in some cases of $c_1=0$ or $c_2=0$. 
In these cases, comparisons of power performance are not fair. However, in valid cases, Zhou-$1$ and Guo-$1$ demonstrate better power performance than the proposed test. 
Since the proposed test is designed to be valid across all null hypotheses, some loss of power is expected. It should also be noted that determining the validity of Zhou-$1$ and Guo-$1$ is already challenging, limiting the practical use of these methods. In contrast, the proposed method is more flexible and reliable.}

%
%
%

\setlength{\tabcolsep}{10pt} 
\renewcommand{\arraystretch}{1} 
\begin{table}[ht]
	\centering
	\caption{Empirical power behaviors for Bernoulli exposures.}
	\begin{tabular}{cccccc}
		\cmidrule{1-6}
		Case  & $n$   & $c_1$ or $c_2$ & Zhou-$1$ & Guo-$1$ & $\text{Bonf}$-$1$ \\
		\cmidrule{1-6}
		\multirow{10}[3]{*}{\shortstack{Fix $\theta_{M}$, \\ Vary $\beta_{A}$}} & \multirow{5}[2]{*}{50} & 0     & 0.032 & 0.092 & 0.036 \\
		&       & 1/8   & 0.104 & 0.136 & 0.086 \\
		&       & 1/4   & 0.306 & 0.282 & 0.172 \\
		&       & 1/2   & 0.806 & 0.708 & 0.452 \\
		&       & 1     & 1     & 0.99  & 0.822 \\
		\cmidrule{2-6}          & \multirow{5}[1]{*}{300} & 0     & 0.04  & 0.044 & 0.036 \\
		&       & 1/8   & 0.436 & 0.366 & 0.148 \\
		&       & 1/4   & 0.934 & 0.926 & 0.516 \\
		&       & 1/2   & 1     & 1     & 0.944 \\
		&       & 1     & 1     & 1     & 0.982 \\
		\cmidrule{1-6}  
		\multirow{10}[3]{*}{\shortstack{Fix $\beta_{A}$, \\ Vary $\theta_{M}$}} & \multirow{5}[1]{*}{50} & 0     & 0.12  & 0.246 & 0.04 \\
		&       & 1/8   & 0.726 & 0.306 & 0.104 \\
		&       & 1/4   & 0.92  & 0.528 & 0.252 \\
		&       & 1/2   & 0.944 & 0.814 & 0.492 \\
		&       & 1     & 0.946 & 0.926 & 0.764 \\
		\cmidrule{2-6}          & \multirow{5}[2]{*}{300} & 0     & 0.114 & 0     & 0.038 \\
		&       & 1/8   & 0.988 & 0.552 & 0.162 \\
		&       & 1/4   & 1     & 1     & 0.468 \\
		&       & 1/2   & 1     & 1     & 0.932 \\
		&       & 1     & 1     & 1     & 0.994 \\
		\cmidrule{1-6}
	\end{tabular}%
	\label{tab:power-Bern}%
\end{table}%

\section{Data Application on LUAD}\label{sec:application}

In recent decades, lung cancer has emerged as one of the most prevalent diseases worldwide \citep{WHO2022cancer}, with smoking being a well-known key risk factor \citep[e.g.,][]{shields2002molecular}. 
Therefore, there is a pressing need to understand the underlying biological mechanisms that link smoking status and lung cancer development \citep{mok2011personalized}. Unraveling these mechanisms could provide invaluable insights for personalized prevention, diagnosis, and treatment strategies.
DNA methylation, an epigenetic modification involving the addition of a methyl group to the DNA molecule, plays a crucial role in gene expression regulation and reflects various biological functions. Consequently, methylation markers are often considered potential mediators between exposures and health outcomes \citep[e.g.,][]{tobi2018dna}. 
In this section, our objective is to investigate how smoking exposure alters DNA methylation patterns, subsequently impacting an individual's risk of developing lung cancer. 

We utilize the data from The Cancer Genome Atlas (TCGA) project which provides a comprehensive characterization of multidimensional genomic alterations, including clinical and molecular data, across various cancer types \citep{tomczak2015review}. Specifically, we focus on the lung adenocarcinoma (LUAD) dataset from TCGA project using the \verb|R| package \verb|cgdsr|. 
To characterize the lung cancer development, we consider the forced expiratory volume in 1 second (FEV1), a commonly used measurement for assessing lung function in lung cancer studies \citep[e.g.,][]{lange2015lung}, as the outcome. 
The smoking status is encoded as a binary exposure $A$, where $A=1$ if the subject currently smokes or used to smoke. {For each gene, the DNA methylation level measured by the Illumina Infinium HumanMethylation450 BeadChip is considered as a mediator}. Instead of including all genes, it is more intuitive to investigate a set related genes each time. For this, we utilize the Gene Ontology (GO) terms that classify genes to groups according to the gene functions. Then, the  DNA methylation levels of the genes in a GO term form a high-dimensional vector of mediators, and we apply the proposed test to each of the $7743$ GO terms from the Molecular Signatures Database \citep{ liberzon2015molecular}, where the largest GO term contains $1781$ genes. In our analysis we also include age, sex, and ethnicity as additional covariates, according to Remark \ref{rem:ext} and Appendix \ref{sec:incorp-covariates}.
After removing missing values, the number of subjects (the sample size $n$) in each GO term varies between $175$ and $211$.

Before conducting the tests, we  examine the structures of $\beta_{A}$ and $\theta_{M}$. 
For $\beta_{A}$, we have its least-squares estimator $\hat{\beta}_{A}=(\hat{\beta}_{A,j}, j=1,\ldots, p) \in \real^{1\times p}$. To gain an insight into the sparsity of $\beta_A$, we compute the following surrogate
\[
\hat{s}_{A} = \sum_{j=1}^{p} 
\left(
\indicator{\hat{\beta}_{A,j}-2\hat{\sigma}_{A,j} \ge 0}
+
\indicator{\hat{\beta}_{A,j}+2\hat{\sigma}_{A,j} \le 0}
\right),
\]
where $\hat{\sigma}_{A,j}^2$ is the usual estimation of $\var(\hat{\beta}_{A,j})$ in the linear regression model by regressing the $j$th mediator on the exposure and the {covariates}. Roughly speaking, $(\hat\beta_{A,j}-2\hat\sigma_{A,j},\hat\beta_{A,j}+2\hat\sigma_{A,j})$ forms an approximate confidence interval for $\beta_{A,j}$, and we may consider $\beta_{A,j}$ to be nonzero if the interval does not contain 0. Therefore, $\hat s_{A}$ provides a rough estimate of the number of nonzero loadings of the vector $\beta_{A}$.
For $\theta_{M}$, since we have its Lasso estimator $\hat \theta_M$, it is intuitive to estimate its sparsity by  the number of nonzero loadings of $\hat \theta_M$, denoted by $\hat{s}_{M}$. 
Out of the 7743 datasets, we observe that, 
\begin{itemize}
	\item 1338 datasets yield $\hat{s}_{A}=\hat{s}_{M}=0$, and for these datasets  there is no strong evidence to rule out  the case of  $\theta_{M}=0$ and/or $\beta_{A}=0$;
	\item 1059 datasets yield $\hat{s}_{A}\ge \sqrt{n}$, suggesting that a significant proportion of the datasets likely fall into the case of dense $\beta_{A}$.
\end{itemize} 
Consequently, it is crucial to adopt a test that remains valid in the above two cases, and the proposed test is the only option in the presence of high-dimensional mediators, to the best of our knowledge.

We apply the test procedure proposed in Section \ref{sec:meth-tesing}, following the same implementation steps as in Section \ref{sec:simu}, to examine whether a candidate gene set associated with specific biological functions  mediates the effect of smoking on lung function via the DNA methylation levels.
By applying the proposed test to each of the aforementioned GO terms, we identify $169$ gene sets with significant mediation effects at the significance level of 5\%, after a Bonferroni correction for multiple tests. Since many GO terms contain overlapping genes, it is not surprising to see this number of significant gene sets.
Table \ref{tab:realdata} lists the top ten most significant gene sets. The majority of these gene sets are associated with immune response or cellular processes, suggesting that smoking may affect lung function through dysregulation of immune responses or disruption of cellular processes. These findings align with previous studies \citep{stampfli2009cigarette}, demonstrating the power of the proposed test for detecting  mediation effect in the challenging setting of high-dimensional mediators.

\begin{sidewaystable}
	\centering
	\caption{Top ten most significant sets of genes whose DNA methylation levels may mediate the effect of smoking on lung function, with adjusted p-values after Bonferroni correction.}
		\begin{tabular}{clccc}
			\cmidrule{1-5}
			GO ID & \multicolumn{1}{c}{Gene Functions} & $n$   & $p$   & P-value \\
			\cmidrule{1-5}
			GO:0050790 & Regulation of catalytic activity & 185   & 1617  & $1.26\times 10^{-8}$ \\
			GO:0031334 & Positive regulation of protein-containing complex assembly & 202   & 184   & $4.01\times 10^{-7}$ \\
			GO:0002684 & Positive regulation of immune system process & 182   & 897   & $5.21\times 10^{-7}$ \\
			GO:0031098 & Stress-activated protein kinase signaling cascade & 201   & 227   & $1.50\times 10^{-6}$ \\
			GO:0043043 & Peptide biosynthetic process & 192   & 639   & $2.58\times 10^{-6}$ \\
			GO:0071692 & Protein localization to extracellular region & 200   & 330   & $2.96\times 10^{-6}$ \\
			GO:2000112 & Regulation of cellular macromolecule biosynthetic process & 196   & 443   & $3.36\times 10^{-6}$ \\
			GO:0044419 & Biological process involved in interspecies interaction between organisms & 180   & 1370  & $3.92\times 10^{-6}$ \\
			GO:0045727 & Positive regulation of translation & 201   & 119   & $5.19\times 10^{-6}$ \\
			GO:0097305 & Response to alcohol & 198   & 222   & $5.24\times 10^{-6}$ \\
			\cmidrule{1-5}
		\end{tabular}%
	\label{tab:realdata}%
\end{sidewaystable}%

\section{Concluding Remarks}\label{sec:concluding}

We developed a test for mediation effect within the LSEM framework \eqref{eq:mediator-exposure} and \eqref{eq:outcome-mediator} when the mediator is of high dimensions. As a major advantage of our test, it remains valid in all cases of the composite null hypothesis. Another advantage of the proposed test is that it allows arbitrary exposure--mediator coefficients. 
In certain applications, sparsity of the mediation effect may be assumed and variable selection among high-dimensional mediators might be of interest. In such cases, our method can be modified accordingly to exploit the additional sparsity structure and to identify the significant mediators, as demonstrated in Appendix \ref{sec:ex-sign-con}.
Although we focus on continuous outcomes, the test may be extended to deal with binary or count outcomes via a high-dimensional generalized linear model. Such extension is beyond the scope of the paper and thus left for future studies.

\backmatter

%
%
%

\bmhead{Acknowledgements}

Zhenhua Lin's research is partially supported by a Singapore MOE Tier 1 grant (A-0008522-00-00) and a NUS startup grant (A-0004816-00-00). Baoluo Sun's work is funded by Singapore MOE Tier 1 grant (A-8000452-00-00). The research of Zijian Guo was partly supported by the NSF grant DMS 2015373 and NIH grants R01GM140463 and R01LM013614.

%
%
%

%
%
%
%

\begin{appendices}

%

\section{Causal Interpretation of Natural Direct and Indirect Effects}\label{sec:causal-interpret}

Let $A$ be a scalar exposure for an individual, $Y$ be an outcome, and $M=(M_1, \ldots, M_p)^{\transpose}$ be $p$ potential mediators that may be on the pathway from exposure to outcome. Let $Y_{a}$ denote the counterfactual value of outcome $Y$ if $A$ were set to the value $a$, $M_{a}$ the counterfactual value of mediators $M$ if $A$ were set to the value $a$, and $Y_{am}$ the counterfactual value of outcome $Y$ if $A$ were set to the value $a$ and $M$ were set to $m$. The controlled direct effect of exposure $A$ on outcome $Y$  for an individual comparing $A=a$ with $A=a^*$, that arises upon intervening to fix the mediators $M$ to some value $m$, is given by $Y_{am}-Y_{a^*m}$. 
Meanwhile, the natural direct effect  of exposure $A$ on outcome $Y$ comparing $A=a$ with $A=a^*$, when intervening to fix the mediators $M$ to their random values if exposure had been $A=a^*$, is given by $Y_{aM_{a^*}}-Y_{a^*M_{a^*}}$. The natural indirect effect for an individual is $Y_{aM_{a}}-Y_{aM_{a^*}}$.
Finally, the exposure total effect, given by $Y_{a}-Y_{a^*}$, can be decomposed into the natural direct and indirect effects. For ease of exposition, we consider the following assumptions in the absence of measured baseline covariates \cite{vanderweele2014mediation}:\\
\begin{assumption}\label{as:no-un-conf-1}
	$Y_{am}\perp A$
\end{assumption}
\begin{assumption}\label{as:no-un-conf-2}
	$Y_{am}\perp M\mid A$
\end{assumption}
\begin{assumption}\label{as:no-un-conf-3}
	$M_{a}\perp A$
\end{assumption}
\begin{assumption}\label{as:no-un-conf-4}
	$Y_{am}\perp M_{a^*}$\\
\end{assumption}
Assumptions \ref{as:no-un-conf-1} and \ref{as:no-un-conf-2} suffice for the identification of the population average controlled direct effect \cite{robins1986new, pearl2001direct}. Furthermore, under Assumptions \ref{as:no-un-conf-1}--\ref{as:no-un-conf-4}, the population average natural direct and indirect effects are nonparametrically identified \cite{pearl2001direct}. In particular, under Assumptions \ref{as:no-un-conf-1}--\ref{as:no-un-conf-4} and the linear structural equation models \eqref{eq:mediator-exposure} and \eqref{eq:outcome-mediator}, the controlled direct effect, and natural direct and indirect effects are given by 
\begin{align*}
	\expect[Y_{am}-Y_{a^*m}] &= \theta_{A}(a-a^*), \\
	\expect[Y_{aM_{a^*}}-Y_{a^*M_{a^*}}] &= \theta_{A}(a-a^*), \\
	\expect[Y_{aM_{a}}-Y_{aM_{a^*}}] &= \beta_{A}^{\transpose}\theta_{M}(a-a^*) = \gamma (a-a^*),
\end{align*}
respectively, as derived in Section 3.2 of \cite{vanderweele2014mediation}. Hence, $\gamma$ can be interpreted as the average natural indirect effect on the outcome $Y$ per unit change in exposure $A$.

\section{Incorporating Covariates}\label{sec:incorp-covariates}

To extend the proposed test in Section \ref{sec:meth-tesing} to incorporate additional measured {covariates}, consider the following extended model, 
\begin{align*}
	M &= A\beta_{A}^{\transpose} + C \beta_{C}^{\transpose} + E, 
	\\
	Y &= A \theta_{A} + C \theta_{C} + M\theta_{M} + Z,
\end{align*}
where $C\in \real^{n\times r}$ represents a matrix of {covariates}, and $\beta_{C}$ and $\theta_{C}$ are the corresponding coefficients.
Taking $X=(A, C, M)\in \real^{n\times (q+r+p)}$ and $\theta = (\theta_A^{\transpose}, \theta_C^{\transpose}, \theta_M^{\transpose})^{\transpose} \in \real^{q+r+p}$, we can still rewrite the second equation by
\begin{equation*}
	Y = X \theta + Z,
\end{equation*}
in analogy to \eqref{eq:Y-X-equation}.

Letting $\beta=(\beta_{A}, \beta_{C})\in \real^{p\times (q+r)}$, we observe that the mediation effect $\gamma=\beta_{A}^{\transpose}\theta_{M}$ corresponds to the first $q$ components of $\beta^{\transpose}\theta_{M}$. 
With $H=(A, C)\in \real^{n\times (q+r)}$, the ordinary least-squares estimator  $\hat{\beta}=((H^{\transpose}H)^{-1}H^{\transpose} M)^{\transpose}$ of $\beta$ and an initial estimator $\hat{\theta} = (\hat{\theta}_A^{\transpose}, \hat{\theta}_C^{\transpose}, \hat{\theta}_M^{\transpose})^{\transpose}$ of $\theta$ form a pilot estimator  $\tilde{\gamma}=\hat{\beta}_A^{\transpose}\hat{\theta}_M$ for $\gamma$, where $\hat{\beta}_{A}$ represents the first $q$ columns of $\hat{\beta}$. Modifying the bias correction procedure outlined in Section \ref{sec:estimation} in a straightforward way, we can obtain a debiased estimator $\hat{\gamma}$ for $\gamma$ when {covariates} are involved. The debiased estimator is given by
\begin{equation*}
	\hat{\gamma} = \tilde{\gamma} + n^{-1} \hat{U}^{\transpose} X^{\transpose} (Y-X\hat{\theta})
	\in \real^{q},
\end{equation*}
where $\hat{U}$ is now the matrix consisting of the projection directions obtained via optimizing \eqref{VePD-constraint-bias}--\eqref{VePD-constraint-lindeberg} with
\[
g_j = (\zerovec{q+r}^{\transpose}, \tilde{g}_j^{\transpose}) \in \real^{q+r+p},
\]
where $\tilde{g}_j$ is the $j$th column of $\hat{\beta}_{A}$ for each $j=1,\ldots, q$.
With an appropriate initial estimator $\hat{\theta}$, the covariance matrix of $\hat{\gamma}$, conditional on $\{X_i, i=1,\ldots, n\}$, is estimated by
\begin{equation*}
	\hat{V} = \frac{\hat\sigma_E^2}{n} \hat{\Sigma}_{H, AA}^{-1} + \frac{\hat\sigma_Z^2}{n} \hat{U}^{\transpose} \hat{\Sigma}_X \hat{U} + \frac{\tau}{n} I_q,
\end{equation*}
where $\hat{\Sigma}_{H, AA}^{-1}$ is the first $q$ leading principal minor  of $\hat{\Sigma}_{H}^{-1}$, corresponding to the exposures $A$ in $\hat{\Sigma}_H^{-1}$. 

\section{Proofs for Proposition \ref{prop:CC-general-require} and Theorems \ref{thm:asym-norm-hgamma-decomp} and \ref{thm:test-bonf}}\label{sec:proofs-prop-and-thm}

In the sequel, we use $c,C_1,C_2,\ldots$ to denote positive constants not depending on $n$ and $p$. In addition, we allow the value of $c$ to vary from place to place. The notation $\overset{p}{\rightarrow}$ denotes convergence in probability.

\begin{proof}[Proof of Proposition~\ref{prop:CC-general-require}]
	
	The claim for the compatibility condition in the proposition is the direct result of Lemma 6.17 of \cite{buhlmann2011statistics}. Specifically, based on the arguments on Page 152 and Problem 14.3 of \cite{buhlmann2011statistics}, with Assumption (A1) and the compatibility condition for $\Sigma_X$, the design matrix $X$  also satisfies the compatibility condition, with probability at least $1-C(q+p)^{-1}$ for some constant $C>0$. 
	
	Assume the design matrix $X$ satisfies the compatibility condition in the sequel. 	
	For the Lasso estimator $\hat{\theta}$, we prove the result by checking the conditions of Corollary 4.5 in \cite{fan2020statistical}. Let $\mathcal{S} = \{j\in \{1, \ldots, q+p\}: |\theta_j|>\sigma_Z \lambda_0\}$. By \eqref{eq:capped-l1}, we have $\|\theta_{\mathcal{S}^c}\|_1 \le \lambda_0 s$ and $|\mathcal{S}| \le s$. 
	Let $R_{ij}=X_{ij} Z_i$, which is centered, for each $i=1, \ldots, n$ and $j=1, \ldots, q+p$. Since both $X_{ij}$ and $Z_i$ are subGaussian, $R_{ij}$ is sub-exponential. Let $\|\cdot\|_{\psi_p}$ be the $\psi_p$-Orlicz norm. By Bernstein's inequality (see the proof of Theorem 2.8.1 in \cite{vershynin2018high}), for any $t>0$, 
	\begin{align}
		\prob\left( \|n^{-1} X^{\transpose} Z\|_{\infty} > t \right) 
		& =
		\prob\left( \max_{j=1, \ldots, q+p} \left|n^{-1} \sum_{i=1}^{n}R_{ij}\right| > t \right) \nonumber \\
		\notag
		& \le 
		\sum_{j=1}^{q+p}2\exp\left[ -\frac{1}{16e^2} \min\left(\frac{t^2}{K_j^2}, \frac{t}{K_j}\right) n \right], \\
		&\le
		2\exp\left[ -\frac{1}{16e^2} \min\left(\frac{t^2}{K_0^2}, \frac{t}{K_0}\right) n +\log(q+p)\right],
		\label{eq:lasso-rescor-concentration}
	\end{align}
	where $K_j=\max_{i}\|R_{ij}\|_{\psi_1}$ and  $K_0=\max_{j}K_j$. 
	Choose $t=\eta \lambda_n$ with $\lambda_n= \frac{\sqrt{1+\eta}}{\eta} 4e K_0 \sqrt{\log(q+p)/n}$ for some constant $\eta \in (0, 1]$. {Since $\log(q+p)\le n/(32e^2)$} for sufficiently large $n$, we deduce that  
	\[
	t = \sqrt{1+\eta} 4 e K_0 \sqrt{\frac{\log(q+p)}{n}} \le K_0.
	\] 
	Combining this choice of $t=\eta \lambda_n$ with \eqref{eq:lasso-rescor-concentration}, we obtain
	\[
	\prob\left( \|n^{-1} X^{\transpose} Z\|_{\infty} > \eta \lambda_n \right)
	\le
	2\exp\left( - \frac{1}{16e^2}\eta^2 \lambda_n^2 n/K_0^2 + \log(q+p) \right)
	\le
	2(q+p)^{-\eta}.
	\]
	Consequently, with the compatibility condition for $X$, by Corollary 4.5 and Equation (4.91) in \cite{fan2020statistical}, we conclude
	\[
	\|\hat{\theta} - \theta\|_1 \lesssim s \lambda_n \asymp s \sqrt{\frac{\log(q+p)}{n}},
	\]
	with probability approaching one as $n \to \infty$.

	For the scaled Lasso estimators $\hat{\theta}$ and $\hat{\sigma}_{Z}$, (B1) and (B2) have been proved by Theorem 4 in \cite{zhang2014confidence} by taking $\epsilon=1/p$ therein.
\end{proof}

\begin{proof}[Proof of Thoerem \ref{thm:asym-norm-hgamma-decomp}]
	Based on \eqref{eq:bias-of-init}, \eqref{eq:gj-def} and \eqref{eq:hgamma-estimator}, we deduce that 
	\begin{align}
		\notag
		(\hat{\gamma} - \gamma) &= \tilde{\gamma} + n^{-1} \hat{U}^{\transpose} X^{\transpose} (Y-X\hat{\theta}) - \gamma \\
		\notag
		&= n^{-1} \hat{\Sigma}_{A}^{-1} A^{\transpose} E_{M} + n^{-1} \hat{U}^{\transpose} X^{\transpose} Z + (\hat{\Sigma}_{X} \hat{U} - G)^{\transpose} (\theta - \hat{\theta})\\
		&= W+B,
	\end{align}
	where $G=(g_1, \ldots, g_q)\in \real^{(q+p)\times q}$ with $g_j$ defined in Section \ref{sec:estimation}, $W = n^{-1} \hat{\Sigma}_{A}^{-1} A^{\transpose} E_{M} + n^{-1} \hat{U}^{\transpose} X^{\transpose} Z$ and $B = (\hat{\Sigma}_{X} \hat{U} - G)^{\transpose} (\theta - \hat{\theta})$. 
	
	We first establish  the bound for $\|DB\|_{\infty}$. With probability tending to one, 
	\begin{align}
		&\notag
		\|DB\|_{\infty} \\
		\notag &= \|D (\hat{\Sigma}_{X} \hat{U} - G)^{\transpose} (\theta - \hat{\theta})\|_{\infty} 
		\le \| (\hat{\Sigma}_{X} \hat{U} - G)D \|_{\infty} \|\theta - \hat{\theta}\|_{\infty} \\
		\notag
		&\le \sum_{j=1}^q \|(\hat{\Sigma}_{X} \hat{u}_j - g_j) / \sqrt{V_{jj}}\|_{\infty} \|\theta - \hat{\theta}\|_1 
		\le C_1  \sum_{j=1}^{q} \frac{\|g_j\|_2 \lambda}{C_2\|g_j\|_2/\sqrt{n}} s \sqrt{\log(q+p)/n} \\
		\notag
		&\le C_3 \sqrt{n} \lambda s q \sqrt{\log(q+p)/n} \asymp s q \log(q+p) / \sqrt{n},
	\end{align} 
	where the the second inequality is due to the basic inequality $\max_i\sum_j |k_{ij}| \le \sum_j \max_i |k_{ij}|$ for any matrix $K=(k_{ij})$, and the third inequality is due to the constraint \eqref{VePD-constraint-bias},  Lemma 1 of \cite{cai2021optimal} and Assumption (B1). 
	
	Next, we show
	\begin{equation}
		W_{j} / \sqrt{V_{jj}} \stackrel{d}{\to} N(0, 1), \quad~\text{for each}~j=1,\ldots, q,
		\label{eq:Wj-asymp-normal}
	\end{equation}
	where $W_{j}$ is the $j$th element in $W$ and 
	\[
	V_{jj}=
	\frac{\sigma_E^2}{n}e_{j}^{\transpose}\hat{\Sigma}_{A}^{-1}e_{j}
	+
	\frac{\sigma_Z^2}{n}\hat{u}_{j}^{\transpose}\hat{\Sigma}_{X}\hat{u}_{j}
	+
	\frac{\tau}{n}
	\] 
	is the $j$th diagonal element in $V$ defined in \eqref{eq:V-def}, where $e_j\in \real^q$ represents the canonical basis vector with $1$ in the $j$th coordinate.
	
	Let $A_i$ and $X_i$ be respectively the $i$th rows of $A$ and $X$, and write
	\[
	W^i = n^{-1}\hat{\Sigma}_{A}^{-1} A_i E_{M,i} + n^{-1}\hat{U}^{\transpose} X_i Z_{i}.
	\]
	Then $W = \sum_{i=1}^{n} W^i$. 
	Since $\hat{U}$ depends only on $X$, conditioning on $X$, we have
	\[
	\expect[W^i| X] = \expect(n^{-1}\hat{\Sigma}_{A}^{-1} A_i E_{M,i}|A)
	+
	\expect(n^{-1}\hat{U}^{\transpose} X_i Z_{i}|X)
	= 0
	\]
	and 
	\begin{align*}
		\var[n^{-1}\hat{\Sigma}_{A}^{-1} A_i E_{M,i}|A] &= n^{-2}\sigma_E^2 \hat{\Sigma}_A^{-1} A_i A_i^{\transpose} \hat{\Sigma}_A^{-1}, \\
		\var[n^{-1}\hat{U}^{\transpose} X_i Z_{i}|X] &= n^{-2}\sigma_Z^2 \hat{U}^{\transpose} X_i X_i^{\transpose} \hat{U}, \\
		\cov[n^{-1}\hat{\Sigma}_{A}^{-1} A_i E_{M,i}, n^{-1}\hat{U}^{\transpose} X_i Z_{i}|X] &= n^{-2}\expect[\hat{\Sigma}_{A}^{-1} A_i E_{M,i} Z_{i} X_i^{\transpose} \hat{U}|X] = 0,
	\end{align*}
	which imply $\expect[W|X]=0$ and $\var[W|X]=\sum_{i=1}^{n} \expect[W^i (W^i)^{\transpose}|X] = V$.
	
	Define $\tilde{W}_{j}^{i}=W_{j}^{i}/\sqrt{V_{jj}}$ for each $i=1,\ldots, n$ and $j=1,\ldots,q$. Then, conditioning on $X$, for each $j=1,\ldots, q$, $\{\tilde{W}_{j}^{i}: i=1,\ldots, n\}$ are independent (but not identically distributed) random variables with $\expect[\tilde{W}_{j}^{i}|X]=0$ and $\sum_{i=1}^{n} \var[\tilde{W}_{j}^{i}|X] = 1$. 
	To establish \eqref{eq:Wj-asymp-normal}, we first check the Lindeberg's condition, that is, for each $j=1,\ldots, q$ and any constant $c>0$, 
	\begin{equation}
		\lim_{n\to \infty} \sum_{i=1}^{n}
		\expect\left[ (\tilde{W}_{j}^{i})^2 \indicator{|\tilde{W}_{j}^{i}|\ge c} | X\right] = 0.
	\end{equation}
	Define the events
	\begin{align}
		\notag
		\mathcal{A}_1 &= \left\{\|\hat{\Sigma}_{A}^{-1}-\Sigma_{A}^{-1} \|_2\lesssim \left(\frac{q+\log n}{n}\right)^{1/2}\right\}, \\
		\notag
		\mathcal{A}_2 &= \left\{ {\|g_j\|_2^2/n \leq c V_{jj}} \text{ for all }j=1,\ldots,q \right\}, \\
		\notag
		\mathcal{A}_3 &= \left\{ \max_{i=1,\ldots, n} | \|A_i\|_2 - \sigma_{A}\sqrt{q} | \lesssim n^{1/4} \right\}, \\
		\notag
		\mathcal{A} &= \mathcal{A}_1 \cap \mathcal{A}_2 \cap \mathcal{A}_3,
	\end{align}
	where  we assume $\sigma_{A}^2=\expect[A_{ij}^2]$ for all $j=1,\dots, q$ without loss of generality, and $c>0$ is a sufficiently large constant. 
	Then, there is a constant $C_1>\zeta$, such that, $\prob\left(\mathcal{A}_1\right)\ge 1 - n^{-1}$ from Equation (1.6) in \cite{kereta2021estimating},  
	{$\prob\left(\mathcal{A}_2\right)\ge 1 - qp^{-C_1}\geq 1-p^{-(C_1-\zeta)}$} based on Lemma 1 in \cite{cai2021optimal}, and $\prob\left(\mathcal{A}_3\right)\ge 1 - e^{-C_1 n^{1/2}/q+\log(nq)}$ from Lemma \ref{lemma:A-norm-concen}.
	Hence, when the event $\mathcal{A}$ holds, conditioning on $X$, for each $j=1,\ldots, q$ and any constant $\uconst>0$,
	\begin{align}
		& \notag
		\sum_{i=1}^{n}\expect\left[ (\tilde{W}_{j}^{i})^2 \indicator{|\tilde{W}_{j}^{i}|\ge \uconst} | X\right] \\
		& \notag\le 
		\frac{2}{n^2} \sum_{i=1}^{n} \expect\left[ \frac{(e_{j}^{\transpose}\hat{\Sigma}_{A}^{-1} A_i E_{M,i})^2}{V_{jj}} \indicator{|\tilde{W}_{j}^{i}|\ge \uconst} | X \right]	\\
		\notag
		&\quad +
		\frac{2}{n^2} \sum_{i=1}^{n} \expect\left[ \frac{(\hat{u}_{j}^{\transpose} X_i Z_{i})^2}{V_{jj}} \indicator{|\tilde{W}_{j}^{i}|\ge \uconst} | X \right]	\\
		\notag
		&= 
		\frac{2}{n^2} \sum_{i=1}^{n} \frac{(e_{j}^{\transpose}\hat{\Sigma}_{A}^{-1} A_i )^2}{V_{jj}} \expect\left[ E_{M,i}^2 \indicator{|\tilde{W}_{j}^{i}|\ge \uconst} | X \right]	\\
		\notag
		&\quad +
		\frac{2}{n^2} \sum_{i=1}^{n} \frac{(\hat{u}_{j}^{\transpose} X_i)^2}{V_{jj}} \expect\left[ Z_{i}^2 \indicator{|\tilde{W}_{j}^{i}|\ge \uconst} | X \right] \\
		\notag
		&\le 
		\frac{2}{n} \frac{e_{j}^{\transpose}\hat{\Sigma}_{A}^{-1} e_{j} }{V_{jj}}
		\max_{1\le i\le n} \expect\left[ E_{M,i}^2 \indicator{|\tilde{W}_{j}^{i}|\ge \uconst} | X \right] 
		+
		2 \max_{1\le i\le n}\expect\left[ \frac{Z_{i}^2}{\sigma_Z^2} \indicator{|\tilde{W}_{j}^{i}|\ge \uconst} | X \right] \\
		\notag
		&\le
		\frac{2}{nV_{jj}} \left(\|\Sigma_{A}^{-1} \|_2^2+\|\hat{\Sigma}_{A}^{-1}-\Sigma_{A}^{-1} \|_2^2\right) 
		\max_{1\le i\le n} \expect\left[ E_{M,i}^2 \indicator{|\tilde{W}_{j}^{i}|\ge \uconst} | X \right]	\\
		\notag
		&\quad +
		2 \max_{1\le i\le n}\expect\left[ \frac{Z_{i}^2}{\sigma_Z^2} \indicator{|\tilde{W}_{j}^{i}|\ge \uconst} | X \right] \\
		\notag
		&\le 
		\frac{C_2}{\tau} \left(\Lambda_{\min}^{-2}(\Sigma_{A})+\frac{q+\log n}{n}\right)
		\max_{1\le i\le n} \expect\left[ E_{M,i}^2 \indicator{|\tilde{W}_{j}^{i}|\ge \uconst} | X \right]	\\
		\notag
		&\quad +
		2 \max_{1\le i\le n}\expect\left[ \frac{Z_{i}^2}{\sigma_Z^2} \indicator{|\tilde{W}_{j}^{i}|\ge \uconst} | X \right] \\
		&\le
		\frac{C_3}{\tau} \max_{1\le i\le n} \expect\left[ E_{M,i}^2 \indicator{|\tilde{W}_{j}^{i}|\ge \uconst} | X \right]
		+
		2 \max_{1\le i\le n}	\expect\left[ \frac{Z_{i}^2}{\sigma_Z^2} \indicator{|\tilde{W}_{j}^{i}|\ge \uconst} | X \right],
		\label{eq:lindeberg-cond-init}
	\end{align}
	where the second inequality is due to the definition of $V_{jj}$, and the fourth inequality is due to the event $\mathcal{A}_1$. 
	
	When $\mathcal A$ holds, for each $i=1,\ldots,n$ and $j=1,\ldots, q$, 
	\begin{align}
		& \notag
		|\tilde{W}_{j}^{i}|  \\
		& \notag \le \frac{1}{n} \frac{|e_{j}^{\transpose}\hat{\Sigma}_{A}^{-1} A_i|}{\sqrt{V_{jj}}} |E_{M,i}| + \frac{1}{n} \frac{|\hat{u}_{j}^{\transpose} X_i |}{\sqrt{V_{jj}}} |Z_{i}| \\
		\notag
		&\le 
		\frac{\|A_i\|_2}{\sqrt{n}} 
		\left(\|\Sigma_{A}^{-1} \|_2+\|\hat{\Sigma}_{A}^{-1}-\Sigma_{A}^{-1} \|_2\right) |E_{M,i}|
		+
		\frac{\|g_{j}\|_2 \mu }{n\sqrt{V_{jj}}} |Z_{i}| \\
		&\le
		C_3 \frac{n^{1/4} + \sigma_{A}\sqrt{q}}{\sqrt{n}} |E_{M,i}|
		+
		\frac{\|g_{j}\|_2 \mu }{n\sqrt{V_{jj}}} |Z_{i}|
		\le
		C_4 \left(\frac{n^{1/4} + \sqrt{q}}{\sqrt{n}} |E_{M,i}|
		+
		\frac{\mu}{\sqrt{n}}|Z_{i}| \right),
		\label{eq:bound-abs-Wji-tilde}
	\end{align}
	where the second inequality is due to the constraint \eqref{VePD-constraint-lindeberg}, the third inequality is due to the events $\mathcal{A}_1$ and $\mathcal{A}_3$, and the last inequality is from the event $\mathcal{A}_2$. 
	Combining \eqref{eq:lindeberg-cond-init} with \eqref{eq:bound-abs-Wji-tilde} leads to
	\begin{align}
		\notag
		\sum_{i=1}^{n}\expect\left[ (\tilde{W}_{j}^{i})^2 \indicator{|\tilde{W}_{j}^{i}|\ge \uconst} | X, {\mathcal{A}}\right]
		& {\leq}
		\frac{C_5}{\tau} \max_{1\le i \le n} \expect\left[ E_{M,i}^2 \indicator{ |E_{M,i}| \ge \frac{\uconst}{2} \frac{\sqrt{n}}{n^{1/4} + \sqrt{q}} } | X, {\mathcal{A}} \right] \nonumber  \\
		&
		\quad+ 
		C_5\max_{1\le i \le n} \expect\left[ E_{M,i}^2 \indicator{ |Z_{i}| \ge \frac{\uconst}{2} \frac{\sqrt{n}}{\mu} } | X, {\mathcal{A}} \right] \\
		\notag
		&\quad +
		\frac{C_5}{\tau} \max_{1\le i \le n} \expect\left[ \frac{Z_{i}^2}{\sigma_Z^2} \indicator{ |E_{M,i}| \ge \frac{\uconst}{2} \frac{\sqrt{n}}{n^{1/4} + \sqrt{q}} } | X, {\mathcal{A}} \right] \nonumber\\
		& \quad + 
		C_5\max_{1\le i \le n} \expect\left[ \frac{Z_{i}^2}{\sigma_Z^2} \indicator{ |Z_{i}| \ge \frac{\uconst}{2} \frac{\sqrt{n}}{\mu} } | X, {\mathcal{A}} \right] \\
		\notag
		&{\leq}
		\frac{C_6}{\tau} \left( \frac{\sqrt{n}}{n^{1/4} + \sqrt{q}} \right)^{-\nu} 
		+
		C_6\left( \frac{\sqrt{n}}{\mu} \right)^{-\nu}  \\
		&\quad +
		\frac{C_6}{\tau} \left( \frac{\sqrt{n}}{n^{1/4} + \sqrt{q}} \right)^{-\nu} 
		+
		C_6\left( \frac{\sqrt{n}}{\mu} \right)^{-\nu} \\
		\notag
		&\to 0
		\qquad (\text{as }n \to \infty),
	\end{align}
	where the last inequality is due to Assumption (A2), {$q\ll\sqrt n$} and $\mu\asymp \log n$. Therefore, conditioning on $X$ and $\mathcal{A}$, by Lindeberg's central limit theorem, for each $j=1,\ldots, q$ and any $t\in \real$,
	\[
	\lim_{n\to \infty}\prob\left( W_{j}/\sqrt{V_{jj}} \le t | X, \mathcal{A} \right) = \Phi(t),
	\]
	where $\Phi(t)$ denotes the cumulative distribution function of the standard normal distribution. That is, $W_{j}/\sqrt{V_{jj}}|X, \mathcal{A} \stackrel{d}{\to} N(0, 1)$. {With the event $\mathcal{A}$, by calculating the characteristic function and applying the bounded convergence theorem,} we have 
	\[
	W_{j}/\sqrt{V_{jj}} \mid \mathcal{A} \stackrel{d}{\to} N(0, 1).
	\]
	Consequently,
	\begin{align*}
		& \lim_{n\to \infty}\prob\left(W_{j}/\sqrt{V_{jj}} \le t \right) \\
		&\le 
		\lim_{n\to \infty}\prob\left(W_{j}/\sqrt{V_{jj}} \le t | \mathcal{A}\right)
		+
		\lim_{n\to \infty}\prob\left(\mathcal{A}^c\right) \\
		&= 
		\Phi(t)
		+ \lim_{n\to \infty}\prob\left(\mathcal{A}_1^c\right)
		+ \lim_{n\to \infty}\prob\left(\mathcal{A}_2^c\right)
		+ \lim_{n\to \infty}\prob\left(\mathcal{A}_3^c\right)
		= 
		\Phi(t),
	\end{align*}
	where $\mathcal{A}^c$ denotes the complement of the event $\mathcal{A}$. This completes the proof.
	
\end{proof}

\begin{proof}[Proof of Theorem \ref{thm:test-bonf}]
	Recall $T=\left(\hat{\gamma}_j / \sqrt{\hat{V}_{jj}}\right)_{j=1}^q$ and $D=\text{diag}\left(1/\sqrt{V_{jj}}\right)$. Define $\hat{D}=\text{diag}\left(1/\sqrt{\hat{V}_{jj}}\right)$. Then $T=\hat{D}\hat{\gamma}$.
	
	We first bound $\left| \sqrt{V_{jj}}/\sqrt{\hat{V}_{jj}} - 1 \right| $ for each $j=1, \ldots, q$.
	Recall that
	\[
	V_{jj} = \frac{\sigma_E^2}{n} e_j^{\transpose}\hat{\Sigma}_A^{-1}e_j + \frac{\sigma_Z^2}{n} \hat{u}_j^{\transpose} \hat{\Sigma}_X \hat{u}_j + \frac{\tau}{n},
	\]
	and
	\[
	\hat{V}_{jj} = \frac{\hat{\sigma}_E^2}{n} e_j^{\transpose}\hat{\Sigma}_A^{-1}e_j + \frac{\hat{\sigma}_Z^2}{n} \hat{u}_j^{\transpose} \hat{\Sigma}_X \hat{u}_j + \frac{\tau}{n}.
	\]
	Then, 
	\begin{align*}
		\left|\frac{V_{jj}}{\hat{V}_{jj}} - 1\right|
		&=
		\frac{\left|(\sigma_E^2-\hat{\sigma}_E^2) e_j^{\transpose}\hat{\Sigma}_A^{-1}e_j + (\sigma_Z^2-\hat{\sigma}_Z^2) \hat{u}_j^{\transpose} \hat{\Sigma}_X \hat{u}_j\right|}{\hat{\sigma}_E^2 e_j^{\transpose}\hat{\Sigma}_A^{-1}e_j + \hat{\sigma}_Z^2 \hat{u}_j^{\transpose} \hat{\Sigma}_X \hat{u}_j + \tau} \\
		&\le
		\left| \frac{(\sigma_E^2-\hat{\sigma}_E^2) e_j^{\transpose}\hat{\Sigma}_A^{-1}e_j}{\hat{\sigma}_E^2 e_j^{\transpose}\hat{\Sigma}_A^{-1}e_j + \hat{\sigma}_Z^2 \hat{u}_j^{\transpose} \hat{\Sigma}_X \hat{u}_j + \tau} \right|
		+
		\left| \frac{(\sigma_Z^2-\hat{\sigma}_Z^2) \hat{u}_j^{\transpose} \hat{\Sigma}_X \hat{u}_j}{\hat{\sigma}_E^2 e_j^{\transpose}\hat{\Sigma}_A^{-1}e_j + \hat{\sigma}_Z^2 \hat{u}_j^{\transpose} \hat{\Sigma}_X \hat{u}_j + \tau} \right| \\
		&\le
		\left| \frac{(\sigma_E^2-\hat{\sigma}_E^2) e_j^{\transpose}\hat{\Sigma}_A^{-1}e_j}{\hat{\sigma}_E^2 e_j^{\transpose}\hat{\Sigma}_A^{-1}e_j + \tau} \right|
		+
		\left| \frac{(\sigma_Z^2-\hat{\sigma}_Z^2) \hat{u}_j^{\transpose} \hat{\Sigma}_X \hat{u}_j}{\hat{\sigma}_Z^2 \hat{u}_j^{\transpose} \hat{\Sigma}_X \hat{u}_j} \right| \\
		&=
		\left| \frac{(\sigma_E^2-\hat{\sigma}_E^2) e_j^{\transpose}\hat{\Sigma}_A^{-1}e_j}{\hat{\sigma}_E^2 e_j^{\transpose}\hat{\Sigma}_A^{-1}e_j + \tau} \right|
		+
		\left| \frac{\sigma_Z^2}{\hat{\sigma}_Z^2} - 1 \right|.
	\end{align*}
	Consider the following cases. 
	\begin{itemize}
		\item  $\theta_{M}\ne 0$, which implies $\sigma_E^2\ne 0$. By Assumption (B2), $\left|\sigma_E^2/\hat{\sigma}_E^2 - 1 \right| \stackrel{p}{\to} 0$. 
		Consequently, 
		\[
		\left| \frac{\sqrt{V_{jj}}}{\sqrt{\hat{V}_{jj}}} - 1 \right|  
		= 
		\frac{\left| \frac{V_{jj}}{\hat{V}_{jj}} - 1 \right| }{\left| \frac{\sqrt{V_{jj}}}{\sqrt{\hat{V}_{jj}}} + 1 \right|}
		\le
		\left| \frac{V_{jj}}{\hat{V}_{jj}} - 1 \right| 
		\le 
		\left| \frac{\sigma_E^2}{\hat{\sigma}_E^2} - 1 \right|
		+
		\left| \frac{\sigma_Z^2}{\hat{\sigma}_Z^2} - 1 \right|
		\stackrel{p}{\to} 0.
		\]
		\item  $\theta_{M}=0$, which implies $\sigma_E^2=0$. In this case, let $\sigma^2=\sigma_E^2+\sigma_Z^2=\sigma_Z^2$ and $\hat{\sigma}^2=\hat{\sigma}_E^2+\hat{\sigma}_Z^2$. Then, one has
		\[
		\hat{\sigma}_E^2=\hat{\sigma}^2-\hat{\sigma}_Z^2 \stackrel{p}{\to} \sigma^2 - \sigma_Z^2 = 0.
		\]
		Consequently, 
		\[
		\left| \frac{\sqrt{V_{jj}}}{\sqrt{\hat{V}_{jj}}} - 1 \right|  
		\le
		\left| \frac{V_{jj}}{\hat{V}_{jj}} - 1 \right| 
		\le 
		\left| \frac{\sigma_Z^2}{\hat{\sigma}_Z^2} - 1 \right|
		+
		\left| \frac{(\sigma_E^2-\hat{\sigma}_E^2)e_j^{\transpose} \hat{\Sigma}_{A}^{-1}e_j}{\tau} \right|
		\stackrel{p}{\to} 0,
		\]
		where we use the fact that, for some constant $C>0$, with probability tending to one, $|e_j^{\transpose} \hat{\Sigma}_{A}^{-1}e_j/\tau| \le C$. 
	\end{itemize}
	Combing the above results, we have
	\begin{equation}
		\left| \|D\hat{D}^{-1}\|_{\infty} - 1 \right| 
		= 
		\left| \max_{j} \frac{\sqrt{V_{jj}}}{\sqrt{\hat{V}_{jj}}} - 1 \right|
		\stackrel{p}{\to} 
		0.
		\label{eq:D-Dhat-inv_p-to_zero}
	\end{equation}
	
	Now we start to establish the validity. Under the null hypothesis, $\gamma=0$. In this case, by Theorem \ref{thm:asym-norm-hgamma-decomp}, we have
	\begin{equation}
		\|T\|_{\infty} = \|\hat{D} \hat{\gamma}\|_{\infty} 
		\le
		\|\hat{D}W\|_{\infty}
		+
		\|\hat{D}B\|_{\infty}.
		\label{eq:T_l-inf_upper}
	\end{equation}
	Define $\tilde{W}=DW\in \real^{q}$. 
	For a fixed $\epsilon\in(0,1/2)$, with \eqref{eq:T_l-inf_upper}, we deduce
	\begin{align}
		\notag
		\prob\left( \|T\|_{\infty} \ge x \right) 
		&\le 
		\prob\left( \|\hat{D}W\|_{\infty} + \|\hat{D}B\|_{\infty} \ge x \right) \\
		\notag
		&=
		\prob\left( \|\hat{D}D^{-1}\tilde{W}\|_{\infty} + \|\hat{D}B\|_{\infty} \ge x \right) \\
		\notag
		&\le
		\prob\left( \|\hat{D}D^{-1}\|_{\infty} \|\tilde{W}\|_{\infty} + \|\hat{D}B\|_{\infty} \ge x \right) \\
		\notag
		&\le
		\prob\left( \|\hat{D}D^{-1}\|_{\infty} \|\tilde{W}\|_{\infty} \ge x -\epsilon\right) + \prob\left( \|\hat{D}B\|_{\infty} \ge \epsilon \right) \\
		\notag
		&\le
		\prob\left( \|\tilde{W}\|_{\infty} \ge (1-\epsilon)(x -\epsilon) \right) 
		+ 
		\prob\left( \left|\|D\hat{D}^{-1}\|_{\infty}-1\right| \ge \epsilon \right) \\
		&\quad + 
		\prob\left(\|\hat{D}B\|_{\infty} \ge \epsilon \right).
		\label{eq:dist-fun-T-infty}
	\end{align}
	For the second term of the right-hand side of \eqref{eq:dist-fun-T-infty}, with \eqref{eq:D-Dhat-inv_p-to_zero}, we have 
	\begin{equation}
		\varlimsup_{n\to \infty} \sup_{\xi \in \mathcal{H}_0(s)} \prob\left( \left|\|D\hat{D}^{-1}\|_{\infty}-1\right| \ge \epsilon \right)  = 0.
		\label{eq:variance-mat-consistency}
	\end{equation}
	For the third term of the right-hand side of \eqref{eq:dist-fun-T-infty}, 
	\begin{align}
		\notag
		\prob\left( \|\hat{D}B\|_{\infty} \ge \epsilon \right)
		&\le
		\prob\left( \|\hat{D}D^{-1}\|_{\infty} \|DB\|_{\infty} \ge \epsilon \right) \\
		\notag
		&\le 
		\prob\left(2 \|DB\|_{\infty} \ge \epsilon \right) 
		+ 
		\prob\left(\|\hat{D}D^{-1}\|_{\infty} \ge 2 \right).
	\end{align}
	By \eqref{eq:variance-mat-consistency}, $\prob\left(\|\hat{D}D^{-1}\|_{\infty} \ge 2 \right)\to 0$. In addition, according to  Theorem \ref{thm:asym-norm-hgamma-decomp}, with probability tending to one, 
	\[
	2 \|DB\|_{\infty} \le \frac{2sq \log(q+p)}{\sqrt{n}} \to 0,
	\]
	as $s \ll \sqrt{n}/(q\log(q+p))$. Hence,
	\begin{equation}
		\varlimsup_{n\to \infty} \sup_{\xi \in \mathcal{H}_0(s)} \prob\left( \|\hat{D}B\|_{\infty} \ge \epsilon \right) = 0.
		\label{eq:DB-residual}
	\end{equation}
	
	Combining \eqref{eq:dist-fun-T-infty}, \eqref{eq:variance-mat-consistency} and \eqref{eq:DB-residual} yields
	\begin{equation}
		\varlimsup_{n\to \infty} \sup_{\xi \in \mathcal{H}_0(s)} \prob\left( \|T\|_{\infty} \ge x \right) 
		\le
		\varlimsup_{n\to \infty} \sup_{\xi \in \mathcal{H}_0(s)} \prob\left( \|\tilde{W}\|_{\infty} \ge (1-\epsilon)(x -\epsilon) \right) 
		\label{eq:dist-fun-T-normal-app}
	\end{equation}
	Applying a union bound and the dominated convergence theorem, we have
	\[
	\varlimsup_{n\to \infty} \sup_{\xi \in \mathcal{H}_0(s)} \prob\left( \|T\|_{\infty} \ge x \right) 
	\le
	2q(1-\Phi(x)).
	\]
	The desired asymptotic validity follows by choosing $x=\Phi^{-1}(1-\alpha/(2q))$.
	
	Next, we establish the lower bound for the power to conclude the proof of the theorem. Rewrite $T=\hat{D}\hat{\gamma}=\hat{D}\gamma-\hat{D}(\gamma-\hat{\gamma})=v-\tilde{v}$, where $\tilde{v}=\hat{D}(\hat{\gamma}-\gamma)$ and $v=\hat{D}\gamma$. 
	Define $j^{\star}=\argmax_{j}|v_j|$. Then we have
	\begin{equation}\label{eq:T}
		\|T\|_{\infty} = \|v-\tilde{v}\|_{\infty}\ge |v_{j^{\star}} - \tilde{v}_{j^{\star}}|.
	\end{equation}
	
	Firstly, by a similar argument to derive \eqref{eq:dist-fun-T-normal-app}, one can show that for each $j=1,\ldots, q$ and any $x\in \real$, 
	\begin{equation}
		\varlimsup_{n\to \infty} \sup_{\xi \in \mathcal{H}_1(s, \delta)} \left| \prob(\tilde{v}_j \le x) - \Phi(x) \right| = 0.
		\label{eq:normal-app_v_tilde}
	\end{equation} 
	That is, each coordinate of $\tilde{v}$ asymptotically admits the standard normal distribution.
	
	Secondly, for a local alternative $\xi\in \mathcal{H}_1(s, \delta)$, with probability tending to one, there are constants $C_1,C_2>0$, such that
	\begin{align}
		\notag
		|v_{j^{\star}}| & = \|v\|_{\infty} = \|\hat{D}\gamma\|_{\infty}
		= 
		n^{-1/2}\|\hat{D}D^{-1}D \delta\|_{\infty} \\
		\notag
		&=
		n^{-1/2} \max_{1\le k\le q}\frac{|\delta_{k}|}{\sqrt{V_{kk}}} \frac{\sqrt{V_{kk}}}{\sqrt{\hat{V}_{kk}}} \\
		\notag
		&\ge 
		n^{-1/2}
		\frac{\delta_{\max}}{\max_{1\le k\le q}\sqrt{V_{kk}}} \frac{1}{2} \\
		&\ge 
		\frac{1}{2} \frac{\delta_{\max}}{\max_{1\le k\le q} \|g_k\|_2+C_2} \nonumber \\
		&\ge 
		\frac{C_1 \delta_{\max}}{\sigma\sqrt{\log(pn)/n}+\beta_{A,\max}+C_1},
		\label{eq:lower-v_j-star}
	\end{align}
	where $\delta_{\max}=\max_{1\le k\le q}|\delta_{k}|$, $\beta_{A,\max}=\max_{k=1,\ldots, q}\|\beta_{A,k}\|_2$, the first inequality is based on \eqref{eq:D-Dhat-inv_p-to_zero}, the second is due to Lemma 1 of \cite{cai2021optimal}, and the last  is due to Lemma \ref{lemma:bound-G_21}.
	
	Let $z_{*}=\Phi^{-1}(1-\alpha/(2q))$ and $\Delta_n = \frac{C_1 \delta_{\max}}{\sigma\sqrt{\log(pn)/n}+\beta_{A,\max}+C_1}$. The bound for the power holds trivially when $\lim_{n\rightarrow\infty} \{1-F(\alpha,\Delta_n,q)\}\leq 0$. When $\lim_{n\rightarrow\infty} \{1-F(\alpha,\Delta_n,q)\}> 0$, based on \eqref{eq:T},  \eqref{eq:normal-app_v_tilde} and \eqref{eq:lower-v_j-star}, we deduce
	\begin{scriptsize}
		\begin{align}
			\notag
			\lim_{n\to \infty} &\inf_{\xi \in \mathcal{H}_1(s, \delta)} 
			\frac{\prob_{\xi}(\phi_{\alpha}=1)}{1 - F\left(\alpha, \Delta_n, q\right)} \\
			\notag
			&=\lim_{n\to \infty} 
			\frac{1}{1 - F\left(\alpha, \Delta_n, q\right)}
			\inf_{\xi \in \mathcal{H}_1(s, \delta)} 
			\left\{\prob_{\xi}(\phi_{\alpha}=1):|v_{j^{\star}}|\ge \Delta_n\right\} \\
			\notag
			&=
			\lim_{n\to \infty} 
			\frac{1}{1 - F\left(\alpha, \Delta_n, q\right)}
			\inf_{\xi \in \mathcal{H}_1(s, \delta)}
			\left\{\prob_{\xi}(\|T\|_{\infty}\ge z_{*}):|v_{j^{\star}}|\ge \Delta_n\right\}  \\
			\notag
			&\ge 
			\lim_{n\to \infty} 
			\frac{1}{1 - F\left(\alpha, \Delta_n, q\right)}
			\inf_{\xi \in \mathcal{H}_1(s, \delta)} 
			\left\{\prob_{\xi}(|v_{j^{\star}} - \tilde{v}_{j^{\star}}|\ge z_{*}):|v_{j^{\star}}|\ge \Delta_n\right\} \\
			\notag
			&=
			\lim_{n\to \infty} 
			\frac{1}{1 - F\left(\alpha, \Delta_n, q\right)}
			\left( 1 - \sup_{\xi \in \mathcal{H}_1(s, \delta)} 
			\left\{\prob_{\xi}(|v_{j^{\star}} - \tilde{v}_{j^{\star}}|\le z_{*}):|v_{j^{\star}}|\ge \Delta_n\right\} \right) \\
			\notag
			&\ge 
			\lim_{n\to \infty} 
			\frac{1}{1 - F\left(\alpha, \Delta_n, q\right)}
			\left( 1 - \sup_{\xi \in \mathcal{H}_1(s, \delta)} 
			\left\{\prob_{\xi}(|v_{j^{\star}} - \tilde{v}_{j^{\star}}|\le z_{*}):|v_{j^{\star}}|\ge \Delta_n\right\} \right) \\
			\notag
			&\ge \lim_{n\to \infty} 
			\frac{1}{1 - F\left(\alpha, \Delta_n, q\right)}
			\left( 1 - \sup_{\xi \in \mathcal{H}_1(s, \delta)} 
			\left\{\prob_{\xi}(\exists k \in \{1, \ldots, q\}: |v_{j^{\star}} - \tilde{v}_{k}|\le z_{*}):|v_{j^{\star}}|\ge \Delta_n\right\} \right)\\
			\notag
			&\ge \lim_{n\to \infty} 
			\frac{1}{1 - F\left(\alpha, \Delta_n, q\right)}
			\left( 1 - \sup_{\xi \in \mathcal{H}_1(s, \delta)} 
			\left\{\sum_{k=1}^q \prob_{\xi}( |v_{j^{\star}} - \tilde{v}_{k}|\le z_{*}):|v_{j^{\star}}|\ge \Delta_n\right\} \right) \\
			\notag
			&\ge \lim_{n\to \infty}
			\frac{1 - q \prob( |\Delta_n - G_0|\le z_{*})}{1 - F\left(\alpha, \Delta_n, q\right)}  = 1,
		\end{align}
	\end{scriptsize}
	where $G_0$ is a standard normal random variable. 
\end{proof}

\section{Lemmas}\label{sec:lemma}

\begin{lemma}\label{lemma:bound-G_21}
	For $j=1,\ldots, q$, let $\beta_{A,j}$ denote the $j$th column of $\beta_{A}$. 
	Under Assumption (A1)  and {$q\ll \sqrt n$}, for a constant $c>0$ and all $n\geq 2$, with probability tending to one, one has 
	\begin{equation}
		\sup_j\|g_j\|_2 \leq  c\sigma\sqrt{\log(pn)/n} + \sup_j\|\beta_{A,j}\|_2,
		\label{eq:gj-norm-upper-lower-bounds}
	\end{equation}
	where $g_j$ is defined in \eqref{eq:gj-def}.
\end{lemma}

\begin{proof}
	Note that, for each $j=1, \ldots, q$, $\|g_j\|_2=\|\tilde{g}_j\|_2$, and
	\begin{equation*}
		\tilde{g}_j = \tilde{G} e_j = n^{-1}E^{\transpose} A \hat{\Sigma}_{A}^{-1} e_j + \beta_{A,j}
		=
		\frac{1}{n}\sum_{i=1}^{n} E_i A_i^{\transpose} \hat{\Sigma}_{A}^{-1} e_j + \beta_{A,j}
	\end{equation*}
	where $e_j\in \real^q$ represents the canonical basis vector with $1$ in the $j$th coordinate. 
	Hence, for each $j=1, \ldots, q$,
	\begin{equation}
		\|g_j\|_2
		\le 
		\left\|n^{-1}\sum_{i=1}^{n} E_i A_i^{\transpose} \hat{\Sigma}_{A}^{-1} e_j\right\|_2 + \|\beta_{A,j}\|_2.
		\label{eq:init_bound_g_l2}
	\end{equation}
	
	We first observe that, for each $i=1,\ldots, n$, we have the following decomposition:
	\begin{equation}
		A_i^{\transpose} \hat{\Sigma}_{A}^{-1} e_j = A_i^{\transpose} \Sigma_{A}^{-1} e_j + A_i^{\transpose} (\hat{\Sigma}_{A}^{-1}-\Sigma_{A}^{-1}) e_j.
		\label{eq:Ai-term-decomp}
	\end{equation}
	For the second term $\Delta_{ij} \define A_i^{\transpose} (\hat{\Sigma}_{A}^{-1}-\Sigma_{A}^{-1}) e_j$, under Assumption (A1) and $q\ll n^{1/2}$, one has
	\begin{align}
		\sup_{ij}|\Delta_{ij}| & = \sup_{ij}|A_i^{\transpose} (\hat{\Sigma}_{A}^{-1}-\Sigma_{A}^{-1}) e_j|
		\le
		\sup_{ij}\|A_i\|_2 \|\hat{\Sigma}_{A}^{-1}-\Sigma_{A}^{-1}\|_2 \nonumber\\
		& \lesssim
		(\sqrt{q}+n^{1/4})\left(\frac{q+\log n}{n}\right)^{1/2} 
		\to 0,
		\label{eq:Deltai-bound}
	\end{align}
	where the second inequality holds with probability tending to one due to Lemma \ref{lemma:A-norm-concen} and Equation (1.6) in \cite{kereta2021estimating}.
	Based on the decomposition \eqref{eq:Ai-term-decomp}, one has 
	\begin{equation}
		\left\|n^{-1}\sum_{i=1}^{n} E_i A_i^{\transpose} \hat{\Sigma}_{A}^{-1} e_j\right\|_2 
		\le 
		\left\|n^{-1}\sum_{i=1}^{n} E_i A_i^{\transpose} \Sigma_{A}^{-1} e_j\right\|_2 + \left\|n^{-1}\sum_{i=1}^{n} E_i \Delta_{ij}\right\|_2.
		\label{eq:sandwich-avg-error}
	\end{equation}
	To establish \eqref{eq:gj-norm-upper-lower-bounds}, we claim that, with probability tending to one, one has 
	\begin{equation}
		\sup_j\left\|n^{-1}\sum_{i=1}^{n} E_i A_i^{\transpose} \Sigma_{A}^{-1} e_j\right\|_2 \lesssim \sigma\sqrt{\frac{\log(pn)}{n}},
		\label{eq:pop-error-norm-concen}
	\end{equation}
	and 
	\begin{equation}
		\sup_j\left\|n^{-1}\sum_{i=1}^{n} E_i \Delta_{ij}\right\|_2 \lesssim \sigma\sqrt{\frac{\log(pn)}{n}}.
		\label{eq:error-prod-norm-concen}
	\end{equation}
	Combining \eqref{eq:pop-error-norm-concen} and \eqref{eq:error-prod-norm-concen} with \eqref{eq:sandwich-avg-error}, we have
	\[
	{\sup_{j}}\left\|n^{-1}\sum_{i=1}^{n} E_i A_i^{\transpose} \hat{\Sigma}_{A}^{-1} e_j\right\|_2
	\lesssim 
	\sigma\sqrt{\frac{\log(pn)}{n}}.
	\] 
	Applying this bound to \eqref{eq:init_bound_g_l2} gives to the desired result \eqref{eq:gj-norm-upper-lower-bounds}.
	
	It remains to prove the claims \eqref{eq:pop-error-norm-concen} and \eqref{eq:error-prod-norm-concen}.
	To ease notation, let $\tilde A_{ij} = A_i^{\transpose}\Sigma_{A}^{-1}e_j$. 
	Since $E_i$ is centered and norm-subGaussian with a common parameter $\sigma$ across all $i$ conditional on $A_i$, we conclude that $E_i \tilde{A}_{ij}$ is norm-subGaussian with the parameter $\sigma |\tilde{A}_{ij}|$. 
	Define the event $\mathcal{A}=\{{\sum_{i=1}^{n}\tilde{A}_{ij}^2} \leq C_2^2 {n}{\text{ for all } j}\}$ for a sufficiently large constant $C_2>0$. Observing that $\tilde A_{ij}$ is a subGaussian random variable with a parameter upper bounded by a common constant across all $j$, we apply {Theorem 3.1.1 in \cite{vershynin2018high} (specifically, Equation (3.1) therein)
		on the norm of the subGaussian random vector $(\tilde A_{1j},\ldots,\tilde A_{nj})$ of independent coordinates} to conclude  $\prob(\mathcal{A}^{c})=o(1)$. Moreover, by Corollary 7 in \cite{jin2019short}, 
	\begin{small}
		\begin{align*}
			\prob\left(\left. {\sup_j}\left\|n^{-1}\sum_{i=1}^{n} E_i \tilde{A}_{ij}\right\|_2 \geq C_3 \sigma\frac{\sqrt{\log(pn)}}{\sqrt{n}}\sqrt{\frac{\sum_{i=1}^{n}\tilde{A}_{ij}^2}{n}}\, \,\right| A_1,\ldots,A_n \right) 
			&\le 
			\frac{q}{n},
		\end{align*}
		where $C_3>0$ is a constant. 
		Consequently,
		\begin{align*}
			& \prob\left({\sup_j}\left\|n^{-1}\sum_{i=1}^{n} E_i \tilde{A}_{ij}\right\|_2 \geq C_3C_2 \sigma\sqrt{\frac{\log(pn)}{n}} \right) \\
			&= \int_{\mathcal{A}} \prob\left({\sup_j}\left.\left\|n^{-1}\sum_{i=1}^{n} E_i \tilde{A}_{ij}\right\|_2 \geq C_3C_2\sigma\sqrt{\frac{\log(pn)}{n}} \,\,\right| A_1,\ldots,A_n\right)  \mathrm{d}F_{A}
			+ 
			\prob\left(\mathcal{A}^c\right) \\
			&\le \int_{\mathcal{A}} \prob\left({\sup_j}\left.\left\|n^{-1}\sum_{i=1}^{n} E_i \tilde{A}_{ij}\right\|_2 \geq C_3 \sigma\sqrt{\frac{\log(pn)}{n}} \sqrt{\frac{\sum_{i=1}^n \tilde{A}_{ij}^2}{n}}\,\,\right| A_1,\ldots,A_n\right)  \mathrm{d}F_{A} \\
			&\quad + 
			\prob\left(\mathcal{A}^c\right)\\
			&\le
			\frac{q}{n} + o(1) \to 0,
		\end{align*}
	\end{small}
	where $F_{A}$ is the joint distribution function of $A_1,\ldots,A_n$. This proves the claim \eqref{eq:pop-error-norm-concen}.
	Similarly, one can establish \eqref{eq:error-prod-norm-concen} based on \eqref{eq:Deltai-bound}.
\end{proof}


\begin{lemma}\label{lemma:A-norm-concen}
	Consider i.i.d. $q$-dimensional centered random vectors $A_1,\ldots,A_n$ such that the (possibly dependent)  elements of $A_{i}$  are  subGaussian with a common parameter. If $q\lesssim n^{1/2}$, then $\max_{1\le i \le n} | \|A_i\|_2 - \sigma_{A}\sqrt{q} | \leq c n^{1/4}$  with probability at least $1 - e^{-c n^{1/2}/q+\log(nq)}$ for some constant $c>0$, where $\sigma_A^2=q^{-1}\sum_{j=1}^q \expect A_{ij}^2$. 
\end{lemma}	

\begin{proof}
	The proof strategy is similar to that of Theorem 3.1.1 in \cite{vershynin2018high} with modifications to possibly dependent subGaussian elements. 
	Note that, for each $i=1,\ldots, n$,
	\[
	\frac{1}{q} \|A_i\|_2^2 - \sigma_{A}^2
	=
	\frac{1}{q} \sum_{j=1}^{q} (A_{ij}^2 - \sigma_{Aj}^2),
	\]
	where $\sigma_{Aj}^2=\expect A_{ij}^2$. 
	Since $A_{ij}$ is subGaussian, $A_{ij}^2 - \sigma_{Aj}^2$ is centered and sub-exponential with a common parameter. Applying the union bound, we deduce that, for any $u\ge 0$,
	\begin{align}
		\notag
		\prob\left( \left| \frac{1}{q} \|A_i\|_2^2 - \sigma_{A}^2 \right| \ge u \right)
		&=
		\prob\left( \left| \frac{1}{q} \sum_{j=1}^{q} (A_{ij}^2 - \sigma_{Aj}^2) \right| \ge u \right) \\
		\notag
		&\le 
		\prob\left( \sum_{j=1}^{q}\left| A_{ij}^2 - \sigma_{Aj}^2 \right| \ge uq \right) \\
		\notag
		&\le 
		\sum_{j=1}^{q}\prob\left( \left| A_{ij}^2 - \sigma_{Aj}^2 \right| \ge u \right) \\
		&\le \sum_{j=1}^{q} e^{-c_1 u}
		= e^{-c_1 u + \log q},
		\label{eq:A-norm-subE-ineq}
	\end{align}
	where the last inequality is due to the definition of a sub-exponential random variable, and $c_1>0$ is a constant. To establish a concentration inequality for $\|A_i\|_2$, we use the following fact that for all $z\ge 0$:
	\[
	|z-1|\ge \delta
	\quad \text{implies} \quad
	|z^2-1|\ge \max(\delta, \delta^2).
	\]
	Therefore, for any $\delta\ge 0$, 
	\begin{align*}
		\prob\left( \left| \frac{1}{\sqrt{q}} \|A_i\|_2 - \sigma_{A} \right| \ge \delta \right)
		&=
		\prob\left( \left| \frac{1}{\sigma_{A}\sqrt{q}} \|A_i\|_2 - 1 \right| \ge \frac{\delta}{\sigma_{A}} \right) \\
		&\le
		\prob\left( \left| \frac{1}{\sigma_{A}^2 q} \|A_i\|_2^2 - 1 \right| \ge \max\left\{\frac{\delta}{\sigma_{A}}, \frac{\delta^2}{\sigma_{A}^2}\right\} \right) \\
		&\le
		\prob\left( \left| \frac{1}{q} \|A_i\|_2^2 - \sigma_{A}^2 \right| \ge \max\left\{\sigma_{A} \delta, \delta^2 \right\} \right) \\
		&\le
		e^{-c_2 \max\{\delta, \delta^2\} + \log q},
	\end{align*}
	where the last inequality is from \eqref{eq:A-norm-subE-ineq}, and $c_2>0$ is a constant. 
	Changing the variable $\delta$ to $t=\sqrt{q}\delta$, we obtain 
	\begin{equation}
		\prob\left( \left| \|A_i\|_2 - \sigma_{A}\sqrt{q} \right| \ge t \right)
		\le
		e^{-c_3\max\{t/\sqrt{q}, t^2/q\}+\log q},
		\label{eq:A-norm-con-ineq}
	\end{equation}
	for some constant $c_3>0$.
	
	Finally, taking $t=n^{1/4}\gtrsim \sqrt{q}$ and applying the union bound, one has
	\[
	\prob\left( \max_{1\le i \le n}\left| \|A_i\|_2 - \sigma_{A}\sqrt{q} \right| \ge n^{1/4} \right) \le e^{-c n^{1/2}/q+\log(nq)},
	\]
	for some constant $c>0$, as desired.
\end{proof}

\section{Additional Simulation Studies}\label{sec:additional-simu}

{
\subsection{Additional Sparse Structures of $\theta_{M}$}
}

In Section \ref{sec:simu-size}, we investigate the test size under a hard sparsity assumption for $\theta_{M}$ when it is non-zero. In this section, we explore two additional sparsity structures similar to those in \cite{cai2021optimal}, with a fixed sparsity parameter $s=5$:
\begin{itemize}
	\item Capped-$\ell_1$ sparsity: $\theta_{M,k} = 0.2 k \indicator{1 \le k \le s}+0.1\lambda_0 \indicator{2s+1 \le k \le p/5+s}$, for $k=1, \ldots, p$ and $\lambda_0 = \sqrt{2\log p / n}$;	
	\item Decaying coefficients: $\theta_{M,k}=0.2 k \indicator{1 \le k \le s}+ (k-s)^{-1.5} \indicator{k\ge 2s+1}$, for $k=1, \ldots, p$.
\end{itemize} 
Since $\theta_{M}$ is not hard-sparse, to construct a null case, we only consider a sparse setting for $\beta_{A}$. Specifically, when the true coefficient matrix $\beta_A\in\real^{p\times q}$ is non-zero, each of its element is given by
\[
\beta_{A,jk} = \begin{cases}
	0.2 \kappa_j(k-s) & \text{ if }s+1\leq k\leq 2s \\
	0 & \text{ otherwise,}
\end{cases}
\]
for $s=5$, $k=1,\ldots, p$ and $j=1,\ldots, q$, where each $\kappa_j(\cdot)$ is a random permutation of $\{1, \ldots, s\}$.
The results are presented in Table \ref{tab:size-Bern-others}, which has similar implications to those in Table \ref{tab:size-Bern}.

\setlength{\tabcolsep}{4pt} 
\renewcommand{\arraystretch}{1} 
\begin{table}[ht]
	\centering
	\caption{Empirical sizes for Bernoulli exposures under other types of sparsity for $\theta_{M}$ with a nonzero $\theta_{A}=c_0 \onevec{q}$, where $c_0=0.5$ when $q=1$, and $c_0=0.3$ when $q=3$.}
	\begin{tabular}{cccccccccccc}
		\cmidrule{1-12}
		Sparsity for $\theta_M$ & Case  & $q$   & $n$   & Zhou-$0$ & Zhou-$1$ & Guo-$0$ & Guo-$1$ & $\chi^2$-$0$ & $\chi^2$-$1$ & $\text{Bonf}$-$0$ & $\text{Bonf}$-$1$ \\
		\cmidrule{1-12}
		\multicolumn{1}{c}{\multirow{16}[16]{*}{\shortstack{Capped-$\ell_1$\\($\theta_M\ne 0$)}}} & \multicolumn{1}{c}{\multirow{4}[4]{*}{\shortstack{$\beta_A=0$\\\newline{}($\Sigma_0=\Sigma_{AR}$)}}} & \multirow{2}[2]{*}{1} & 50    & 0.034 & 0.034 & 0.092 & 0.092 & 0.040 & 0.036 & 0.040 & 0.036 \\
		&       &       & 300   & 0.044 & 0.044 & 0.040 & 0.040 & 0.040 & 0.038 & 0.040 & 0.038 \\
		\cmidrule{3-12}          &       & \multirow{2}[2]{*}{3} & 50    & 0.042 & 0.042 & 0.168 & 0.168 & 0.054 & 0.046 & 0.052 & 0.050 \\
		&       &       & 300   & 0.040 & 0.040 & 0.044 & 0.044 & 0.030 & 0.024 & 0.030 & 0.028 \\
		\cmidrule{2-12}          & \multicolumn{1}{c}{\multirow{4}[4]{*}{\shortstack{$\beta_A\ne 0$-Sparse\\\newline{}($\Sigma_0=\Sigma_{AR}$)}}} & \multirow{2}[2]{*}{1} & 50    & 0.038 & 0.038 & 0.104 & 0.104 & 0.046 & 0.044 & 0.046 & 0.044 \\
		&       &       & 300   & 0.048 & 0.048 & 0.040 & 0.040 & 0.042 & 0.040 & 0.042 & 0.040 \\
		\cmidrule{3-12}          &       & \multirow{2}[2]{*}{3} & 50    & 0.052 & 0.052 & 0.188 & 0.188 & 0.050 & 0.044 & 0.064 & 0.060 \\
		&       &       & 300   & 0.092 & 0.092 & 0.044 & 0.044 & 0.030 & 0.028 & 0.032 & 0.028 \\
		\cmidrule{2-12}          & \multicolumn{1}{c}{\multirow{4}[4]{*}{\shortstack{$\beta_A=0$\\\newline{}($\Sigma_0=\Sigma_{CS}$)}}} & \multirow{2}[2]{*}{1} & 50    & 0.058 & 0.058 & 0.064 & 0.064 & 0.032 & 0.032 & 0.032 & 0.032 \\
		&       &       & 300   & 0.056 & 0.056 & 0.056 & 0.056 & 0.048 & 0.048 & 0.048 & 0.048 \\
		\cmidrule{3-12}          &       & \multirow{2}[2]{*}{3} & 50    & 0.052 & 0.052 & 0.076 & 0.076 & 0.032 & 0.030 & 0.034 & 0.034 \\
		&       &       & 300   & 0.066 & 0.066 & 0.054 & 0.054 & 0.038 & 0.038 & 0.038 & 0.038 \\
		\cmidrule{2-12}          & \multicolumn{1}{c}{\multirow{4}[4]{*}{\shortstack{$\beta_A\ne 0$-Sparse\\\newline{}($\Sigma_0=\Sigma_{CS}$)}}} & \multirow{2}[2]{*}{1} & 50    & 0.054 & 0.054 & 0.060 & 0.060 & 0.046 & 0.046 & 0.046 & 0.046 \\
		&       &       & 300   & 0.064 & 0.064 & 0.058 & 0.058 & 0.046 & 0.044 & 0.046 & 0.044 \\
		\cmidrule{3-12}          &       & \multirow{2}[2]{*}{3} & 50    & 0.050 & 0.050 & 0.080 & 0.080 & 0.042 & 0.042 & 0.036 & 0.034 \\
		&       &       & 300   & 0.062 & 0.062 & 0.052 & 0.052 & 0.052 & 0.052 & 0.056 & 0.056 \\
		\cmidrule{1-12}
		\multicolumn{1}{c}{\multirow{16}[16]{*}{\shortstack{Decaying \\ Coefficients \\ ($\theta_M\ne 0$)}}} & \multicolumn{1}{c}{\multirow{4}[4]{*}{\shortstack{$\beta_A=0$\\\newline{}($\Sigma_0=\Sigma_{AR}$)}}} & \multirow{2}[2]{*}{1} & 50    & 0.034 & 0.034 & 0.092 & 0.092 & 0.038 & 0.034 & 0.038 & 0.034 \\
		&       &       & 300   & 0.050 & 0.050 & 0.044 & 0.044 & 0.042 & 0.038 & 0.042 & 0.038 \\
		\cmidrule{3-12}          &       & \multirow{2}[2]{*}{3} & 50    & 0.044 & 0.044 & 0.174 & 0.174 & 0.050 & 0.048 & 0.052 & 0.050 \\
		&       &       & 300   & 0.044 & 0.044 & 0.048 & 0.048 & 0.034 & 0.030 & 0.032 & 0.032 \\
		\cmidrule{2-12}          & \multicolumn{1}{c}{\multirow{4}[4]{*}{\shortstack{$\beta_A\ne 0$-Sparse\\\newline{}($\Sigma_0=\Sigma_{AR}$)}}} & \multirow{2}[2]{*}{1} & 50    & 0.038 & 0.038 & 0.112 & 0.112 & 0.044 & 0.044 & 0.044 & 0.044 \\
		&       &       & 300   & 0.052 & 0.050 & 0.044 & 0.044 & 0.046 & 0.044 & 0.046 & 0.044 \\
		\cmidrule{3-12}          &       & \multirow{2}[2]{*}{3} & 50    & 0.050 & 0.050 & 0.194 & 0.192 & 0.046 & 0.044 & 0.064 & 0.060 \\
		&       &       & 300   & 0.088 & 0.088 & 0.048 & 0.048 & 0.038 & 0.036 & 0.038 & 0.036 \\
		\cmidrule{2-12}          & \multicolumn{1}{c}{\multirow{4}[4]{*}{\shortstack{$\beta_A=0$\\\newline{}($\Sigma_0=\Sigma_{CS}$)}}} & \multirow{2}[2]{*}{1} & 50    & 0.056 & 0.056 & 0.058 & 0.058 & 0.034 & 0.032 & 0.034 & 0.032 \\
		&       &       & 300   & 0.062 & 0.062 & 0.058 & 0.058 & 0.046 & 0.046 & 0.046 & 0.046 \\
		\cmidrule{3-12}          &       & \multirow{2}[2]{*}{3} & 50    & 0.052 & 0.052 & 0.070 & 0.070 & 0.034 & 0.034 & 0.034 & 0.034 \\
		&       &       & 300   & 0.064 & 0.064 & 0.050 & 0.050 & 0.040 & 0.040 & 0.042 & 0.042 \\
		\cmidrule{2-12}          & \multicolumn{1}{c}{\multirow{4}[4]{*}{\shortstack{$\beta_A\ne 0$-Sparse\\\newline{}($\Sigma_0=\Sigma_{CS}$)}}} & \multirow{2}[2]{*}{1} & 50    & 0.056 & 0.056 & 0.058 & 0.058 & 0.050 & 0.050 & 0.050 & 0.050 \\
		&       &       & 300   & 0.054 & 0.054 & 0.056 & 0.056 & 0.042 & 0.042 & 0.042 & 0.042 \\
		\cmidrule{3-12}          &       & \multirow{2}[2]{*}{3} & 50    & 0.054 & 0.054 & 0.072 & 0.072 & 0.042 & 0.042 & 0.032 & 0.032 \\
		&       &       & 300   & 0.060 & 0.060 & 0.052 & 0.052 & 0.050 & 0.050 & 0.060 & 0.060 \\
		\cmidrule{1-12}
	\end{tabular}%
	\label{tab:size-Bern-others}%
\end{table}%

To investigate the power behavior of the proposed test under alternative sparsity structures, we consider similar data generating processes as in Section \ref{sec:simu-power}, but with different settings for $\theta_{M}$:
\begin{itemize}
	\item Capped-$\ell_1$ sparsity: $\theta_{M,k} = c_2 \left(0.3 k \indicator{1 \le k \le s}+0.1\lambda_0 \indicator{s+1 \le k \le p/5} \right)$, for $k=1, \ldots, p$ and $\lambda_0 = \sqrt{2\log p / n}$;
	\item Decaying coefficients: $\theta_{M,k}=c_2 \left(0.3 k \indicator{1 \le k \le s}+ k^{-1.5} \indicator{k\ge s+1} \right)$, for $k=1, \ldots, p$.
\end{itemize}
The results are shown in Table \ref{tab:power-Bern-others}, which has similar insights as those in Table \ref{tab:power-Bern}.

\setlength{\tabcolsep}{10pt} 
\renewcommand{\arraystretch}{1} 
\begin{table}[ht]
	\centering
	\caption{Empirical power behaviors for Bernoulli exposures under other types of sparsity for $\theta_{M}$.}
	\begin{tabular}{@{\extracolsep\fill}ccccccc}
		\cmidrule{1-7}
		Sparsity & Case  & $n$   & $c_1$ or $c_2$ & Zhou-$1$ & Guo-$1$ & $\text{Bonf}$-$1$ \\
		\cmidrule{1-7}
		\multirow{20}[7]{*}{Capped-$\ell_1$ Sparsity} & \multirow{10}[4]{*}{\shortstack{Fix $\theta_{M}$, \\ Vary $\beta_{A}$}} & \multirow{5}[2]{*}{50} & 0     & 0.034 & 0.102 & 0.036 \\
		&       &       & 1/8   & 0.108 & 0.15  & 0.09 \\
		&       &       & 1/4   & 0.306 & 0.274 & 0.17 \\
		&       &       & 1/2   & 0.806 & 0.732 & 0.46 \\
		&       &       & 1     & 1     & 0.984 & 0.826 \\
		\cmidrule{3-7}          &       & \multirow{5}[2]{*}{300} & 0     & 0.042 & 0.044 & 0.04 \\
		&       &       & 1/8   & 0.444 & 0.376 & 0.174 \\
		&       &       & 1/4   & 0.94  & 0.928 & 0.58 \\
		&       &       & 1/2   & 1     & 1     & 0.966 \\
		&       &       & 1     & 1     & 1     & 0.984 \\
		\cmidrule{2-7}          & \multirow{10}[3]{*}{\shortstack{Fix $\beta_{A}$, \\ Vary $\theta_{M}$}} & \multirow{5}[2]{*}{50} & 0     & 0.12  & 0.246 & 0.04 \\
		&       &       & 1/8   & 0.738 & 0.286 & 0.104 \\
		&       &       & 1/4   & 0.922 & 0.518 & 0.258 \\
		&       &       & 1/2   & 0.942 & 0.824 & 0.494 \\
		&       &       & 1     & 0.946 & 0.922 & 0.766 \\
		\cmidrule{3-7}          &       & \multirow{5}[1]{*}{300} & 0     & 0.114 & 0     & 0.038 \\
		&       &       & 1/8   & 0.996 & 0.558 & 0.188 \\
		&       &       & 1/4   & 1     & 1     & 0.508 \\
		&       &       & 1/2   & 1     & 1     & 0.962 \\
		&       &       & 1     & 1     & 1     & 0.996 \\
		\cmidrule{1-7}   
		\multirow{20}[7]{*}{Decaying Coefficients} & \multirow{10}[3]{*}{\shortstack{Fix $\theta_{M}$, \\ Vary $\beta_{A}$}} & \multirow{5}[1]{*}{50} & 0     & 0.038 & 0.092 & 0.04 \\
		&       &       & 1/8   & 0.112 & 0.134 & 0.096 \\
		&       &       & 1/4   & 0.322 & 0.278 & 0.182 \\
		&       &       & 1/2   & 0.816 & 0.732 & 0.474 \\
		&       &       & 1     & 1     & 0.99  & 0.838 \\
		\cmidrule{3-7}          &       & \multirow{5}[2]{*}{300} & 0     & 0.042 & 0.048 & 0.042 \\
		&       &       & 1/8   & 0.444 & 0.374 & 0.168 \\
		&       &       & 1/4   & 0.942 & 0.928 & 0.572 \\
		&       &       & 1/2   & 1     & 1     & 0.966 \\
		&       &       & 1     & 1     & 1     & 0.982 \\
		\cmidrule{2-7}          & \multirow{10}[4]{*}{\shortstack{Fix $\beta_{A}$, \\ Vary $\theta_{M}$}} & \multirow{5}[2]{*}{50} & 0     & 0.12  & 0.246 & 0.04 \\
		&       &       & 1/8   & 0.758 & 0.284 & 0.108 \\
		&       &       & 1/4   & 0.938 & 0.546 & 0.266 \\
		&       &       & 1/2   & 0.952 & 0.856 & 0.516 \\
		&       &       & 1     & 0.954 & 0.932 & 0.782 \\
		\cmidrule{3-7}          &       & \multirow{5}[2]{*}{300} & 0     & 0.114 & 0     & 0.038 \\
		&       &       & 1/8   & 0.996 & 0.586 & 0.184 \\
		&       &       & 1/4   & 1     & 1     & 0.506 \\
		&       &       & 1/2   & 1     & 1     & 0.96 \\
		&       &       & 1     & 1     & 1     & 0.996 \\
		\cmidrule{1-7}
	\end{tabular}%
	\label{tab:power-Bern-others}%
\end{table}%

\subsection{Normally Distributed Exposures}

Tables \ref{tab:size-Gaussian}--\ref{tab:power-Gaussian} respectively present the empirical sizes and power behaviors of the test procedures when each entry of the exposures $A_{ij}$, for $i=1,\ldots, n$ and $j=1,\ldots, q$, is independently sampled from the Gaussian distribution $N(0, 0.5^2)$. 
The results are similar to  those for the Bernoulli distribution.

\setlength{\tabcolsep}{4pt} 
\renewcommand{\arraystretch}{0.85} 
\begin{table}[htbp]
	\centering
	\caption{Empirical sizes for exposures sampled from the Gaussian distribution $N(0, 0.5^2)$ with a nonzero $\theta_{A}=c_0 \onevec{q}$, where $c_0=0.5$ for $q=1$ and $c_0=0.3$ for $q=3$.}
		\begin{tabular}{cccccccccccc}
			\cmidrule{1-12}
			Sparsity for $\theta_M$ & Case  & $q$   & $n$   & Zhou-$0$ & Zhou-$1$ & Guo-$0$ & Guo-$1$ & $\chi^2$-$0$ & $\chi^2$-$1$ & $\text{Bonf}$-$0$ & $\text{Bonf}$-$1$ \\
			\cmidrule{1-12}
			\multicolumn{1}{c}{\multirow{24}[24]{*}{\shortstack{Zero\\\newline{}($\theta_M=0$)}}} & \multicolumn{1}{c}{\multirow{4}[4]{*}{\shortstack{$\beta_A=0$\\\newline{}($\Sigma_0=\Sigma_{AR}$)}}} & \multirow{2}[2]{*}{1} & 50    & 0.186 & 0.164 & 0.212 & 0.212 & 0.042 & 0.032 & 0.042 & 0.032 \\
			&       &       & 300   & 0.710 & 0.700 & 0.000 & 0.000 & 0.036 & 0.036 & 0.036 & 0.036 \\
			\cmidrule{3-12}          &       & \multirow{2}[2]{*}{3} & 50    & 0.152 & 0.122 & 0.550 & 0.550 & 0.016 & 0.012 & 0.036 & 0.024 \\
			&       &       & 300   & 0.532 & 0.514 & 0.000 & 0.000 & 0.048 & 0.042 & 0.032 & 0.030 \\
			\cmidrule{2-12}          & \multicolumn{1}{c}{\multirow{4}[4]{*}{\shortstack{$\beta_A\ne 0$-Sparse\\\newline{}($\Sigma_0=\Sigma_{AR}$)}}} & \multirow{2}[2]{*}{1} & 50    & 0.146 & 0.146 & 0.208 & 0.208 & 0.034 & 0.034 & 0.034 & 0.034 \\
			&       &       & 300   & 0.190 & 0.188 & 0.000 & 0.000 & 0.036 & 0.036 & 0.036 & 0.036 \\
			\cmidrule{3-12}          &       & \multirow{2}[2]{*}{3} & 50    & 0.188 & 0.180 & 0.644 & 0.644 & 0.050 & 0.034 & 0.048 & 0.038 \\
			&       &       & 300   & 0.526 & 0.526 & 0.000 & 0.000 & 0.050 & 0.044 & 0.030 & 0.030 \\
			\cmidrule{2-12}          & \multicolumn{1}{c}{\multirow{4}[4]{*}{\shortstack{$\beta_A\ne 0$-Dense\\\newline{}($\Sigma_0=\Sigma_{AR}$)}}} & \multirow{2}[2]{*}{1} & 50    & 0.176 & 0.172 & 0.244 & 0.244 & 0.046 & 0.040 & 0.046 & 0.040 \\
			&       &       & 300   & 0.210 & 0.208 & 0.000 & 0.000 & 0.050 & 0.042 & 0.050 & 0.042 \\
			\cmidrule{3-12}          &       & \multirow{2}[2]{*}{3} & 50    & 0.150 & 0.136 & 0.586 & 0.586 & 0.040 & 0.026 & 0.042 & 0.026 \\
			&       &       & 300   & 0.722 & 0.720 & 0.000 & 0.000 & 0.050 & 0.046 & 0.052 & 0.048 \\
			\cmidrule{2-12}          & \multicolumn{1}{c}{\multirow{4}[4]{*}{\shortstack{$\beta_A=0$\\\newline{}($\Sigma_0=\Sigma_{CS}$)}}} & \multirow{2}[2]{*}{1} & 50    & 0.084 & 0.074 & 0.128 & 0.128 & 0.062 & 0.050 & 0.062 & 0.050 \\
			&       &       & 300   & 0.200 & 0.186 & 0.084 & 0.084 & 0.038 & 0.038 & 0.038 & 0.038 \\
			\cmidrule{3-12}          &       & \multirow{2}[2]{*}{3} & 50    & 0.072 & 0.054 & 0.222 & 0.222 & 0.048 & 0.022 & 0.044 & 0.026 \\
			&       &       & 300   & 0.218 & 0.204 & 0.058 & 0.058 & 0.048 & 0.036 & 0.034 & 0.032 \\
			\cmidrule{2-12}          & \multicolumn{1}{c}{\multirow{4}[4]{*}{\shortstack{$\beta_A\ne 0$-Sparse\\\newline{}($\Sigma_0=\Sigma_{CS}$)}}} & \multirow{2}[2]{*}{1} & 50    & 0.130 & 0.122 & 0.152 & 0.152 & 0.038 & 0.030 & 0.038 & 0.030 \\
			&       &       & 300   & 0.646 & 0.640 & 0.088 & 0.084 & 0.038 & 0.038 & 0.038 & 0.038 \\
			\cmidrule{3-12}          &       & \multirow{2}[2]{*}{3} & 50    & 0.110 & 0.090 & 0.268 & 0.268 & 0.076 & 0.032 & 0.072 & 0.040 \\
			&       &       & 300   & 0.690 & 0.690 & 0.092 & 0.092 & 0.058 & 0.054 & 0.052 & 0.044 \\
			\cmidrule{2-12}          & \multicolumn{1}{c}{\multirow{4}[4]{*}{\shortstack{$\beta_A\ne 0$-Dense\\\newline{}($\Sigma_0=\Sigma_{CS}$)}}} & \multirow{2}[2]{*}{1} & 50    & 0.086 & 0.082 & 0.118 & 0.116 & 0.058 & 0.048 & 0.058 & 0.048 \\
			&       &       & 300   & 0.346 & 0.336 & 0.094 & 0.094 & 0.060 & 0.054 & 0.060 & 0.054 \\
			\cmidrule{3-12}          &       & \multirow{2}[2]{*}{3} & 50    & 0.070 & 0.052 & 0.232 & 0.228 & 0.050 & 0.034 & 0.052 & 0.028 \\
			&       &       & 300   & 0.328 & 0.322 & 0.094 & 0.094 & 0.066 & 0.064 & 0.058 & 0.052 \\
			\cmidrule{1-12}
			\multicolumn{1}{c}{\multirow{24}[24]{*}{\shortstack{Hard\\ ($\theta_M\ne 0$)}}} & \multicolumn{1}{c}{\multirow{4}[4]{*}{\shortstack{$\beta_A=0$\\\newline{}($\Sigma_0=\Sigma_{AR}$)}}} & \multirow{2}[2]{*}{1} & 50    & 0.034 & 0.034 & 0.068 & 0.068 & 0.040 & 0.038 & 0.040 & 0.038 \\
			&       &       & 300   & 0.054 & 0.054 & 0.048 & 0.048 & 0.048 & 0.046 & 0.048 & 0.046 \\
			\cmidrule{3-12}          &       & \multirow{2}[2]{*}{3} & 50    & 0.052 & 0.052 & 0.164 & 0.164 & 0.038 & 0.036 & 0.054 & 0.048 \\
			&       &       & 300   & 0.052 & 0.052 & 0.046 & 0.046 & 0.060 & 0.054 & 0.052 & 0.048 \\
			\cmidrule{2-12}          & \multicolumn{1}{c}{\multirow{4}[4]{*}{\shortstack{$\beta_A\ne 0$-Sparse\\\newline{}($\Sigma_0=\Sigma_{AR}$)}}} & \multirow{2}[2]{*}{1} & 50    & 0.040 & 0.040 & 0.096 & 0.096 & 0.038 & 0.034 & 0.038 & 0.034 \\
			&       &       & 300   & 0.046 & 0.046 & 0.048 & 0.048 & 0.042 & 0.040 & 0.042 & 0.040 \\
			\cmidrule{3-12}          &       & \multirow{2}[2]{*}{3} & 50    & 0.054 & 0.054 & 0.228 & 0.228 & 0.048 & 0.044 & 0.054 & 0.052 \\
			&       &       & 300   & 0.084 & 0.084 & 0.046 & 0.046 & 0.058 & 0.054 & 0.054 & 0.050 \\
			\cmidrule{2-12}          & \multicolumn{1}{c}{\multirow{4}[4]{*}{\shortstack{$\beta_A\ne 0$-Dense\\\newline{}($\Sigma_0=\Sigma_{AR}$)}}} & \multirow{2}[2]{*}{1} & 50    & 0.038 & 0.038 & 0.074 & 0.074 & 0.046 & 0.038 & 0.046 & 0.038 \\
			&       &       & 300   & 0.056 & 0.056 & 0.048 & 0.048 & 0.034 & 0.032 & 0.034 & 0.032 \\
			\cmidrule{3-12}          &       & \multirow{2}[2]{*}{3} & 50    & 0.046 & 0.046 & 0.172 & 0.172 & 0.030 & 0.022 & 0.048 & 0.040 \\
			&       &       & 300   & 0.058 & 0.058 & 0.046 & 0.046 & 0.052 & 0.050 & 0.052 & 0.052 \\
			\cmidrule{2-12}          & \multicolumn{1}{c}{\multirow{4}[4]{*}{\shortstack{$\beta_A=0$\\\newline{}($\Sigma_0=\Sigma_{CS}$)}}} & \multirow{2}[2]{*}{1} & 50    & 0.064 & 0.064 & 0.066 & 0.066 & 0.046 & 0.046 & 0.046 & 0.046 \\
			&       &       & 300   & 0.054 & 0.054 & 0.056 & 0.056 & 0.050 & 0.050 & 0.050 & 0.050 \\
			\cmidrule{3-12}          &       & \multirow{2}[2]{*}{3} & 50    & 0.066 & 0.066 & 0.074 & 0.074 & 0.046 & 0.046 & 0.050 & 0.050 \\
			&       &       & 300   & 0.068 & 0.068 & 0.064 & 0.064 & 0.034 & 0.030 & 0.034 & 0.034 \\
			\cmidrule{2-12}          & \multicolumn{1}{c}{\multirow{4}[4]{*}{\shortstack{$\beta_A\ne 0$-Sparse\\\newline{}($\Sigma_0=\Sigma_{CS}$)}}} & \multirow{2}[2]{*}{1} & 50    & 0.068 & 0.068 & 0.068 & 0.068 & 0.054 & 0.054 & 0.054 & 0.054 \\
			&       &       & 300   & 0.062 & 0.062 & 0.056 & 0.056 & 0.052 & 0.052 & 0.052 & 0.052 \\
			\cmidrule{3-12}          &       & \multirow{2}[2]{*}{3} & 50    & 0.066 & 0.066 & 0.074 & 0.074 & 0.062 & 0.062 & 0.056 & 0.050 \\
			&       &       & 300   & 0.058 & 0.058 & 0.064 & 0.064 & 0.042 & 0.040 & 0.052 & 0.052 \\
			\cmidrule{2-12}          & \multicolumn{1}{c}{\multirow{4}[4]{*}{\shortstack{$\beta_A\ne 0$-Dense\\\newline{}($\Sigma_0=\Sigma_{CS}$)}}} & \multirow{2}[2]{*}{1} & 50    & 0.068 & 0.068 & 0.066 & 0.066 & 0.046 & 0.046 & 0.046 & 0.046 \\
			&       &       & 300   & 0.066 & 0.066 & 0.056 & 0.056 & 0.062 & 0.062 & 0.062 & 0.062 \\
			\cmidrule{3-12}          &       & \multirow{2}[2]{*}{3} & 50    & 0.068 & 0.066 & 0.072 & 0.072 & 0.046 & 0.046 & 0.048 & 0.048 \\
			&       &       & 300   & 0.070 & 0.070 & 0.062 & 0.062 & 0.044 & 0.044 & 0.054 & 0.054 \\
			\cmidrule{1-12}
		\end{tabular}%
	\label{tab:size-Gaussian}%
\end{table}%

\setlength{\tabcolsep}{4pt} 
\renewcommand{\arraystretch}{1} 
\begin{table}[htbp]
	\centering
	\caption{Continued-Empirical sizes for exposures sampled from the Gaussian distribution $N(0, 0.5^2)$ with a nonzero $\theta_{A}=c_0 \onevec{q}$, where $c_0=0.5$ for $q=1$ and $c_0=0.3$ for $q=3$.}
	\begin{tabular}{@{\extracolsep\fill}cccccccccccc}
		\cmidrule{1-12}
		Sparsity for $\theta_M$ & Case  & $q$   & $n$   & Zhou-$0$ & Zhou-$1$ & Guo-$0$ & Guo-$1$ & $\chi^2$-$0$ & $\chi^2$-$1$ & $\text{Bonf}$-$0$ & $\text{Bonf}$-$1$ \\
		\cmidrule{1-12}
		\multicolumn{1}{c}{\multirow{16}[16]{*}{\shortstack{Capped-$\ell_1$\\($\theta_M\ne 0$)}}} & \multicolumn{1}{c}{\multirow{4}[4]{*}{\shortstack{$\beta_A=0$\\\newline{}($\Sigma_0=\Sigma_{AR}$)}}} & \multirow{2}[2]{*}{1} & 50    & 0.034 & 0.034 & 0.074 & 0.074 & 0.040 & 0.038 & 0.040 & 0.038 \\
		&       &       & 300   & 0.048 & 0.048 & 0.044 & 0.044 & 0.048 & 0.046 & 0.048 & 0.046 \\
		\cmidrule{3-12}          &       & \multirow{2}[2]{*}{3} & 50    & 0.050 & 0.050 & 0.176 & 0.174 & 0.038 & 0.028 & 0.050 & 0.048 \\
		&       &       & 300   & 0.052 & 0.052 & 0.040 & 0.040 & 0.052 & 0.048 & 0.052 & 0.044 \\
		\cmidrule{2-12}          & \multicolumn{1}{c}{\multirow{4}[4]{*}{\shortstack{$\beta_A\ne 0$-Sparse\\\newline{}($\Sigma_0=\Sigma_{AR}$)}}} & \multirow{2}[2]{*}{1} & 50    & 0.044 & 0.044 & 0.104 & 0.104 & 0.036 & 0.034 & 0.036 & 0.034 \\
		&       &       & 300   & 0.048 & 0.048 & 0.044 & 0.044 & 0.040 & 0.040 & 0.040 & 0.040 \\
		\cmidrule{3-12}          &       & \multirow{2}[2]{*}{3} & 50    & 0.054 & 0.054 & 0.228 & 0.228 & 0.046 & 0.046 & 0.056 & 0.046 \\
		&       &       & 300   & 0.080 & 0.080 & 0.040 & 0.040 & 0.048 & 0.046 & 0.046 & 0.044 \\
		\cmidrule{2-12}          & \multicolumn{1}{c}{\multirow{4}[4]{*}{\shortstack{$\beta_A=0$\\\newline{}($\Sigma_0=\Sigma_{CS}$)}}} & \multirow{2}[2]{*}{1} & 50    & 0.066 & 0.066 & 0.070 & 0.070 & 0.044 & 0.044 & 0.044 & 0.044 \\
		&       &       & 300   & 0.054 & 0.054 & 0.056 & 0.056 & 0.054 & 0.054 & 0.054 & 0.054 \\
		\cmidrule{3-12}          &       & \multirow{2}[2]{*}{3} & 50    & 0.064 & 0.064 & 0.070 & 0.070 & 0.046 & 0.044 & 0.048 & 0.048 \\
		&       &       & 300   & 0.066 & 0.066 & 0.060 & 0.060 & 0.036 & 0.036 & 0.040 & 0.038 \\
		\cmidrule{2-12}          & \multicolumn{1}{c}{\multirow{4}[4]{*}{\shortstack{$\beta_A\ne 0$-Sparse\\\newline{}($\Sigma_0=\Sigma_{CS}$)}}} & \multirow{2}[2]{*}{1} & 50    & 0.066 & 0.066 & 0.070 & 0.070 & 0.052 & 0.052 & 0.052 & 0.052 \\
		&       &       & 300   & 0.056 & 0.056 & 0.056 & 0.056 & 0.050 & 0.050 & 0.050 & 0.050 \\
		\cmidrule{3-12}          &       & \multirow{2}[2]{*}{3} & 50    & 0.064 & 0.064 & 0.068 & 0.068 & 0.060 & 0.060 & 0.054 & 0.054 \\
		&       &       & 300   & 0.074 & 0.074 & 0.062 & 0.062 & 0.048 & 0.048 & 0.056 & 0.056 \\
		\cmidrule{1-12}
		\multicolumn{1}{c}{\multirow{16}[16]{*}{\shortstack{Decaying \\ Coefficients \\ ($\theta_M\ne 0$)}}} & \multicolumn{1}{c}{\multirow{4}[4]{*}{\shortstack{$\beta_A=0$\\\newline{}($\Sigma_0=\Sigma_{AR}$)}}} & \multirow{2}[2]{*}{1} & 50    & 0.034 & 0.034 & 0.076 & 0.076 & 0.042 & 0.038 & 0.042 & 0.038 \\
		&       &       & 300   & 0.050 & 0.050 & 0.046 & 0.046 & 0.050 & 0.048 & 0.050 & 0.048 \\
		\cmidrule{3-12}          &       & \multirow{2}[2]{*}{3} & 50    & 0.052 & 0.052 & 0.156 & 0.156 & 0.036 & 0.032 & 0.052 & 0.048 \\
		&       &       & 300   & 0.048 & 0.048 & 0.046 & 0.046 & 0.052 & 0.048 & 0.050 & 0.048 \\
		\cmidrule{2-12}          & \multicolumn{1}{c}{\multirow{4}[4]{*}{\shortstack{$\beta_A\ne 0$-Sparse\\\newline{}($\Sigma_0=\Sigma_{AR}$)}}} & \multirow{2}[2]{*}{1} & 50    & 0.048 & 0.048 & 0.098 & 0.098 & 0.034 & 0.032 & 0.034 & 0.032 \\
		&       &       & 300   & 0.062 & 0.062 & 0.046 & 0.046 & 0.040 & 0.040 & 0.040 & 0.040 \\
		\cmidrule{3-12}          &       & \multirow{2}[2]{*}{3} & 50    & 0.056 & 0.056 & 0.228 & 0.228 & 0.048 & 0.048 & 0.060 & 0.052 \\
		&       &       & 300   & 0.086 & 0.086 & 0.046 & 0.046 & 0.054 & 0.050 & 0.052 & 0.050 \\
		\cmidrule{2-12}          & \multicolumn{1}{c}{\multirow{4}[4]{*}{\shortstack{$\beta_A=0$\\\newline{}($\Sigma_0=\Sigma_{CS}$)}}} & \multirow{2}[2]{*}{1} & 50    & 0.066 & 0.066 & 0.068 & 0.068 & 0.044 & 0.044 & 0.044 & 0.044 \\
		&       &       & 300   & 0.058 & 0.058 & 0.054 & 0.054 & 0.050 & 0.050 & 0.050 & 0.050 \\
		\cmidrule{3-12}          &       & \multirow{2}[2]{*}{3} & 50    & 0.062 & 0.062 & 0.060 & 0.060 & 0.046 & 0.046 & 0.046 & 0.046 \\
		&       &       & 300   & 0.078 & 0.078 & 0.060 & 0.060 & 0.038 & 0.034 & 0.038 & 0.038 \\
		\cmidrule{2-12}          & \multicolumn{1}{c}{\multirow{4}[4]{*}{\shortstack{$\beta_A\ne 0$-Sparse\\\newline{}($\Sigma_0=\Sigma_{CS}$)}}} & \multirow{2}[2]{*}{1} & 50    & 0.068 & 0.068 & 0.068 & 0.068 & 0.052 & 0.052 & 0.052 & 0.052 \\
		&       &       & 300   & 0.064 & 0.064 & 0.054 & 0.054 & 0.048 & 0.048 & 0.048 & 0.048 \\
		\cmidrule{3-12}          &       & \multirow{2}[2]{*}{3} & 50    & 0.066 & 0.066 & 0.062 & 0.062 & 0.062 & 0.062 & 0.054 & 0.054 \\
		&       &       & 300   & 0.064 & 0.064 & 0.060 & 0.060 & 0.046 & 0.046 & 0.058 & 0.058 \\
		\cmidrule{1-12}
	\end{tabular}%
	\label{tab:size-Gaussian-others}%
\end{table}%

\setlength{\tabcolsep}{10pt} 
\renewcommand{\arraystretch}{0.75} 
\begin{table}[htbp]
	\centering
	\caption{Empirical power behaviors for exposures sampled from $N(0, 0.5^2)$.}
	\begin{tabular}{ccccccc}
		\cmidrule{1-7}
		Sparsity & Case  & $n$   & $c_1$ or $c_2$ & Zhou-$1$ & Guo-$1$ & $\text{Bonf}$-$1$ \\
		\cmidrule{1-7}
		\multirow{20}[8]{*}{Hard Sparsity} & \multirow{10}[4]{*}{\shortstack{Fix $\theta_{M}$, \\ Vary $\beta_{A}$}} & \multirow{5}[2]{*}{50} & 0     & 0.036 & 0.076 & 0.03 \\
		&       &       & 1/8   & 0.128 & 0.134 & 0.054 \\
		&       &       & 1/4   & 0.308 & 0.266 & 0.128 \\
		&       &       & 1/2   & 0.774 & 0.684 & 0.43 \\
		&       &       & 1     & 0.994 & 0.982 & 0.794 \\
		\cmidrule{3-7}          &       & \multirow{5}[2]{*}{300} & 0     & 0.042 & 0.05  & 0.04 \\
		&       &       & 1/8   & 0.446 & 0.394 & 0.188 \\
		&       &       & 1/4   & 0.948 & 0.93  & 0.54 \\
		&       &       & 1/2   & 1     & 1     & 0.952 \\
		&       &       & 1     & 1     & 1     & 0.992 \\
		\cmidrule{2-7}          & \multirow{10}[4]{*}{\shortstack{Fix $\beta_{A}$, \\ Vary $\theta_{M}$}} & \multirow{5}[2]{*}{50} & 0     & 0.134 & 0.204 & 0.034 \\
		&       &       & 1/8   & 0.716 & 0.268 & 0.104 \\
		&       &       & 1/4   & 0.884 & 0.504 & 0.234 \\
		&       &       & 1/2   & 0.936 & 0.796 & 0.502 \\
		&       &       & 1     & 0.946 & 0.938 & 0.754 \\
		\cmidrule{3-7}          &       & \multirow{5}[2]{*}{300} & 0     & 0.138 & 0     & 0.052 \\
		&       &       & 1/8   & 0.986 & 0.544 & 0.178 \\
		&       &       & 1/4   & 1     & 1     & 0.486 \\
		&       &       & 1/2   & 1     & 1     & 0.914 \\
		&       &       & 1     & 1     & 1     & 0.996 \\
		\cmidrule{1-7}
		\multirow{20}[7]{*}{Capped-$\ell_1$ Sparsity} & \multirow{10}[4]{*}{\shortstack{Fix $\theta_{M}$, \\ Vary $\beta_{A}$}} & \multirow{5}[2]{*}{50} & 0     & 0.038 & 0.072 & 0.03 \\
		&       &       & 1/8   & 0.128 & 0.112 & 0.056 \\
		&       &       & 1/4   & 0.31  & 0.26  & 0.136 \\
		&       &       & 1/2   & 0.782 & 0.678 & 0.434 \\
		&       &       & 1     & 0.994 & 0.988 & 0.794 \\
		\cmidrule{3-7}          &       & \multirow{5}[2]{*}{300} & 0     & 0.044 & 0.052 & 0.046 \\
		&       &       & 1/8   & 0.46  & 0.394 & 0.21 \\
		&       &       & 1/4   & 0.954 & 0.93  & 0.596 \\
		&       &       & 1/2   & 1     & 1     & 0.97 \\
		&       &       & 1     & 1     & 1     & 0.998 \\
		\cmidrule{2-7}          & \multirow{10}[3]{*}{\shortstack{Fix $\beta_{A}$, \\ Vary $\theta_{M}$}} & \multirow{5}[2]{*}{50} & 0     & 0.134 & 0.204 & 0.034 \\
		&       &       & 1/8   & 0.726 & 0.286 & 0.106 \\
		&       &       & 1/4   & 0.888 & 0.496 & 0.236 \\
		&       &       & 1/2   & 0.934 & 0.818 & 0.512 \\
		&       &       & 1     & 0.946 & 0.94  & 0.756 \\
		\cmidrule{3-7}          &       & \multirow{5}[1]{*}{300} & 0     & 0.138 & 0     & 0.052 \\
		&       &       & 1/8   & 0.992 & 0.558 & 0.184 \\
		&       &       & 1/4   & 1     & 1     & 0.526 \\
		&       &       & 1/2   & 1     & 1     & 0.932 \\
		&       &       & 1     & 1     & 1     & 1 \\
		\cmidrule{1-7}
		\multirow{20}[7]{*}{Decaying Coefficients} & \multirow{10}[3]{*}{\shortstack{Fix $\theta_{M}$, \\ Vary $\beta_{A}$}} & \multirow{5}[1]{*}{50} & 0     & 0.038 & 0.082 & 0.03 \\
		&       &       & 1/8   & 0.13  & 0.122 & 0.058 \\
		&       &       & 1/4   & 0.314 & 0.274 & 0.138 \\
		&       &       & 1/2   & 0.792 & 0.68  & 0.45 \\
		&       &       & 1     & 0.994 & 0.984 & 0.806 \\
		\cmidrule{3-7}          &       & \multirow{5}[2]{*}{300} & 0     & 0.046 & 0.052 & 0.042 \\
		&       &       & 1/8   & 0.456 & 0.394 & 0.212 \\
		&       &       & 1/4   & 0.95  & 0.93  & 0.598 \\
		&       &       & 1/2   & 1     & 1     & 0.97 \\
		&       &       & 1     & 1     & 1     & 0.998 \\
		\cmidrule{2-7}          & \multirow{10}[4]{*}{\shortstack{Fix $\beta_{A}$, \\ Vary $\theta_{M}$}} & \multirow{5}[2]{*}{50} & 0     & 0.134 & 0.204 & 0.034 \\
		&       &       & 1/8   & 0.746 & 0.288 & 0.114 \\
		&       &       & 1/4   & 0.904 & 0.514 & 0.242 \\
		&       &       & 1/2   & 0.946 & 0.83  & 0.528 \\
		&       &       & 1     & 0.958 & 0.946 & 0.788 \\
		\cmidrule{3-7}          &       & \multirow{5}[2]{*}{300} & 0     & 0.138 & 0     & 0.052 \\
		&       &       & 1/8   & 0.992 & 0.572 & 0.184 \\
		&       &       & 1/4   & 1     & 1     & 0.528 \\
		&       &       & 1/2   & 1     & 1     & 0.932 \\
		&       &       & 1     & 1     & 1     & 1 \\
		\cmidrule{1-7}
	\end{tabular}%
	\label{tab:power-Gaussian}%
\end{table}%

\section{A Modified Test Procedure with Sign Consistency}\label{sec:ex-sign-con}

In certain applications, sparsity of the mediation effect may be a reasonable assumption and variable selection among high-dimensional mediators might be of interest. 
For example, \cite{luo2020high} developed a high-dimensional mediation analysis procedure to select important DNA methylation markers in identifying epigenetic pathways between environmental exposures and health outcomes.
\cite{guo2022statistical} studied how crucial financial metrics, selected from numerous ones, mediate the relationship between company sectors and stock price recovery after COVID-19 pandemic outbreak.
In this section, we present a modification to the proposed test procedure to exploit  sparsity of the mediation effect when it is assumed. Specifically, we assume that the regression coefficient $\theta$ in \eqref{eq:Y-X-equation} is $\ell_0$ sparse, that is,  $|\mathcal{S}|=s$ for some positive integer $s$, where $\mathcal{S}=\supp{\theta}=\{j: \theta_j \ne 0\}$ represents the support set of  $\theta$. 


We take the Lasso estimator with tuning parameter $\lambda_n$ as the initial estimator for $\theta$. To ensure the sign consistency of this initial estimator $\hat{\theta}$ for $\theta$, we assume the uniform signal strength condition \citep{zhang2014confidence} and the mutual incoherence condition \citep{wainwright2009sharp} (or the irrepresentable condition in \cite{zhao2006model}), which are two crucial conditions for $\ell_1$-penalized estimators in the high-dimensional literature \citep[see, e.g., Chapter 4.3.1 of ][]{fan2020statistical}. Previous studies have demonstrated the necessity of these conditions for the sign consistency of $\ell_1$-penalized estimators. For more details, see \cite{zhang2014confidence, wainwright2009sharp, zhao2006model, van2009conditions}. 

\begin{enumerate}
	\item[(C1)] Mutual Incoherence: $\|\Sigma_{\mathcal{S}^c\mathcal{S}} \Sigma_{\mathcal{S}\mathcal{S}}^{-1}\|_{\infty} \le 1-\omega$, for some $\omega\in (0,1]$.
	\item[(C2)] Uniform Signal Strength: $\theta_{\min}=\min_{j\in \mathcal{S}}|\theta_{j}| \gtrsim \lambda_n\sqrt{s}$. 
\end{enumerate}

\begin{proposition}\label{prop:sign-consistency-lasso}
	Suppose that Assumptions (A1), (C1) and (C2) hold, and consider the Lasso estimator $\hat{\theta}$ with a tuning parameter $\lambda_n \asymp \sqrt{\log(q+p) / n}$. Then, for all  sufficiently large $n$ and $s^3\log(q+p-s) \lesssim n$, with probability approaching one as $n\rightarrow\infty$, the Lasso estimator $\hat{\theta}$ is  unique, satisfies $\|\hat{\theta}_{\mathcal{S}} - \theta_{\mathcal{S}}\|_{\infty} \lesssim \lambda_n\sqrt{s}$, and possesses the sign consistency $\hat{\mathcal{S}}=\mathcal{S}$, where $\hat{\mathcal{S}}=\text{supp}(\hat{\theta})$. Moreover, (B1) is satisfied with $\|\hat{\theta}_{\mathcal{S}}-\theta_{\mathcal{S}}\|_1 \lesssim s \sqrt{\log(s)/n}$.
\end{proposition}

Based on Proposition \ref{prop:sign-consistency-lasso}, we can modify the estimation strategy in Section \ref{sec:estimation} when $\ell_0$ sparsity is assumed, as follows. Given the Lasso estimator $\hat{\theta}=(\hat{\theta}_A, \hat{\theta}_M)$ for $\theta$ and the ordinary least-squares estimator $\hat{\beta}_{A}= ((A^{\transpose}A)^{-1}A^{\transpose} M)^{\transpose}$ for $\beta_A$, we still consider the pilot estimator $\tilde{\gamma}=\hat{\beta}_A^{\transpose}\hat{\theta}_M$ for $\gamma$. 
Recall that, for each $j=1,\ldots, q$, $g_j = (\zerovec{q}^{\transpose}, \tilde{g}_j^{\transpose}) \in \real^{q+p}$ where $(\tilde{g}_1, \ldots, \tilde{g}_q)=\hat{\beta}_{A}$.
Let $\hat{\mathcal{S}}$ denote the support of $\hat{\theta}$, and $\hat{s}=|\hat{\mathcal{S}}|$ denote its size. 
Moreover, let $\hat{\theta}_{\hat{\mathcal{S}}}$ and $g_{j, \hat{\mathcal{S}}}$ be respectively the sub-vector of $\hat{\theta}$ and $g_{j}$ with coordinates in $\hat{\mathcal{S}}$, and $X_{\hat{\mathcal{S}}}$ be the sub-matrix of $X$ with columns in $\hat{\mathcal{S}}$. Then  we propose the folowing modified debiased estimator $\hat{\gamma}_j$ for each $j=1,\ldots, q$,
\[
\hat{\gamma}_j = \tilde{\gamma}_j + n^{-1} \hat{u}_{\hat{\mathcal{S}}, j}^{\transpose} X_{\hat{\mathcal{S}}}^{\transpose} (Y-X_{\hat{\mathcal{S}}}\hat{\theta}_{\hat{\mathcal{S}}})
\]
with 
\begin{align*}
	\hat{u}_{\hat{\mathcal{S}},j}=\arg\min_{u\in \real^{\hat s}} u^{\transpose} \hat{\Sigma}_{X_{\hat{\mathcal{S}}}} u \quad \text{subject to } \quad  
	&\| \hat{\Sigma}_{X_{\hat{\mathcal{S}}}} u - g_{j,\hat{\mathcal{S}}} \|_{\infty} \le \|g_{j,\hat{\mathcal{S}}}\|_2 \lambda \\
	&\big|g_{j,\hat{\mathcal{S}}}^{\transpose}\hat{\Sigma}_{X_{\hat{\mathcal{S}}}} u-\|g_{j,\hat{\mathcal{S}}}\|_2^2\big| \le \|g_{j,\hat{\mathcal{S}}}\|_2^2 \lambda, \\
	&\|X_{\hat{\mathcal{S}}} u\|_{\infty} \le \|g_{j,\hat{\mathcal{S}}}\|_2\mu,
\end{align*}
where $\lambda \asymp \sqrt{\log(s) /n}$ , $\mu \asymp \log n$, and $\hat{\Sigma}_{X_{\hat{\mathcal{S}}}}=n^{-1} X_{\hat{\mathcal{S}}}^{\transpose} X_{\hat{\mathcal{S}}}$. 
With $$\hat{U}_{\hat{\mathcal{S}}}=(\hat{u}_{\hat{\mathcal{S}},1}, \ldots, \hat{u}_{\hat{\mathcal{S}},q})\in \real^{\hat s \times q},$$ the modified de-biased estimator is also represented by
\begin{equation*}
	\hat{\gamma} = \tilde{\gamma} + n^{-1} \hat{U}_{\hat{\mathcal{S}}}^{\transpose} X_{\hat{\mathcal{S}}}^{\transpose} (Y-X_{\hat{\mathcal{S}}}\hat{\theta}_{\hat{\mathcal{S}}}).
\end{equation*}
The corresponding asymptotic variance, conditional on $\{X_i, i=1,\ldots, n\}$, is
\begin{equation*}
	V = \frac{\sigma_E^2}{n} \hat{\Sigma}_A^{-1} + \frac{\sigma_Z^2}{n} \hat{U}_{\hat{\mathcal{S}}}^{\transpose} \hat{\Sigma}_{X_{\hat{\mathcal{S}}}} \hat{U}_{\hat{\mathcal{S}}} + \frac{\tau}{n} I_q.
\end{equation*}
The validity and power behavior of the modified test procedure can be re-established in analogy to Theorems \ref{thm:asym-norm-hgamma-decomp} and \ref{thm:test-bonf}, and thus the details are omitted.


\begin{proof}[Proof of Proposition \ref{prop:sign-consistency-lasso}]
	We apply Theorem 1 in \cite{wainwright2009sharp}. To this end, we need to show that, with probability tending to one, there exist constants $c_0^{\prime}, C^{\prime}>0$ and $\omega^{\prime}\in(0,1]$, such that
	\begin{align}
		&\Lambda_{\min}(\hat{\Sigma}_{\mathcal{S}\mathcal{S}}) \ge c_0^{\prime}, 
		\label{eq:smallest-eigenvalue-sample}\\
		&\|\hat{\Sigma}_{\mathcal{S}^c\mathcal{S}}\hat{\Sigma}_{\mathcal{S}\mathcal{S}}^{-1}\|_{\infty} \le 1-\omega^{\prime}, 
		\label{eq:incoherence-sample}\\
		& n^{-1} \max_{j\in \mathcal{S}^c} \|X_{\cdot j}\|_2^2 \le C^{\prime} 
		\label{eq:col-bounded-sample},
	\end{align}
	where $X_{\cdot j}$ is the $j$th column of $X$, $\hat{\Sigma}_{\mathcal{S}\mathcal{S}}=n^{-1} X_{\mathcal{S}}^{\transpose}X_{\mathcal{S}}$ and $\hat{\Sigma}_{\mathcal{S}^c\mathcal{S}}=n^{-1} X_{\mathcal{S}^c}^{\transpose}X_{\mathcal{S}}$. Equation \eqref{eq:incoherence-sample} is handled in Claim \ref{claim:incoherence-pop2sample}, and  \eqref{eq:smallest-eigenvalue-sample} is due to Weyl's inequality and Corollary 5 in \cite{kereta2021estimating}. Specifically, under Assumption (A1),
	\[
	\Lambda_{\min}(\hat{\Sigma}_{\mathcal{S}\mathcal{S}})
	\ge 
	\Lambda_{\min}(\Sigma_{\mathcal{S}\mathcal{S}})
	-
	\|\hat{\Sigma}_{\mathcal{S}\mathcal{S}}-\Sigma_{\mathcal{S}\mathcal{S}}\|_{2} 
	\ge 
	c_0-C_1 \sqrt{\frac{s+\log n}{n}}
	\ge
	c_0^{\prime} > 0
	\]
	with probability at least $1-n^{-1}$, for some constants $c_0^{\prime}, C_1 > 0$ and sufficiently large $n$; here, recall $c_0$ is defined in Assumption (A1).
	
	To establish \eqref{eq:col-bounded-sample}, note that, for each $j\in \mathcal{S}^c$, $\|X_{\cdot j}\|_2$ is the norm of a random vector with independent and subGaussian coordinates. Then by the concentration inequality of norm \citep[Eq (3.1) in][]{vershynin2018high} and the union bound,
	\begin{align*}
		\prob\left(   \max_{j\in \mathcal{S}^c} n^{-1}\|X_{\cdot j}\|_2^2 - \max_{j\in \mathcal{S}^c}\mu_{2j} \ge u \right)
		&\le
		\prob\left( \max_{j\in \mathcal{S}^c} \left| n^{-1} \|X_{\cdot j}\|_2^2 - \mu_{2j} \right| \ge u \right) \\
		&\le 
		(q+p-s)\exp\left(-C_3 n \min(u^2, u)\right) \\
		&=
		\exp\left(-C_3 n \min(u^2, u) + \log(q+p-s)\right),
	\end{align*}
	where $\mu_{2j} = \expect[X_{1j}^2]>0$ and $C_3>0$ is a constant. Since $\Lambda_{\max}(\Sigma_X)\le C_0$ by Assumption (A1), $\mu_{2j}\le C_0$ for each $j=1,\ldots,q+p$. Taking $u$ to be a suitable constant in the above inequality leads to \eqref{eq:col-bounded-sample}. 
	
	With \eqref{eq:smallest-eigenvalue-sample}--\eqref{eq:col-bounded-sample} and Assumption (C2), taking $\lambda_n \asymp \sqrt{\frac{\log(q+p)}{n}}$, by Theorem 1 in \cite{wainwright2009sharp}, we conclude that, with probability tending to one, the Lasso estimator is unique, possesses sign consistency $\hat{\mathcal{S}}=\mathcal{S}$, and satisfies $\|\hat{\theta}_{\mathcal{S}} - \theta_{\mathcal{S}}\|_{\infty} \lesssim \lambda_n\sqrt{s}$. 
	Moreover, it is clear that the (sample version) mutual incoherence condition \eqref{eq:incoherence-sample} implies the (sample version) uniform irrepresentable condition \citep{van2009conditions}, which further implies the (sample version) compatibility condition \citep[Theorem 9.1 in][]{van2009conditions}. Hence, following the arguments in the proof for Proposition \ref{prop:CC-general-require}, one can show that (B1) is satisfied.
\end{proof}

\begin{claim}
	\label{claim:incoherence-pop2sample}
	Under Assumptions (A1) and (C1), if  $s^3\log(q+p-s) \ll n$ 
	then the sample covariance matrix $\hat{\Sigma}_{X}$  satisfies the sample version of the mutual incoherence condition, i.e.,
	\[\prob\left(\|\hat{\Sigma}_{\mathcal{S}^c\mathcal{S}}\hat{\Sigma}_{\mathcal{S}\mathcal{S}}^{-1}\|_{\infty} \le 1-\omega/2\right) \to 1,
	\]
	as $n\to \infty$.
\end{claim}

\begin{proof}[Proof of Claim \ref{claim:incoherence-pop2sample}]
	We decompose the sample covariance matrix in the mutual incoherence condition as  $\hat{\Sigma}_{\mathcal{S}^c\mathcal{S}}\hat{\Sigma}_{\mathcal{S}\mathcal{S}}^{-1} = T_1 + T_2 + T_3 + T_4$, where
	\begin{align*}
		T_1 &= \Sigma_{\mathcal{S}^c\mathcal{S}}(\hat{\Sigma}_{\mathcal{S}\mathcal{S}}^{-1}-\Sigma_{\mathcal{S}\mathcal{S}}^{-1})\\
		T_2 &= (\hat{\Sigma}_{\mathcal{S}^c\mathcal{S}}-\Sigma_{\mathcal{S}^c\mathcal{S}}) \Sigma_{\mathcal{S}\mathcal{S}}^{-1}\\
		T_3 &= (\hat{\Sigma}_{\mathcal{S}^c\mathcal{S}}-\Sigma_{\mathcal{S}^c\mathcal{S}})(\hat{\Sigma}_{\mathcal{S}\mathcal{S}}^{-1}-\Sigma_{\mathcal{S}\mathcal{S}}^{-1})\\
		T_4 &= \Sigma_{\mathcal{S}^c\mathcal{S}} \Sigma_{\mathcal{S}\mathcal{S}}^{-1}
	\end{align*}
	as in the proof of Lemma 6 in \cite{ravikumar2010high}. The fourth term can be bounded by the population mutual incoherence condition, specifically,
	\begin{equation}
		\|T_4\|_{\infty} = \|\Sigma_{\mathcal{S}^c\mathcal{S}} \Sigma_{\mathcal{S}\mathcal{S}}^{-1}\|_{\infty} \le 1-\omega.
		\label{eq:T4-bound}
	\end{equation}
	
	For the first term $T_1$, we re-factorize it as	
	\[
	T_1 = \Sigma_{\mathcal{S}^c\mathcal{S}}\Sigma_{\mathcal{S}\mathcal{S}}^{-1}(\Sigma_{\mathcal{S}\mathcal{S}}-\hat{\Sigma}_{\mathcal{S}\mathcal{S}}) \hat{\Sigma}_{\mathcal{S}\mathcal{S}}^{-1},
	\]
	and then bound it by applying the sub-multiplicative property of the matrix $\ell_{\infty}$ norm:
	\begin{align}
		\notag
		\|T_1\|_{\infty} &\le \|\Sigma_{\mathcal{S}^c\mathcal{S}}\Sigma_{\mathcal{S}\mathcal{S}}^{-1}\|_{\infty} \|\Sigma_{\mathcal{S}\mathcal{S}}-\hat{\Sigma}_{\mathcal{S}\mathcal{S}}\|_{\infty}\|\hat{\Sigma}_{\mathcal{S}\mathcal{S}}^{-1}\|_{\infty} \\
		\notag
		&\le 
		\|\Sigma_{\mathcal{S}^c\mathcal{S}}\Sigma_{\mathcal{S}\mathcal{S}}^{-1}\|_{\infty} \sqrt{s}\|\Sigma_{\mathcal{S}\mathcal{S}}-\hat{\Sigma}_{\mathcal{S}\mathcal{S}}\|_{2} \sqrt{s}\left(\|\Sigma_{\mathcal{S}\mathcal{S}}^{-1}\|_{2}+\|\hat{\Sigma}_{\mathcal{S}\mathcal{S}}^{-1}-\Sigma_{\mathcal{S}\mathcal{S}}^{-1}\|_{2}\right) \\
		&\le 
		(1-\omega) s \|\Sigma_{\mathcal{S}\mathcal{S}}-\hat{\Sigma}_{\mathcal{S}\mathcal{S}}\|_{2} \left(c_0^{-1}+\|\hat{\Sigma}_{\mathcal{S}\mathcal{S}}^{-1}-\Sigma_{\mathcal{S}\mathcal{S}}^{-1}\|_{2}\right),
		\label{eq:T1-init-bound}
	\end{align}
	where the last inequality is due to Assumpsion (A1). With Assumpsion (A1), by Corollary 5 and Equation (1.6) in \cite{kereta2021estimating}, one has
	\[
	\prob\left( \|\Sigma_{\mathcal{S}\mathcal{S}}-\hat{\Sigma}_{\mathcal{S}\mathcal{S}}\|_{2} \gtrsim \sqrt{\frac{s+\log n}{n}} \right) \le n^{-1}
	\]
	and \[
	\prob\left( \|\hat{\Sigma}_{\mathcal{S}\mathcal{S}}^{-1}-\Sigma_{\mathcal{S}\mathcal{S}}^{-1}\|_{2} \gtrsim \sqrt{\frac{s+\log n}{n}} \right) \le n^{-1}.
	\]
	Combining the above two inequalities with \eqref{eq:T1-init-bound} yields
	\begin{equation}
		\prob\left( \|T_1\|_{\infty} \ge \omega/6 \right) \le 2n^{-1},
		\label{eq:T1-bound}
	\end{equation}
	for $s^3 \ll n$ and sufficiently large $n$.
	
	For the second term $T_2$, we first derive a bound for $\|\hat{\Sigma}_{\mathcal{S}^c\mathcal{S}}-\Sigma_{\mathcal{S}^c\mathcal{S}}\|_{\infty}$. By definition, the $(j,k)$th element of the difference matrix $\hat{\Sigma}_X-\Sigma_X$ can be represented as an average of i.i.d. random variables, $R_{jk}=n^{-1}\sum_{i=1}^{n} R_{jk}^{(i)}$ where $R_{jk}^{(i)}=X_{ij}X_{ik}-\expect[X_{ij}X_{ik}]$ is a centered sub-exponential random variable. Hence,
	\begin{align}
		&\notag
		\prob\left(\|\hat{\Sigma}_{\mathcal{S}^c\mathcal{S}}-\Sigma_{\mathcal{S}^c\mathcal{S}}\|_{\infty} \ge t \right)\\
		&=
		\prob\left(\max_{j\in \mathcal{S}^c} \sum_{k\in \mathcal{S}} |R_{jk}| \ge t \right) 
		\le
		(q+p-s) \prob\left(\sum_{k\in \mathcal{S}} |R_{jk}| \ge t \right) \notag \\
		\notag
		&\le (q+p-s) s \prob\left( |R_{jk}| \ge \frac{t}{s} \right) \\
		\notag
		&\le 
		\exp\left( -C_1 n \min\left\{\frac{t^2}{s^2}, \frac{t}{s}\right\} + \log(q+p-s) + \log(s) \right) \\
		&\le 
		\exp\left( -C_1 n \min\left\{\frac{t^2}{s^2}, \frac{t}{s}\right\} + 2\log(q+p-s) \right),
		\label{eq:Sig-ScS-linf-concen}
	\end{align}
	for some constant $C_1>0$.
	Taking $t\asymp \frac{1}{\sqrt{s}}\frac{\omega}{6}$ in the above inequality, we have
	\begin{align}
		\notag
		\prob\left(\|T_2\|_{\infty} \ge \omega/6\right) 
		&\le
		\prob\left( \|\hat{\Sigma}_{\mathcal{S}^c\mathcal{S}}-\Sigma_{\mathcal{S}^c\mathcal{S}}\|_{\infty} \|\Sigma_{\mathcal{S}\mathcal{S}}^{-1}\|_{\infty} \ge \omega/6 \right) \\
		\notag
		&\le 
		\prob\left( \|\hat{\Sigma}_{\mathcal{S}^c\mathcal{S}}-\Sigma_{\mathcal{S}^c\mathcal{S}}\|_{\infty} \sqrt{s} \|\Sigma_{\mathcal{S}\mathcal{S}}^{-1}\|_{2} \ge \omega/6 \right) \\
		&\le
		\exp\left( -C_2 n/s^3 + \log(q+p-s) \right)
		\label{eq:T2-bound}
	\end{align}
	for come constant $C_2>0$. 
	The above probability on the right-hand side would tend to zero as $n$ goes to infinity, once $s^3\log(q+p-s) \ll n$ and $n$ is sufficiently large.
	
	For the third term $T_3$, taking  $t\asymp \omega/6$ in \eqref{eq:Sig-ScS-linf-concen}, with Equation (1.6) in \cite{kereta2021estimating}, we deduce that
	\begin{small}
		\begin{align}
			& \notag
			\prob\left(\|T_3\|_{\infty} \ge \omega/6\right) \\
			&\notag\le
			\prob\left( \|\hat{\Sigma}_{\mathcal{S}^c\mathcal{S}}-\Sigma_{\mathcal{S}^c\mathcal{S}}\|_{\infty} \|\hat{\Sigma}_{\mathcal{S}\mathcal{S}}^{-1}-\Sigma_{\mathcal{S}\mathcal{S}}^{-1}\|_{\infty} \ge \omega/6 \right) \\
			\notag
			&\le
			\prob\left( \|\hat{\Sigma}_{\mathcal{S}^c\mathcal{S}}-\Sigma_{\mathcal{S}^c\mathcal{S}}\|_{\infty} \sqrt{s}\|\hat{\Sigma}_{\mathcal{S}\mathcal{S}}^{-1}-\Sigma_{\mathcal{S}\mathcal{S}}^{-1}\|_{2} \ge \omega/6,~ \|\hat{\Sigma}_{\mathcal{S}\mathcal{S}}^{-1}-\Sigma_{\mathcal{S}\mathcal{S}}^{-1}\|_{2} \lesssim \sqrt{\frac{s+\log n}{n}} \right) \\
			\notag
			&\quad 
			+ n^{-1} \\
			\notag
			&\le
			\prob\left( \|\hat{\Sigma}_{\mathcal{S}^c\mathcal{S}}-\Sigma_{\mathcal{S}^c\mathcal{S}}\|_{\infty} \ge C_3 \omega/6 \right) + n^{-1} \\
			&\le
			\exp\left( -C_4 n/s^2 + \log(q+p-s) \right) +  n^{-1},
			\label{eq:T3-bound}
		\end{align}
	\end{small}
	for some constants $C_3, C_4>0$.
	
	Combining \eqref{eq:T4-bound}, \eqref{eq:T1-bound}, \eqref{eq:T2-bound} and \eqref{eq:T3-bound} together, we conclude that
	\[
	\prob\left(\|\hat{\Sigma}_{\mathcal{S}^c\mathcal{S}}\hat{\Sigma}_{\mathcal{S}\mathcal{S}}^{-1}\|_{\infty} \le 1-\omega/2\right) \to 1,
	\]
	as $n\to \infty$. 
\end{proof}




\end{appendices}


\bibliography{ref}

\end{document}